\documentclass[final]{siamart}

\usepackage{mathrsfs}
\usepackage{amssymb}
\usepackage{mathdots}

\newcommand{\gathen}[1]{#1}

\usepackage{enumerate}

\newcommand{\M}{\ensuremath{\mathsf{M}}}

\def\lg {\ensuremath{\mathrm{lg}}}

\def\pol {\ensuremath{\mathsf{pol}}}
\def\rev {\ensuremath{\mathsf{rev}}}
\def\red {\ensuremath{\mathsf{red}}}
\def\crt {\ensuremath{\mathsf{crt}}}
\def\comb {\ensuremath{\mathsf{comb}}}

\def\NN {\ensuremath{\mathbb{N}}}
\def\BB {\ensuremath{\mathbb{B}}}
\def\ZZ {\ensuremath{\mathbb{Z}}}
\def\RR {\ensuremath{\mathbb{R}}}
\def\A {\ensuremath{\mathbb{A}}}
\def\B {\ensuremath{\mathbb{B}}}

\def\J {\ensuremath{\mathbb{J}}}

\def\CC {\ensuremath{\mathbb{M}}}

\def\V {\ensuremath{\mathbb{V}}}
\def\W {\ensuremath{\mathbb{W}}}
\def\D {\ensuremath{\mathbb{D}}}

\def\U {\ensuremath{\mathbb{U}}}
\def\Y {\ensuremath{\mathbb{Y}}}
\def\I {\ensuremath{\mathbb{I}}}

\def\sMM {\ensuremath{\mathscr{C}_{\rm mul}}}

\def\sMS {\ensuremath{\mathscr{C}_{\rm solve}}}

\def\sMI {\ensuremath{\mathscr{C}_{\rm inv}}}

\def\sMMa {\ensuremath{\mathscr{M}_{\rm mat}}}

\def\mB {\ensuremath{\mathsf{B}}}
\def\mA {\ensuremath{\mathsf{A}}}
\def\mE {\ensuremath{\mathsf{E}}}

\def\mL {\ensuremath{\mathsf{L}}}
\def\mR {\ensuremath{\mathsf{R}}}
\def\mP {\ensuremath{\mathsf{P}}}
\def\mQ {\ensuremath{\mathsf{Q}}}
\def\mM {\ensuremath{\mathsf{M}}}
\def\mN {\ensuremath{\mathsf{N}}}
\def\mG {\ensuremath{\mathsf{G}}}
\def\mH {\ensuremath{\mathsf{H}}}

\def\mC {\ensuremath{\mathsf{C}}}
\def\mD {\ensuremath{\mathsf{D}}}
\def\sC {\ensuremath{\mathscr{C}}}
\def\sD {\ensuremath{\mathscr{D}}}
\def\sM {\ensuremath{\mathscr{M}}}
\def\mY {\ensuremath{\mathsf{Y}}}

\def\mZ {\ensuremath{\mathsf{Z}}}
\def\mS {\ensuremath{\mathsf{S}}}

\def\mU {\ensuremath{\mathsf{U}}}
\def\mV {\ensuremath{\mathsf{V}}}
\def\mW {\ensuremath{\mathsf{W}}}

\def\x {\ensuremath{\mathsf{x}}}
\def\y {\ensuremath{\mathsf{y}}}
\def\z {\ensuremath{\mathsf{z}}}
\def\t {\ensuremath{\mathsf{t}}}

\def\e {\ensuremath{\mathsf{e}}}
\def\g {\ensuremath{\mathsf{g}}}
\def\h {\ensuremath{\mathsf{h}}}
\def\a {\ensuremath{\mathsf{a}}}
\def\m {\ensuremath{\mathsf{m}}}
\def\n {\ensuremath{\mathsf{n}}}
\def\b {\ensuremath{\mathsf{b}}}
\def\r {\ensuremath{\mathsf{r}}}
\def\s {\ensuremath{\mathsf{s}}}
\def\u {\ensuremath{\mathsf{u}}}
\def\v {\ensuremath{\mathsf{v}}}
\def\w {\ensuremath{\mathsf{w}}}
\def\c {\ensuremath{\mathsf{c}}}

\def\F {\ensuremath{\mathbb{F}}}

\def\XX {\ensuremath{\mathbb{X}}}
\def\K {\ensuremath{\mathbb{K}}}
\def\vP {\ensuremath{\mathbf{P}}}
\def\vQ {\ensuremath{\mathbf{Q}}}

\def\mcL {\ensuremath{\mathcal{L}}}

\usepackage{calc}
\def\bdiv {\ensuremath{{\rm~div~}}}

\def\largestrec{{\sf largest\_rec}}
\def\largest{{\sf largest}}
\def\lpinv{{\sf lp\_inv}}

\usepackage{color}
\usepackage{soul}

\renewcommand{\le}{\leqslant}
\renewcommand{\ge}{\geqslant}

\newcommand{\TheTitle}{On matrices with displacement structure:\\generalized operators and faster algorithms} 
\newcommand{\TheAuthors}{A. Bostan, C.-P. Jeannerod, C. Mouilleron, and \'E. Schost}

\headers{On matrices with displacement structure}{\TheAuthors}

\title{{\TheTitle}} 

\author{
  A.~Bostan\thanks{Inria, France
  (\email{alin.bostan@inria.fr}).}
  \and
  C.-P.~Jeannerod\thanks{Inria, Universit\'e de Lyon,
  laboratoire LIP (CNRS, ENSL, Inria, UCBL),  France
  (\email{claude-pierre.jeannerod@inria.fr}).}
  \and
  C.~Mouilleron\thanks{ENSIIE, \'Evry, France
  (\email{christophe.mouilleron@ens-lyon.org}).}
  Part of this work was done when C.~Mouilleron was a member of laboratoire LIP (CNRS, ENSL, Inria, UCBL).
  \and
  \'E.~Schost\thanks{David R. Cheriton School of Computer Science, 
  University of Waterloo, ON, Canada
  (\email{eschost@uwaterloo.ca}).}
}

\usepackage{amsopn}
\DeclareMathOperator{\rank}{rank}
\DeclareMathOperator{\diag}{diag}

\ifpdf
\hypersetup{
  pdftitle={\TheTitle},
  pdfauthor={\TheAuthors}
}
\fi

\begin{document}

\maketitle

\begin{abstract}
For matrices with displacement structure,
basic operations like multiplication,
inversion, and linear system solving can all be expressed in terms of the following task: 
evaluate the product $\mA\mB$, where $\mA$ is a
structured $n \times n$ matrix of displacement rank $\alpha$, and
$\mB$ is an arbitrary $n\times\alpha$ matrix. Given~$\mB$ and a
so-called {\em generator} of $\mA$, this product is classically
computed with a cost ranging from $O(\alpha^2 \sM(n))$ to $O(\alpha^2
\sM(n)\log(n))$ arithmetic operations, depending on the type of structure
of $\mA$; here, $\sM$ is a cost function for polynomial multiplication.
In this paper, we first generalize
classical displacement operators, based on block diagonal matrices with
companion diagonal blocks, and then design fast algorithms to perform
the task above for this extended class of structured matrices.  The
cost of these algorithms ranges from $O(\alpha^{\omega-1} \sM(n))$ to
$O(\alpha^{\omega-1} \sM(n)\log(n))$, 
with $\omega$ such that two $n \times n$ matrices over a field can be multiplied
using $O(n^\omega)$ field operations.
By combining this result with classical randomized regularization techniques,
we obtain faster Las Vegas algorithms for 
structured inversion and linear system solving.
\end{abstract}



\section{Introduction}

Exploiting the structure of data is key to develop fast
algorithms and, in the context of linear algebra, this principle is at the
heart of algorithms for {\em displacement structured matrices}. 
These algorithms can speed up for instance the inversion of a given matrix whenever
this matrix has a structure close to that of a Toeplitz, Hankel,
Vandermonde, or Cauchy matrix. 
The idea is to represent structured matrices succinctly by means of their
{\em generators} with respect to suitable {\em displacement
  operators}, and to operate on this succinct data structure.

\smallskip\noindent{\bf Displacement operators.}
Let $\F$ be a field.  To measure the extent to which a matrix
$\mA \in \F^{m\times n}$ possesses some structure, it is customary to
use its {\em displacement rank}, that is, the rank of its image
through a {\em displacement operator}~\cite{KaKuMo79}.
There exist two broad classes of displacement operators: 
{\em Sylvester operators}, of the form
$$\nabla_{\mM,\mN}: \mA \in\F^{m\times n} \mapsto \mM\mA-\mA\mN
\in\F^{m\times n},$$ and {\em Stein operators}, of the
form $$\Delta_{\mM,\mN}: \mA \in\F^{m\times n} \mapsto \mA-\mM\mA\mN
\in\F^{m\times n};$$ in both cases, $\mM$ and $\mN$ are fixed
matrices in $\F^{m\times m}$ and $\F^{n\times n}$, respectively.  For
any such operator, say $\mcL$, the rank of $\mcL(\mA)$ is called
the $\mcL$-{\em displacement rank} of $\mA$, or simply its displacement
rank if $\mcL$ is clear from the context.  Loosely speaking, the matrix 
$\mA$ is called {\em structured} (with respect to the operator $\mcL$) if its
$\mcL$-displacement rank is small compared to its sizes $m$ and $n$.

We say that a matrix pair $(\mG,\mH)$ in $\F^{m\times \alpha}\times
\F^{n \times \alpha}$ is an $\mcL$-{\em generator of length} $\alpha$
of $\mA\in\F^{m\times n}$ if it satisfies $\mcL(A) = \mG\mH^t$, with
$\mH^t$ the transpose of $\mH$.  Again, if $\mcL$ is clear from the context,
we shall simply say {\em generator of length} $\alpha$.  Provided
$\mcL$ is invertible and when $\alpha$ is small, such a generator can
play the role of a succinct data structure to represent the matrix
$\mA$.  The smallest possible value for $\alpha$ is the rank of
$\mcL(\mA)$.

If $\mA$ is invertible
and structured with respect to $\mcL$, then its inverse $\mA^{-1}$ is
structured with respect to the operator~$\mcL'$ 
obtained by swapping $\mM$ and $\mN$:
\begin{equation}\label{eq:mcLP}
\mcL=\nabla_{\mM,\mN} \,\,\implies\,\, \mcL'=\nabla_{\mN,\mM}, \qquad 
  \mcL=\Delta_{\mM,\mN} \,\,\implies\,\, \mcL'=\Delta_{\mN,\mM}.
\end{equation}
More precisely, the above claim says that the ranks of $\mcL(\mA)$ and
$\mcL'(\mA^{-1})$ are the same.
(See~\cite[Theorem~1.5.3]{Pan01} and,
for the case not handled there, see~\cite[p.~771]{KaKuMo79};
for completeness, we give a proof in Appendix~\ref{app:proof-inverse-is-structured}.)

\smallskip\noindent{\bf Some classical operators.}
Consider first Toeplitz-like matrices.  For
$\varphi$ in $\F$, it is customary to define the $m\times m$ cyclic down-shift matrix
\begin{equation} \label{eq:Z_m_phi}
\ZZ_{m,\varphi}=
\begin{bmatrix}
 & & & \varphi\\
1 &  \\
&\ddots & \\
& & 1 & 
\end{bmatrix}
\in\F^{m\times m}.
\end{equation}
Then, using Sylvester operators, a matrix $\mA \in \F^{m \times n}$
will be called Toeplitz-like if $\ZZ_{m,\varphi} \, \mA - \mA\,
  \ZZ_{n,\psi}$ has a low rank compared to $m$ and $n$. (This rank is
  independent of the choice of $\varphi$ and $\psi$, 
up to an additive constant of absolute value at most two.)
Using Stein operators, one would consider instead
the rank of $ \mA -\ZZ_{m,\varphi} \,\mA\, \ZZ_{n,\psi}^t$.

Besides $\ZZ_{m,\varphi}$ and its transpose it is also useful to consider 
the diagonal matrix~$\D(\x)$, whose diagonal is given by $\x \in \F^m$.
Specifically, popular choices take 
\begin{equation} \label{eq:classical-disp-operators}
(\mM,\mN) \in \left\{\ZZ_{m,\varphi},\ZZ_{m,\varphi}^t,\D(\x)\right\} 
\times \left\{\ZZ_{n,\psi},\ZZ_{n,\psi}^t,\D(\y)\right\},
\end{equation}
with $\varphi,\psi\in\F$, $\x\in\F^m$ and $\y\in\F^n$,
for either Sylvester or Stein operators. This family covers in
particular 
Toeplitz and Hankel-like structures, where both $\mM$ and $\mN$
are (transposed) cyclic down-shift matrices; 
Vandermonde-like structures, where one of the matrices $\mM$ and $\mN$ is a
(transposed) cyclic down-shift matrix and the other one is diagonal;
and
Cauchy-like structures, where both $\mM$ and $\mN$ are diagonal.
We say that a displacement operator 
is of {\it Toeplitz \!/\! Hankel type} if it falls into the first category,
and of {\it Vandermonde \!/\! Cauchy type} if it falls into the second
or third one.

\smallskip\noindent{\bf Block-diagonal companion matrices.}
The first contribution of this paper is to generalize 
the classical displacement operators seen before. 
Let $\F[x]$ be the univariate polynomial ring over $\F$, and for
$\delta \ge 0$, let $\F[x]_\delta$ be the $\F$-vector space of polynomials of degree
less than $\delta$.
Furthermore, for
$F = \sum_i f_i x^i \in \F[x]$ monic of degree $\delta>0$, 
denote by $\CC_F$ the
$\delta \times \delta$ companion matrix associated with $F$:
\[
\CC_F = \begin{bmatrix}
   & & & -f_0 \\
1 & & & -f_1 \\       
& \ddots & & \vdots \\
& & 1 & -f_{\delta-1}
\end{bmatrix} \in \F^{\delta \times \delta}.
\]
Two special cases are the $m\times m$ matrix $\ZZ_{m,\varphi}$ 
defined in~(\ref{eq:Z_m_phi})
and the $1 \times 1$ matrix given by $x_0 \in \F$, 
which are associated with the polynomials $x^m-\varphi$ and $x-x_0$, respectively.

Given $d>0$,
let now $\vP=P_1,\dots,P_d$ denote a family of monic nonconstant polynomials in $\F[x]$,
and let $m=m_1+\cdots+m_d$ with $m_i=\deg(P_i) \ge 1$ for all $i$. 
We will call {\em block-diagonal companion matrix} associated with $\vP$ the $m
\times m$ block-diagonal matrix $\CC_\vP$ whose $i$th diagonal block is the $m_i \times m_i$
companion matrix $\CC_{P_i}$:
\[
\CC_\vP = \begin{bmatrix}
\CC_{P_1} & & \\ & \ddots & \\ & & \CC_{P_d}
\end{bmatrix} \in \F^{m \times m}.
\]
We write $P$ to denote the product $P_1 \cdots P_d$,
which is the characteristic polynomial of~$\CC_\vP$.
Finally, we associate with $\vP$ the
following assumption on its elements:
\[
\textnormal{$\mathscr{H}_\vP$: \quad $P_1,\dots,P_d$ are pairwise coprime.}
\]

\smallskip\noindent{\bf Associated displacement operators.}
With $\vP$ as above, we now consider another family of monic nonconstant
polynomials, namely $\vQ=Q_1,\dots,Q_e$, with respective degrees $n_1,\dots,n_e$; 
we let $Q= Q_1\cdots Q_e$ and $n=n_1+\cdots+n_e$. 
In what follows, we assume that both assumptions 
$\mathscr{H}_\vP$ and $\mathscr{H}_\vQ$ hold.

Let then $\CC_\vP \in \F^{m\times m}$ and $\CC_\vQ \in \F^{n\times n}$
be the block-diagonal companion matrices associated with $\vP$ and
$\vQ$. 
In this paper, we  consider the following eight operators:
$$ 
\nabla_{\CC_\vP,\CC_\vQ^t},\quad\nabla_{\CC_\vP^t,\CC_\vQ},\quad\nabla_{\CC_\vP,\CC_\vQ},\quad\nabla_{\CC_\vP^t,\CC_\vQ^t},
$$
$$
\Delta_{\CC_\vP,\CC_\vQ^t},\quad\Delta_{\CC_\vP^t,\CC_\vQ},\quad\Delta_{\CC_\vP,\CC_\vQ},\quad\Delta_{\CC_\vP^t,\CC_\vQ^t}.
$$
We shall call them the {\em operators associated with
$(\vP,\vQ)$ of Sylvester or Stein type}, respectively,
and say that they have {\em format} $(m,n)$.
Two of these operators will be highlighted,
$\nabla_{\CC_\vP,\CC_\vQ^t}$ and $\Delta_{\CC_\vP,\CC_\vQ^t}$, and we will
call them the {\em basic} operators associated with $(\vP,\vQ)$. Indeed,
we will be able to derive convenient formulas to invert them; this
will in turn allow us to deal with the other six operators by suitable
reductions.

Tables~\ref{tab3} and~\ref{tab3b} show that
the classical structures defined by~(\ref{eq:classical-disp-operators})
follow as special cases of these operators, for suitable choices of $(\vP,\vQ)$
making $\CC_\vP$ or $\CC_\vQ$ a diagonal or a (transposed) cyclic down-shift matrix.

\renewcommand{\arraystretch}{1.1}
\begin{table}[h]
\caption{Some particular cases for the Sylvester operator $\nabla_{\CC_\vP,\CC_\vQ^t}$.}\label{tab3}
\begin{center}
\begin{tabular}{c||c|c}
& $e=1$ and $Q_1=x^n - \psi$ & $e=n$ and $Q_j = x-y_j$ 
\\\hline\hline
$d=1$ and $P_1 = x^m-\varphi$ & Hankel-like & 
\begin{tabular}{c} 
Vandermonde${}^t$-like \vspace*{-0.1cm} \\  (for the points $1/y_j$) 
\end{tabular}
\\\hline
$d=m$ and $P_i = x-x_i$
& 
\begin{tabular}{c} 
Vandermonde-like \vspace*{-0.1cm} \\ 
(for the points $1/x_i$)
\end{tabular}
& 
\begin{tabular}{c} 
Cauchy-like \vspace*{-0.15cm} \\ (for the points $x_i$ and $y_j$)
\end{tabular}
\end{tabular}
\end{center}
\end{table}
\vspace*{-0.25cm}
\begin{table}[h]
\caption{Some particular cases for the Stein operator $\Delta_{\CC_\vP,\CC_\vQ^t}$.}\label{tab3b}
\begin{center}
\begin{tabular}{c||c|c}
& $e=1$ and $Q_1=x^n - \psi$ & $e=n$ and $Q_j = x-y_j$ 
\\\hline\hline
$d=1$ and $P_1 = x^m-\varphi$ & Toeplitz-like & 
\begin{tabular}{c} 
Vandermonde${}^t$-like \vspace*{-0.1cm} \\  (for the points $y_j$) 
\end{tabular}
\\\hline
$d=m$ and $P_i = x-x_i$ 
& 
\begin{tabular}{c} 
Vandermonde-like \vspace*{-0.1cm} \\ 
(for the points $x_i$)
\end{tabular}
& 
\begin{tabular}{c} 
Cauchy-like \vspace*{-0.15cm} \\ \hspace*{-0.15cm} (for the points $1/x_i$ and $y_j$)\hspace*{-0.15cm}
\end{tabular}
\end{tabular}
\end{center}
\end{table}
\renewcommand{\arraystretch}{1.0}
\bigskip
In other words, 
our family of operators covers all classical operators of
Toeplitz, Hankel, Vandermonde, and Cauchy type in a unified manner.
However, this formalism also includes many new cases as
``intermediate'' situations. For example, taking $e=1$ and $Q_1=x^n$,
we see that the matrix of the {\em multiple reduction map} modulo $P_1,\dots,P_d$
has displacement rank $1$ for the Stein operator
$\Delta_{\CC_\vP,\CC_\vQ^t}$, so our operators will allow us to address 
problems related to Chinese remaindering.

\smallskip\noindent{\bf Problem statement.} 
Our goal is to study the complexity of basic computations with
matrices that are structured for the operators seen above: multiply a
matrix $\mA\in \F^{m\times n}$ by a matrix $\mB$, invert $\mA$, and
solve the linear system $\mA\x=\b$.
Formally, let $\mcL:\F^{m\times n} \to \F^{m\times n}$ be a
displacement operator of either Sylvester or Stein type,
and assume that $\mcL$ is invertible.
Then, the problems we address read as follows.

\smallskip
\begin{description}
\item[${\sf mul}(\mcL,\alpha,\beta)\!:$] given an $\mcL$-generator
  $(\mG,\mH)\in \F^{m\times \alpha}\times \F^{n\times \alpha}$ of a
  matrix $\mA \in \F^{m\times n}$ with $\alpha \le \min(m,n)$ and
  given a matrix $\mB \in \F^{n\times \beta}$, compute the product
  $\mA \mB \in \F^{m\times \beta}$.
  
\smallskip
  
\item[${\sf inv}(\mcL,\alpha)\!:$] 
  given an $\mcL$-generator $(\mG,\mH)\in \F^{m\times \alpha} \times
  \F^{m\times \alpha}$ of a matrix $\mA \in \F^{m\times m}$ with
  $\alpha \le m$, return an $\mcL'$-generator $(\mG',\mH')$ of length
  at most $\alpha$ for $\mA^{-1}$, or assert that $\mA$ is not
  invertible. Here, $\mcL'$ is the operator in~\eqref{eq:mcLP}.
  
\smallskip
  
\item[${\sf solve}(\mcL,\alpha)\!:$] given an $\mcL$-generator
  $(\mG,\mH)\in \F^{m\times \alpha}\times \F^{n\times \alpha}$ of a
  matrix $\mA \in \F^{m\times n}$ with $\alpha \le \min(m,n)$ and
  given a vector $\b \in \F^m$, find a solution $\x\in\F^n$ to
  $\mA\,\x = \b$, or assert that no such solution exists.  (When
  $\b=0$ and $\mA$ has not full column rank, then a nonzero
  vector $\x$ should be returned; we call this a {\it nontrivial solution}.) 
\end{description}

\smallskip
\noindent
For the problems above to make sense, we need the displacement operator
$\mcL$ to be invertible. For the Sylvester operators associated with
$(\vP,\vQ)$, this occurs if and only if $\gcd(P,Q)=1$; 
for Stein operators, the condition becomes $\gcd(P,\widetilde Q)=1$
with $\widetilde Q=x^n Q(1/x)$ or, equivalently,
$\gcd(\widetilde P, Q)=1$ with $\widetilde P = x^m P(1/x)$~\cite[Theorem~4.3.2]{Pan01}.

We allow probabilistic algorithms. For instance, for inversion in size
$m$, we will see algorithms that use $r=O(m)$ random elements in $\F$,
for which success is conditional to avoiding a hypersurface in
$\overline \F^{r}$ of degree $m^{O(1)}$.

To analyze all upcoming algorithms, we count all arithmetic operations
in $\F$ at unit cost, so the {\em time} spent by an algorithm is
simply the number of such operations it performs.  Then, our goal is
to give upper bounds on the cost functions 
$\sMM(.,.,.)$, $\sMS(.,.)$, $\sMI(.,.)$, 
which are such that the problems 
${\sf mul}(\mcL,\alpha,\beta)$, 
${\sf solve}(\mcL,\alpha)$, 
${\sf inv}(\mcL,\alpha)$,
can be solved in respective times 
$$
\sMM(\mcL,\alpha,\beta), 
\quad
\sMS(\mcL,\alpha),
\quad
\sMI(\mcL,\alpha).
$$ 
In particular, $\sMM(\mcL,\alpha,1)$ denotes the cost of 
a structured matrix-vector product
$\mA \,\v$.

A few further problems can be solved as direct extensions of these operations.
Indeed, a simple transformation (see e.g.~\cite[Theorem~1.5.2]{Pan01}) shows
that if a matrix~$\mA$ is structured with respect to one of the operators
associated with $(\vP,\vQ)$, its transpose~$\mA^t$ is structured as well, 
with respect to one of the operators associated with $(\vQ,\vP)$. 
This is summarized as follows, 
where $(\mM,\mN) \in \{\CC_\vP,\CC_\vP^t\} \times \{\CC_\vQ,\CC_\vQ^t\}$:
\begin{align*}
\nabla_{\mM,\mN}(\mA) = \mG \mH^t & \iff  \nabla_{\mN^t,\mM^t}(\mA^t) = (-\mH)\, \mG^t, \\
\Delta_{\mM,\mN}(\mA) = \mG \mH^t & \iff  \Delta_{\mN^t,\mM^t}(\mA^t) = \mH\, \mG^t.
\end{align*}
Thus, from our results on multiplication, inversion, and linear system 
solving for the matrix $\mA$, one directly obtains similar results for 
the same operations with~$\mA^t$.

\smallskip\noindent{\bf Cost measures.} 
We let $\sM$ be a multiplication time function for $\F[x]$, that is, $\sM:\NN_{>0}
\to \RR_{>0}$ is such that two polynomials of degree less than $d$ 
can be multiplied in $\sM(d)$ arithmetic operations in $\F$;
$\sM$ must also satisfy the super-linearity properties of~\cite[Ch.~8]{GaGe13}.  
It follows from~\cite{ScSt71,Schonhage77} 
that we can always take $\sM(d) \in O(d \, \lg(d) \, \lg\lg(d))$,
where $\lg(d) = \log(\max(d, 2))$ and $\log$ is the binary logarithm.
(This definition of $\lg(d)$ ensures that expressions like 
$d \, \lg(d) \, \lg\lg(d)$ do not vanish at $d=1$.)
Hence, from now on we assume that $\sM$ is quasi-linear in $d$,
so that $\sM(d) = O(d^{1+\epsilon})$ for all $\epsilon>0$,
and that $\sM(1)=1$.

Here and hereafter, $\omega$ denotes any real number such that 
two $n\times n$ matrices over~$\F$ can be multiplied 
using $O(n^\omega)$ arithmetic operations in $\F$.
We have $\omega < 2.38$~\cite{LeGall14} as $n\to\infty$ and,
for moderate dimensions (say, $n<10^6$), upper bounds on $\omega$
range from $2.81$~\cite{Strassen69} down to $2.78$~\cite{Pan82};
on the other hand, we have the trivial lower bound $\omega \ge 2$.
(Our cost analyses will be given for $\omega>2$ for simplicity, 
but note that their extension to the case $\omega=2$ would be straightforward 
and imply no more than some extra logarithmic factors.)

The running time of our algorithms will depend on the cost of some
polynomial matrix operations. Let $\sMMa(d,n):\RR_{\ge 1}\times \NN_{>0} \to \RR_{>0}$ 
be such that two $n\times n$ matrices over $\F[x]_d$ can be
multiplied in $\sMMa(d,n)$ arithmetic operations~in~$\F$. (It will be convenient
to allow non-integer values of $d$; this does not change the
definition.) One can use the following estimates:
$$\sMMa(d,n)=\begin{cases}
  O(\sM(d)n^\omega) & \text{in general,}\\[2mm]
  O(d n^\omega+\sM(d)n^2) & \text{if char}(\F)=0 \text{~or~} |\F|\ge 2d;
\end{cases}$$
the former follows from~\cite{CaKa91} and the latter is
from~\cite{BoSc05}, using evaluation and interpolation at geometric
sequences; in both cases, this is $\widetilde{O}(d n^\omega)$,
where the soft-O notation $\widetilde{O}(\cdot)$  
means that we omit polylogarithmic factors.
Our costs will be expressed using not quite $\sMMa$, but two related functions
$\sMMa'$ and $\sMMa''$, which should be thought of as being ``close''
to $\sMMa$ up to logarithmic factors.  These two functions are defined
as follows.  For $n$ a power of two,
$$\sMMa'(d,n) = \sum_{k=0}^{\log(n)} 2^k \sMMa\left(2^k d,
\frac{n}{2^k} \right )$$ and, for general $n$,
$\sMMa'(d,n)=\sMMa'(d,\overline{n})$ with $\overline{n}$ the smallest
power of two greater than or equal to $n$.  Using the super-linearity of 
$\sM$, the above choices for $\sMMa$ give
\begin{equation} \label{eq:M_prime_mat}
\sMMa'(d,n)=\begin{cases}
  O(\sM(dn)n^{\omega-1}) & \text{in general,}\\[2mm]
  O(d n^\omega+\sM(dn)n\lg(n)) & \text{if char}(\F)=0 \text{~or~} |\F|\ge 2dn;
\end{cases}
\end{equation}
in both cases, this is again $\widetilde{O}(d n^\omega)$. 
Note on the other hand that $dn^2 \le n\,\sM(dn) \le \sMMa'(d,n)$.
Now, for $d$ a power of two, let
\[
\sMMa''(d,n) = \sum_{k=0}^{\log(d)} 2^k \sMMa'\left(\frac{d}{2^k}, n \right )
\]
and, for general $d$, let $\sMMa''(d,n)=\sMMa''(\bar{d},n)$ with $\bar d$ 
the smallest power of two greater than or equal to $d$.
Using the two previous bounds on $\sMMa'$, we deduce that 
$$\sMMa''(d,n)=\begin{cases}
  O\big(\sM(dn)n^{\omega-1} \lg(d) \big) & \text{in general,}\\[2mm]
  O\big((d n^\omega+\sM(dn) n\lg(n)) \lg(d) \big) & \text{if char}(\F)=0 \text{~or~} |\F|\ge 2dn;
\end{cases}$$
again, in both cases this is $\widetilde{O}(d n^\omega)$.  Conversely, we have
$dn^\omega = O\big(\sMMa''(d,n)\big)$, since $\sMMa''(d,n) \ge
d\,\sMMa'(1,n) \ge d\,\sMMa(1,n)$.  Remark also that if we assume,
similarly to the super-linearity assumptions on $\sM$, that for every
positive integer $n$ the function $d \mapsto \sMMa'(d,n)/d$ is
nondecreasing, then we have in all cases the simple bound \[
\sMMa''(d,n) \le \sMMa'(d,n)(1+ \log(d)). \]
Finally, note that for $d=1$ and $\omega > 2$,
$$\sMMa''(1,n) = \sMMa'(1,n) = O(n^\omega).$$ 
To check the second equality, one may fix some $\epsilon$ such that $0 < \epsilon < \omega -2$,
so that $\sM(m) = O(m^{1+\epsilon})$ and thus $\sMMa'(1,n) = O\big(\sum_{k \ge 0} n^\omega (2^k)^{2+\epsilon-\omega}\big)$ 
with $2+\epsilon-\omega < 0$.

To any family of polynomials $\vP=P_1,\dots,P_d$ with $m_i=\deg(P_i) >
0$ for all $i$, we will associate two cost measures, $\sC(\vP)$ and
$\sD(\vP)$, related to Chinese remaindering and (extended) GCDs; both
range between $O(m)$ or $O(\sM(m))$ and $O(\sM(m) \lg(m))$, where we
write $m=m_1+\cdots+m_d$ as before. In both cases, we give a general
formula and point out particular cases with smaller estimates.

First, in the particular case where 
$P_i = x-x_i$ for all $i$ with the $x_i$ in geometric progression,
we take $\sC(\vP)=\sM(m)$. In all other cases, we take
$$\sC(\vP) = \sM(m)\left ( 1+ \sum_{1 \le i \le d} \frac {m_i}m
\log\left (\frac{m}{m_i} \right )\right ).$$ Since the value of the
above sum ranges from zero to $\log (d)$ (see for
example~\cite[p.~38]{BuClSh97}), the function $\sC(\vP)$ always
satisfies \[ \sM(m) \le \sC(\vP) \le \sM(m)(1+\log (d)); \] in
particular, we have $\sC(\vP) = O(\sM(m) \lg(d))$.

The definition of the function $\sD(\vP)$ starts with a particular
case as well: when $d=1$ and $P_1=x^m-\varphi$ with $\varphi \in\F$,
we take $\sD(\vP)=m$. Otherwise, we write
$$\sD(\vP) = \sum_{1\le i \le d} \sM(m_i) \lg(m_i).$$

Table~\ref{tab:particular-cases-for-the-cost-measures-C-and-D}
displays the values of $\sC(\vP)$ and $\sD(\vP)$ in the two special
cases mentioned above, but also their estimates when $d=O(1)$
or $m_i=O(1)$ for $1 \le i \le d$.  We see that when $d=O(1)$ we are
essentially in the best case for $\sC(\vP)$ and in the worst case for
$\sD(\vP)$, and that the situation is reversed when $m_i = O(1)$ for
all $i$.  In other words, the families with best and worst cases for
$\sD(\vP)$ are the opposite of those for $\sC(\vP)$.

\renewcommand{\arraystretch}{1.5}
\begin{table}[ht]
\caption{Upper bounds for $\sC(\vP)$ and $\sD(\vP)$ in some special cases.}\label{tab:particular-cases-for-the-cost-measures-C-and-D}
\begin{center}
\begin{tabular}{ccc}
& $\sC(\vP)$ & $\sD(\vP)$ \\ \hline
\begin{tabular}{c}
$P_i=x-x_i$ for all $i$, with \vspace*{-0.3cm} \\
the $x_i$ in geometric progression 
\end{tabular} & $\sM(m)$ & $m$  \\
$d=1$ and $P_1=x^m-\varphi$ & $\sM(m)$ & $m$ \\
$m_i = O(1)$ for all $i$ & \quad $O(\sM(m)\lg(m))$ \quad & $O(m)$\\
$d=O(1)$ & $O(\sM(m))$ & \quad $O(\sM(m)\lg(m))$ \quad 
\end{tabular}
\end{center}
\end{table}
\renewcommand{\arraystretch}{1.0}

Finally, when considering two families of polynomials $\vP$ and $\vQ$, 
we will write $$\sC(\vP,\vQ)=\sC(\vP)+\sC(\vQ) 
\qquad \textnormal{and}\qquad 
\sD(\vP,\vQ)=\sD(\vP)+\sD(\vQ).$$

\smallskip\noindent{\bf Background work.} 
Following landmark publications such 
as~\cite{Morf74,KaKuMo79,M80,BA80}, the literature on structured linear
systems has vastly developed, and we refer especially to~\cite{Pan01, Pan15} 
for comprehensive overviews.
It was recognized early that a key ingredient for the development of
fast algorithms was the ability to perform matrix-vector products
efficiently. For Toeplitz-like and Hankel-like systems, this was done
in~\cite{M80,BA80}; for Vandermonde and Cauchy structures, this is
in~\cite{GoOl94,GoOl94b}. In our notation, these references establish bounds of the form
$$\sMM(\mcL,\alpha,1) = O(\alpha \sM(m)) \quad\text{or}\quad
\sMM(\mcL,\alpha,1) = O(\alpha \sM(m)\lg(m)),$$ for operators of
format $(m,m)$ of Toeplitz \!/\! Hankel type or Vandermonde \!/\!
Cauchy type, respectively.

If a matrix $\mB$ has $\beta$ columns, then the product $\mA \mB$ can be
computed by performing~$\beta$ matrix-vector products with $\mA$.  In other
words, we can take 
$$\sMM(\mcL,\alpha,\beta) \le \beta\, \sMM(\mcL,\alpha,1),$$
which yields the bounds $O(\alpha \beta \sM(m))$ and $O(\alpha
\beta \sM(m) \lg(m))$ for the cases seen above.

Such matrix-matrix products, with $\beta \simeq \alpha$, are the main
ingredients in the Morf \!/\! Bitmead-Anderson (MBA) algorithm to
invert structured matrices~\cite{BA80,M80} of Toeplitz / Hankel type,
which combines the displacement rank approach  
of~\cite{FrMoKaLj79,KaKuMo79,KaKuMo79b} with Strassen's divide-and-conquer 
approach for dense matrix inversion~\cite{Strassen69}. This algorithm 
initially required several
genericity conditions to hold but, based on the randomized regularization 
technique of Kaltofen and Saunders~\cite{KaSa91}, it was then further extended
in~\cite{K94,Kaltofen95} to handle arbitrary matrices 
(see also~\cite[pp.~204--208]{BiPa94} and~\cite[\S5.5--5.7]{Pan01}). 
For operators of format $(m,m)$ of Toeplitz, Hankel, Vandermonde or Cauchy
type, the cost of inversion $\sMI(\mcL,\alpha)$ can then (roughly
speaking) be taken in $$O(\sMM(\mcL,\alpha,\alpha)\lg(m));$$
see~\cite{PZ00} for the case of Vandermonde-like and
Cauchy-like matrices, 
and~\cite{OlsPan98,Pan99,Pan00} for a unified treatment of all four structures.
Using the above upper bound on the cost of
multiplication $\sMM(\mcL,\alpha,\alpha)$, we see that this is $\alpha
\lg(m)$ times the cost of a matrix-vector product (strictly speaking,
some precautions are necessary to make such a statement; for instance,
the algorithm becomes probabilistic).

This results in running times of the form $O(\alpha^2 \sM(m)\lg(m))$
or $O(\alpha^2 \sM(m) \lg(m)^2)$ for inverting respectively Toeplitz \!/\!
Hankel-like matrices and Vandermonde \!/\! Cauchy-like matrices of size
$m$ and displacement rank $\alpha$. In the Vandermonde \!/\! Cauchy case,
another approach due to Pan~\cite{Pan90} proceeds by reduction to the
Toeplitz \!/\! Hankel case; the running time is then reduced to $O(\alpha^2
\sM(m)\lg(m))$.

Going beyond the four basic structures, Olshevsky and Shokrollahi
introduced a common generalization thereof, where the displacement
matrices take the form of block-diagonal Jordan
matrices~\cite{OlSh00}. We will discuss their result in more detail
after stating our main theorems.  For the moment, we mention that their
paper gives an algorithm for matrix-vector product with running times
ranging from
$$\sMM(\mcL,\alpha,1) = O(\alpha \sM(m)) \quad\text{to}\quad
\sMM(\mcL,\alpha,1) = O(\alpha \sM(m)\lg(m)),$$ depending on the
block configuration of the supporting Jordan matrices; for the four
basic structures seen above, we recover the running times seen before.
In~\cite{OlSh99}, this result is used to sketch an extension of the
MBA algorithm to such matrices. As before, the running time for
inversion increases by a factor $\alpha\lg(m)$ compared to the time
for matrix-vector multiplication, ranging from 
$$\sMI(\mcL,\alpha) = O(\alpha^2 \sM(m)\lg(m)) \quad\text{to}\quad
\sMI(\mcL,\alpha) = O(\alpha^2 \sM(m)\lg(m)^2),$$ depending on the
operator.

When $\alpha=O(1)$, all above results for matrix-vector product or
matrix inversion are within a polylogarithmic factor of being
linear-time. However, when $\alpha$ is not constant, this is not the
case anymore; in the worst case where $\alpha \simeq m$, the running
times for inversion grow like $m^3$ (neglecting logarithmic factors
again), whereas dense linear algebra algorithms take time
$O(m^\omega)$. In~\cite{BoJeSc08}, Bostan, Jeannerod and Schost gave
algorithms for inversion of structured matrices of the four classical
types with running time $O(\alpha^{\omega-1} \sM(m)\lg(m)^2)$.  This
is satisfactory for $\alpha\simeq m$, since we recover the cost of dense
linear algebra, up to logarithmic factors. However, for $\alpha =
O(1)$, this algorithm is slightly slower (by a factor $\lg(m)$) than
the MBA algorithm. 

\smallskip\noindent{\bf Main results.} 
The results in this paper cover in a uniform manner the general class
of operators based on companion matrices introduced above. 
Inspired by~\cite{BoJeSc08}, we obtain 
algorithms for structured matrix multiplication, inversion, and linear system solving
whose running times grow with $\alpha$ as $\alpha^{\omega-1}$ instead
of $\alpha^2$; however, we manage to avoid the loss of the $\lg(m)$
factor observed in~\cite{BoJeSc08}, thus improving on the results of that paper. 
All our algorithms assume exact arithmetic (and we do not claim any 
kind of accuracy guarantee when using finite precision arithmetic).
Note also that they are deterministic for multiplication
and, similarly to~\cite{BoJeSc08},
randomized (Las Vegas) for inversion and linear system solving.

Let $\vP$ and $\vQ$ be as before. We give in Section~\ref{sec:invert}
inversion formulas for the basic operators
$\mcL=\nabla_{\CC_\vP,\CC_\vQ^t}$ or $\mcL=\Delta_{\CC_\vP,\CC_\vQ^t}$
that allow us to recover a matrix $\mA$ from one of its generators,
generalizing simultaneously previous results for Toeplitz, Hankel,
Vandermonde, and Cauchy displacement operators. From these formulas,
one readily deduces that a matrix-vector product $\mA \u$ can be
computed in time $$\sMM(\mcL,\alpha,1)=O( \sD(\vP,\vQ)+\alpha
\sC(\vP,\vQ)),$$ which is $\widetilde{O}(\alpha p)$ for $p=\max(m,n)$.

For matrix-matrix products, the direct approach which is based on the estimate
$\sMM(\mcL,\alpha,\beta) \le \beta \,\sMM(\mcL,\alpha,1)$ thus leads to
$\sMM(\mcL,\alpha,\beta) = \widetilde{O}(\alpha \beta p)$, which is sub-optimal, as we pointed out
before. The following theorem improves on this straightforward
approach, for all operators associated with $(\vP,\vQ)$.

\begin{theorem}\label{theo:mainSylv}
  Let $\vP=P_1,\dots,P_d$ and $\vQ=Q_1,\dots,Q_e$ be 
  two families of polynomials in $\F[x]$.
  Assume that in each of these families the polynomials are monic, nonconstant and  
	pairwise coprime,
  and denote $m = \sum_{i=1}^d \deg(P_i)$ and $n = \sum_{i=1}^e \deg(Q_i)$,
Then, for any invertible operator $\mcL$ associated with $(\vP,\vQ)$, we can take
  $$\sMM(\mcL, \alpha, \beta)=O\left
  (\frac{\beta'}{\alpha'}\sMMa'\left(\frac{p}{\alpha'},\alpha'\right)
  + \sD(\vP,\vQ)+\beta'\sC(\vP,\vQ) \right),$$
  with $p=\max(m,n)$,  $\alpha'=\min(\alpha,\beta)$, and $\beta'=\max(\alpha,\beta)$. 
Furthermore, 
  $$\sMI(\mcL, \alpha), \,\,\,  \sMS(\mcL, \alpha) \,=\, O\left (\sMMa''\left(\frac{p}{\alpha},\alpha\right) \right)\!.$$ 
\end{theorem}
Using the estimates for $\sMMa'$ given in~(\ref{eq:M_prime_mat}), the 
estimate for $\sMM(\mcL, \alpha, \beta)$ becomes
$$\begin{cases}
 O\left( {\alpha'}^{\omega-2}\beta'\sM(p) + \beta'\sM(p)\lg(p) \right)  & \text{in general,}\\[2mm]
 O\left( {\alpha'}^{\omega-2}\beta' p  + \beta' \sM(p)\lg(p) \right)  & \text{if char}(\F)=0 \text{~or~} |\F|\ge 2p.
\end{cases}$$
Thus, these results provide a continuum between two extreme cases.
When $\alpha=\beta=O(1)$, the cost is $O(\sM(p)\lg(p))$; when $\alpha$ 
and $\beta$ are large, the first term is dominant, with total
costs respectively 
$O({\alpha'}^{\omega-2}\beta'\sM(p))$ and
$O({\alpha'}^{\omega-2}\beta'p)$.  
Disregarding logarithmic factors,
this is always $\widetilde{O}( {\alpha'}^{\omega-2}\beta'p)$: this
matches the cost (up to logarithmic factors) of dense, 
unstructured linear algebra algorithms for multiplying matrices of dimensions $\alpha' \times p$ and $p \times \beta'$ with $\alpha' \le
\min(p,\beta')$. 

In fact, when $\alpha = p \le \beta$
we recover exactly the cost bound $O(\beta p^{\omega-1})$
of the multiplication of two dense unstructured matrices of dimensions $p \times p$ and $p \times \beta$,
since then $\sMM(\mcL, p, \beta) = O\big( \frac{\beta}{p} \sMMa'(1,p) + \M(p)\lg(p)\big)$
and $\sMMa'(1,p) = O(p^\omega)$.

For inversion and system solving, the bounds given above on $\sMMa''$ yield
$$\begin{cases}
  O\big(\alpha^{\omega-1}\sM(p)\lg(p)\big) & \text{in general,}\\[2mm]
  O\big( \alpha^{\omega-1}p \,\lg(p)+ \alpha\lg(\alpha)\sM(p)\lg(p) \big) & \text{if char}(\F)=0 \text{~or~} |\F|\ge 2p;
\end{cases}$$
and, when $\alpha = p$ (so that the input matrix is $p\times p$ and not structured with respect to the operator~$\mathcal{L}$),
the obtained cost is $O(\sMMa''(1,p)) \subset O(p^\omega)$, as can be expected.

The following theorem highlights situations where we can take
both $\sC(\vP,\vQ)$ and $\sD(\vP,\vQ)$ in $O(\sM(p))$.  In this case,
we get slightly better bounds than in the general case for the
problems ${\sf mul}$ (for inversion and solve, the terms
above are always negligible, for any choice of $\vP$ and $\vQ$). In
view of Tables~\ref{tab3} and~\ref{tab3b}, we see that these special
cases correspond to operators of Toeplitz \!/\! Hankel type, or
Vandermonde \!/\!  Cauchy type, when their supports are points in
geometric progression.

\begin{theorem}\label{theo:mainSylv2}
  All notation and assumptions being as in Theorem~\ref{theo:mainSylv},
  suppose in addition that one of the following holds:
  \begin{itemize}
  \item either $d=1$ and $P_1=x^m-\varphi$ for some $\varphi\in\F$, or
    $d=m$ and
    there exist  $u,q \in \F$ such that 
    $P_i=x- u q^i$ for all $i$;
  \item either $e=1$ and $Q_1=x^n-\psi$ for some $\psi\in\F$, or
    $e=n$ and
    there exist  $v,r \in \F$ such that 
    $Q_j=x- v r^j$ for all $j$.
  \end{itemize}
Then, for any invertible operator $\mcL$ associated with $(\vP,\vQ)$, we can take
$$\sMM(\mcL, \alpha, \beta)=O\left (\frac{\beta'}{\alpha'}\sMMa'\left(\frac{p}{\alpha'},\alpha'\right)
+\beta'\sM(p) \right)\!,$$
with $p=\max(m,n)$,  $\alpha'=\min(\alpha,\beta)$, and $\beta'=\max(\alpha,\beta)$.  
\end{theorem}
Using once again the bounds on $\sMMa'$ given in~(\ref{eq:M_prime_mat}), this is thus
$$\begin{cases}
  O({\alpha'}^{\omega-2}\beta'\sM(p)) & \text{in general,}\\[2mm]
  O({\alpha'}^{\omega-2}\beta' p +\beta'\lg(\alpha') \sM(p)) & \text{if char}(\F)=0 \text{~or~} |\F|\ge 2p.
\end{cases}$$

\smallskip\noindent{\bf Comparison with previous work.} 
Let us briefly compare our
results with previous ones, in increasing order of generality.

For classical operators of Hankel, Toeplitz, Vandermonde or Cauchy type,
our results match classical ones in the cases $\alpha=\beta=O(1)$ (for
multiplication) or $\alpha=O(1)$ (for inversion). When we drop such
assumptions, our results improve on all previous ones. For instance,
for the inversion of Toeplitz-like or Hankel-like matrices, the best previous
results were either $O({\alpha}^2 \sM(m)\lg(m))$
in~\cite{M80,BA80,K94}, or $O({\alpha}^{\omega-1}\sM(m)\lg(m)^2)$
in~\cite{BoJeSc08}. We improve them simultaneously, by dropping the cost to
$O({\alpha}^{\omega-1}\sM(m)\lg(m))$.

We mentioned earlier the work of Olshevsky and
Shokrollahi~\cite{OlSh99,OlSh00}, who consider Sylvester operators
where the displacement matrices are block-diagonal, with Jordan blocks. A
Jordan block of size $m$ associated with $\lambda \in \F$ is similar
to the companion matrix of $(x-\lambda)^m$; the similarity matrix is
the ``Pascal'' matrix of the mapping $F(x)\mapsto F(x+\lambda)$, for
$F \in \F[x]_m$. Because the similarity matrix and its inverse can be
applied in time $O(\sM(m))$ to a vector~\cite{AhStUl75}, the problem
considered by Olshevsky and Shokrollahi is essentially equivalent to a
particular case of ours.

As it turns out, for these particular cases, the results
in~\cite{OlSh99,OlSh00} for matrix \!/\!  vector product and inversion
are equivalent to ours, when the displacement rank $\alpha$ is
$O(1)$. For larger $\alpha$, our results improve on those
of~\cite{OlSh99,OlSh00}, whose costs are quadratic in $\alpha$.

Finally, to the best of our knowledge, no previous work addressed the
general case of block-companion matrices that we deal with in this
paper.

\smallskip\noindent{\bf An application.}
We are able to address problems related to {\em simultaneous
  approximation} using our operators. Suppose that we are given
$\vP=P_1,\dots,P_d$, with sum of degrees $m$, together with ``residuals''
$$R_{1,1},\dots,R_{1,\alpha},\dots,R_{d,1},\dots,R_{d,\alpha},$$ with
all $R_{i,j}$ in $\F[x]$ such that $\deg(R_{i,j}) < \deg(P_i)$ for all
$i,j$. We are looking for polynomials $f_1,\dots,f_\alpha$ in $\F[x]$
such that $f_1 R_{i,1} + \cdots + f_\alpha R_{i,\alpha} = 0 \bmod P_i$
holds for $i=1,\dots,d$, and with prescribed degree bounds (the
polynomial $f_j$ should be in $\F[x]_{n_j}$, with $n_1 + \cdots +
n_\alpha = O(m)$).  

The matrix of the corresponding linear system, where the coefficients
of $f_1,\dots,f_\alpha$ are unknown, has size $O(m)$ and displacement
rank $O(\alpha)$ for the operator $\Delta_{\CC_\vP,\CC_\vQ^t}$, where
$\vQ$ is the ``vector'' consisting of the unique polynomial
$Q_1=x^m-\varphi$, where $\varphi$ is chosen such that the assumptions
of Theorem~\ref{theo:mainSylv} hold. Hence, we can solve such a system
using $O(\alpha^{\omega-1} \sM(m) \log(m))$ base field operations. 
This is to be compared with an algorithm given in~\cite{ChJeNeScVi15},
that has cost $O((\alpha+d)^{\omega-1} \sM(m) \log(m)^2)$, but which does
not require the assumption that $P_1,\dots,P_d$ be pairwise coprime.

This problem generalizes Hermite-Pad\'e (with $d=1$ and $P_1=x^m$) and
so-called M-Pad\'e approximation, with $P_i = (x-x_i)^{m_i}$, for
pairwise distinct $x_i$ in $\F$. A typical instance is the
reconstruction of the minimal polynomial of an algebraic function.
Suppose that $f$ is a root of a polynomial $S \in \F[x][y]$, so $f$
lies in the algebraic closure of $\F(x)$. Given the power series
expansion of $f$ at high enough precision around a point~$x_1$, it is
possible to reconstruct $S$ by means of Hermite-Pad\'e approximation.
Using our algorithm, we obtain the same output for the
same cost, starting this time from approximations at points $x_1,\dots,x_d$,
each of them requiring smaller precision.

\smallskip\noindent{\bf Organization of the paper.} 
Section~\ref{sec:prelim} introduces notation used 
all along this paper, and discusses a few classical questions about
polynomial arithmetic, such as Chinese
remaindering. Section~\ref{sec:invert} gives {\em inversion formulas}
for the operators we discuss, which generalize well-known results for
classical operators. In Section~\ref{sec:equiv}, we use these formulas
in order to reduce most of our questions to similar questions for {\em
  basic operators}, or even Toeplitz / Hankel operators.
Section~\ref{sec:multiplication_alogrithms} gives the main
algorithmic result in this paper, an improved algorithm for the
multiplication of structured matrices. 
Finally, in
Section~\ref{sec:inverse}, we show how this new multiplication algorithm
can be used within the MBA approach to accelerate structured inversion and structured system solving.

\smallskip\noindent{\bf Acknowledgements.} 
We thank an anonymous referee for many detailed comments.

\section{Preliminaries}\label{sec:prelim}

In this section, we review some classical operations on polynomials
and matrices: we introduce a short list of useful matrices, such as
multiplication matrices, reversal matrices and Krylov matrices, and we
discuss questions related to modular computations, or multiple
reduction and Chinese remaindering.

Throughout the article, we will use some well-known
complexity results on polynomial arithmetic; as a general rule, they
can be found in the books~\cite{GaGe13},~\cite{BuClSh97}, or~\cite{Pan01}.

\subsection{Basic notation}

Matrices (resp.~vectors) are written in upper-case (resp. lower-case)
{\sf sans-serif} font. If $\mA$ (resp.~$\mB$, $\mC$, \dots) is a
matrix, $\a_i$ (resp.~$\b_i$, $\c_i$, \dots) is its $i$th column.  If
$\x$ (resp.~$\y$, $\z$, \dots) is a vector, its $i$th entry is written
$x_i$ (resp.~$y_i,z_i$, \dots).  Special matrices (cyclic, companion,
diagonal,~\dots) will be written with blackboard bold letters ($\ZZ$,
$\CC$, $\D$, \dots). The entries of an $m \times n$ matrix $\mA$ are
indexed from $0$ to $m-1$ (row indices) and from $0$ to $n-1$ (column
indices).

We will also use the following notation. 
\begin{itemize}
\item For $\u=[u_0~\cdots~u_{m-1}]^t\in
  \F^m$, we write $\pol(\u)$ to denote the polynomial $u_0 + \cdots +
  u_{m-1}x^{m-1} \in \F[x]_m$. 
\item For any polynomial $F\in\F[x]$ of degree at most $d$, $\rev(F,d)$
  denotes its reverse polynomial $x^dF(1/x) \in \F[x]$.
\end{itemize}

For $m \ge 1$, $\J_m$ is the $m \times m$ reversal matrix, with $1$s
on the antidiagonal only.  More generally, for $1 \le \ell \le m$,
$\J_{\ell,m}$ is the $m \times m$ matrix with $1$s on the upper
antidiagonal entries of indices $(\ell-1,0),\dots,(0,\ell-1)$, so
$\J_{m,m}=\J_m$; by convention, for $\ell=0$, $\J_{0,m}$ is the zero
matrix of size $m$.

Finally, for $\mA \in \F^{m\times m}$, $\v \in \F^m$, and $\ell \ge 1$, 
we write $\K(\mA, \v, \ell)$ to denote the Krylov matrix in $\F^{m\times \ell}$ whose
columns are $\v, \mA\v, \dots, \mA^{\ell-1}\v$.

\subsection{Chinese remaindering and related problems}\label{ssec:poly}

In this subsection, we introduce additional notation and give basic results
related to multiple reduction and Chinese remaindering.  Consider
pairwise-coprime monic polynomials $\vP=P_1,\dots,P_d$ in $\F[x]$ with
$\deg(P_i) = m_i$, and let $m=m_1+\cdots+m_d$ and $P=P_1 \cdots P_d$. 

Recall that we associated with the family $\vP$ the cost functions
$\sC(\vP)$ and $\sD(\vP)$. They will help us measure the
cost of the following operations. To $\vP$, we associate the {\em
  multiple reduction} mapping, and its inverse, {\em Chinese
  remaindering}, defined by
$$\begin{array}{cccc}
\red_\vP:~ & \F[x]_m & \to & \F[x]_{m_1} \times \cdots \times \F[x]_{m_d}\\
& A & \mapsto & (A \bmod P_1,\dots,A \bmod P_d)
 \end{array}$$
and
$$\begin{array}{cccc}
\crt_\vP:~ & \F[x]_{m_1} \times \cdots \times \F[x]_{m_d} & \to & \F[x]_m \\
  & (A_1,\dots,A_d) & \mapsto & A,\text{~such that $A \bmod P_i=A_i$ for all $i$.}
 \end{array}$$
A related question is {\em linear recombination}, defined as the following isomorphism:
$$\begin{array}{cccc}
\comb_\vP:~ & \F[x]_{m_1} \times \cdots \times \F[x]_{m_d} & \to & \F[x]_m \\
  & (A_1,\dots,A_d) & \mapsto & A_1 P_2 \cdots P_d + \cdots + P_1 \cdots P_{d-1} A_d.
 \end{array}$$
Fast algorithms for these three operations lead to the costs given in the next lemma.
Although such costs can be found in or deduced from~\cite{BuClSh97, GaGe13}, 
they do not seem to have been presented in this way, with a common precomputation
involving both functions $\sC(\vP)$ and $\sD(\vP)$, and then an extra cost involving only $\sC(\vP)$.

\begin{lemma}\label{lemma:WX}
  Let $\vP=P_1,\ldots,P_d$ as above be given and assume that $\mathscr{H}_\vP$ holds. Then,
  after a precomputation of time $O(\sC(\vP)+\sD(\vP))$
  that yields $P$ as a by-product, one can apply the maps $\red_\vP$,
  its inverse $\crt_\vP$, as well as $\comb_\vP$ and its inverse, to
  any vector in time $O(\sC(\vP))$.
\end{lemma}
\begin{proof}
 For $1\le i\le d$, let us define $E_i = P/P_i \bmod P_i$ and $F_i = 1/E_i\bmod P_i$.
 (In particular, if $d=1$ then $E_1 = F_1 = 1$, while if $P_i$ is the linear polynomial $x-x_i$ 
 then $E_i = 1/F_i$ is the value of the derivative of $P$ at $x_i$.)
 Using these polynomials, we first note how the mappings $\crt_\vP$ and $\comb_\vP$ relate to each other: 
 Chinese remaindering of $A_1,\dots,A_d$ is done by computing $A_iF_i \bmod P_i$ for all $i$, and then applying
 $\comb_\vP$ to these products; conversely, to perform linear recombination it suffices to
 compute the products $A_iE_i\bmod P_i$ and then to apply $\crt_\vP$.
 Note also that the inverse of $\comb_\vP$ can be obtained by first using $\red_\vP$ and then
 multiplying by the $F_i$'s modulo $P_i$.
 
  Suppose we are in the case where $P_i=x-x_i$ with the $x_i$'s in
  {\em geometric} progression. Then, we can compute $P, E_1,\ldots,E_d$ and
  apply $\red_\vP$ and $\crt_\vP$ in time $O(\sC(\vP))$ using the
  algorithms of~\cite{BoSc05}. Using the remark above on the relation
  between $\crt_\vP$ and $\comb_\vP$, the claim carries over to
  $\comb_\vP$ and its inverse, so the lemma is proved in this case.

  In the general case, we can precompute $P$, as well as apply
  $\red_\vP$ and $\comb_\vP$ (without any further precomputation) in
  time $O(\sC(\vP))$ using respectively Theorems~2.19,~3.19 and Step
  2 (or 5) of Theorem 3.21 in~\cite{BuClSh97}.
 
  Let $P^\star = P/P_1 + \cdots + P/P_d$. Since $P^\star$ is obtained
  by applying $\comb_\vP$ to the polynomials $(1,\dots,1)$, it can
  thus be precomputed in time $O(\sC(\vP))$. The modular images
  $P^\star \bmod P_i$, which coincide with the polynomials $E_i$ seen
  above, can be precomputed in the same time by applying $\red_\vP$.
  Finally, we precompute the inverses $F_i=1/E_i \bmod P_i$ in time
  $O(\sD(\vP))$ by fast extended GCD (when $d=1$, this is
  straightforward, since we have seen that $E_1=1$; otherwise, each
  $F_i$ is obtained in time $O(\sM(m_i)\lg(m_i))$ using the fast
  extended Euclidean algorithm~\cite[Cor.~3.14]{BuClSh97}).

  Using the first paragraph, we can then perform Chinese 
  remaindering for the cost $O(\sC(\vP))$ of $\comb_\vP$, plus 
  modular multiplications by the $F_i$'s. The cost of the latter
  is $O(\sM(m))$, by the super-linearity of $\sM$, which is
  $O(\sC(\vP))$. The same holds for the inverse of $\comb_\vP$,
  which is reduced to $\red_\vP$ and modular multiplications by the
  $F_i$'s.
\end{proof}

\medskip

In matrix terms, we will need the following notation in the 
next sections (here and hereafter, we use canonical monomial
bases for vector spaces such as $\F[x]_m$):
\begin{itemize}
\item $\XX_\vP \in\F^{m \times m}$ is the matrix of the mapping $\comb_\vP$;
\item $\W_\vP \in\F^{m \times m}$ is the matrix of the mapping $\red_\vP$.
\end{itemize}

\begin{lemma}\label{lemma:W}
  Let $\vP=P_1,\ldots,P_d$ as above be given and assume that
  $\mathscr{H}_\vP$ holds.  Then, after a precomputation of time
  $O(\sC(\vP)+\sD(\vP))$, one can compute  $\XX_\vP\,\u$, as well as $\W_\vP\,\u$,
  $\W_\vP^t\,\u$, $\W_\vP^{-1}\,\u$ and $\W_\vP^{-t}\,\u$ for any $\u$
  in $\F^m$ in time $O(\sC(\vP))$.
\end{lemma}
\begin{proof}
  For $\XX_\vP\,\u$, $\W_\vP\,\u$ and $\W_\vP^{-1}\,\u$, this is
  nothing else than the previous lemma. To compute with the transpose
  of $\W_\vP$, we could rely on the transposition
  principle~\cite{BoLeSc03}, but it is actually possible to give
  direct algorithms.  Algorithm~{\sf TSimulMod} in~\cite{BoLeSaScWi04}
  shows how to reduce the computation of $\W_\vP^t\,\u$ to an
  application of $\comb_\vP$ and a few power series multiplications /
  inversions that involve only $P$ and the $P_i$'s, and whose costs
  add up to $O(\sM(m))$. Thus, the claimed $O(\sC(\vP))$ follows from
  the previous lemma. It is straightforward to invert that algorithm
  step-by-step, obtaining an algorithm for computing $\W_\vP^{-t}\,\u$
  in time $O(\sC(\vP))$ as well.
\end{proof}

\subsection{Modular arithmetic}\label{ssec:modular}

Given a monic polynomial $P \in \F[x]$ of degree $m$, we will
frequently work with the residue class ring $\A=\F[x]/P$; doing so, we
will identify its elements with polynomials in $\F[x]_m$. In this
subsection, we review several questions related to computations in
$\A$, such as modular multiplication and its transpose.

As a preamble, we define two useful matrices related to $P$. First,
for $\ell \ge 1$, $\W_{P,\ell} \in \F^{m\times \ell}$ denotes the
matrix of reduction modulo $P$, which maps a polynomial $F \in
\F[x]_\ell$ to $F \bmod P$. Computationally, applying $\W_{P,\ell}$
to a vector amounts to doing a Euclidean division.
Next, writing $P=p_{0} + p_{1} x + \cdots +p_{m} x^{m}$ (so that $p_{m}=1$), we
will denote by $\Y_P$ the $m \times m$ Hankel matrix
\[
\Y_P = 
\begin{bmatrix} p_{1} & \cdots & p_{m-1} & 1\\[-1mm]
\vdots & \iddots & 1 & \\[-1mm]
p_{m-1} & \iddots & & \\
1 & & &
 \end{bmatrix}.
\]
Operations with $\Y_P$ are fast, due to the triangular Hankel nature
of this matrix.  The only non-trivial part of the following lemma
concerns the inverse of $\Y_P$. It is proved in~\cite[\S2.5]{Pan01}
for triangular Toeplitz matrices; the extension to the Hankel case is
straightforward,
since $\Y_P^{-1} = (\J_m \Y_P)^{-1} \J_m$ with $\J_m \Y_P$ triangular Toeplitz.
\begin{lemma}\label{lemma:Y}
  Given $P$, for $\u$ in $\F^m$, one can compute 
  $\Y_P\,\u$ and $\Y_P^{-1}\,\u$ in $O(\sM(m))$.
\end{lemma}

Let us now discuss multiplication in $\A$, often called {\em modular
  multiplication}. Computationally, it is done by a standard
polynomial multiplication, followed by a Euclidean division;
altogether, this takes time $O(\sM(m))$.

In matrix terms, recall that $\CC_P$ denotes the $m \times m$
companion matrix associated with $P$; equivalently, this is the matrix
of the multiplication-by-$x$ endomorphism in $\A$. More generally, for
$F$ in $\F[x]$, we denote by $\CC_{F,P}$ the $m \times m$ matrix of
multiplication by $F$ in $\A$; that is, the matrix whose $(i,j)$ entry
is the coefficient of $x^i$ in $Fx^j \bmod P$,
for $i,j=0,\dots,m-1$. Thus, by what was said above, applying $\CC_{F,P}$
to a vector can be done in time $O(\sM(m))$. Note also that
$\CC_{x,P}=\CC_P$ and that $\CC_{F,P}=F(\CC_P)$.

Yet more generally, for $\ell \ge 1$, $\CC_{F,P,\ell} \in \F^{m\times
  \ell}$ is the matrix whose entries are the coefficients of the same
polynomials $Fx^j \bmod P$ as above, but now for
$j=0,\dots,\ell-1$. The following easy lemma simply says that for $F$
in $\F[x]_m$ and $G$ in $\F[x]_\ell$, $(GF) \bmod P = ((G \bmod P)F)
\bmod P$.

\begin{lemma}\label{lemma:CC}
  Let $F$ be in $\F[x]_m$ and $\ell \ge 1$. Then
  $\CC_{F,P,\ell}=\CC_{F,P}\, \W_{P,\ell}$.
\end{lemma}

\medskip

We will also use a dual operation, involving the space $\A^*$ of
$\F$-linear forms $\A \to \F$. Such a linear form will be given by the
values it takes on the monomial basis $1,x,\dots,x^{m-1}$ of
$\A$. Then, {\em transposed multiplication} is the operation mapping
$(F,\lambda) \in \A\times\A^*$ to the linear form $\lambda'=F\circ
\lambda$, defined by $\lambda'(x^i) = \lambda(x^i F)$, where the
multiplication takes place in $\A$. The name reflects the fact that
for fixed $F$, transposed multiplication by $F$ is indeed the dual map
of the multiplication-by-$F$ endomorphism of $\A$. In matrix terms,
transposed product by $F$ amounts to multiplication by
$\CC_{F,P}^t$.

Transposed products can be computed in the same time $O(\sM(m))$ as
``standard'' modular multiplications, by a dual form of modular
multiplication~\cite{BoLeSc03}. However, such an algorithm relies on
middle product techniques~\cite{HaQuZi02}, which are not
straightforward to describe. The following lemma shows an alternative
approach, with same cost up to a constant factor: to
perform a transposed product by $F$, it suffices to do a modular
multiplication by $F$, with a pre- and post-multiplication by $\Y_P$
and its inverse (which can both be done in $O(\sM(m))$, see
Lemma~\ref{lemma:Y}). For $F=x$, this result is sometimes referred
to as saying that $\Y_P$ is a {\em symmetrizer} 
for the polynomial $P$; see~\cite[p.~455]{LaTi85}. 

\begin{lemma}\label{lemma:transp}
  For all $F \in \F[x]_{m}$, we have $\Y_P\, \CC_{F,P}^t= \, \CC_{F,P} \,
  \Y_P$. Equivalently, $\CC_{F,P}^t= \Y_P^{-1}\, \, \CC_{F,P} \,
  \Y_P$.
\end{lemma}
\begin{proof}
  By linearity, it is enough to consider the case of $F=x^\ell$, for
  $0\le \ell <m$. The case $\ell=0$ is clear, since
  $\CC_{1,P}$ is the identity matrix. For $\ell=1$, this result is
  well-known, see for instance~\cite[Ex.~3, Ch.~13]{LaTi85}; it can be
  proved by simple inspection. Once the claim is established for $x$,
  an easy induction shows that it holds for $x^\ell$, for an arbitrary
  $\ell \ge 1$, using the fact that $\CC_{x^{\ell+1},P} = \CC_P
  \,\CC_{x^\ell,P} = \CC_{x^\ell,P}\, \CC_P$.
\end{proof}

We now discuss Krylov matrices derived from (transposed) companion matrices.

\begin{lemma}\label{lemma:2}
  Let $\v \in \F^m$, $\ell \ge 1$, and $F = \pol(\v) \in \F[x]_m$. 
  Then
  \begin{itemize}
  \item $\K(\CC_P,\v,\ell)=\CC_{F,P}\, \W_{P,\ell}=\CC_{F,P,\ell}$ and
  \item $\K(\CC^t_P,\v,\ell)=\Y_P^{-1}\, \CC_{G,P}\, \W_{P,\ell}$ with $G=\pol(\Y_P\,\v)$.
  \end{itemize}
\end{lemma}
\begin{proof} 
  The first assertion is clear. Indeed, the columns of the Krylov matrix
  $\K(\CC_P,\v,\ell)$ are the coefficient vectors of the polynomials $x^j F
  \bmod P$, for $0\le j<\ell$, so that
  $\K(\CC_P,\v,\ell)=\CC_{F,P,\ell}$; the claim now follows from
  Lemma~\ref{lemma:CC}. For the second assertion, the fact
  that $\CC_{P}^t= \Y_P^{-1}\, \, \CC_{P} \, \Y_P$,
  which follows from Lemma~\ref{lemma:transp}, implies that $\K(\CC^t_P,\v,\ell)=\Y_P^{-1}\,
  \K(\CC_P,\Y_p\,\v,\ell)$. Using the first point concludes the
  proof. 
\end{proof}

\subsection{Computations with a family of polynomials}

We  now consider a family of monic polynomials $\vP=P_1,\dots,P_d$,
and briefly discuss the extension of the previous claims to this
context. As before, we write $m_i=\deg(P_i)$ and $m=m_1 +\cdots + m_d$.

In terms of matrices, we have already mentioned the definition of the
block-diagonal companion matrix $\CC_\vP$ associated with $\vP$, whose
blocks are $\CC_{P_1},\dots,\CC_{P_d}$.  Similarly, $\Y_{\vP}$
will denote the block-diagonal matrix with Hankel blocks
$\Y_{P_1},\dots,\Y_{P_d}$.  The next lemma, which is about
$\Y_{\vP}$ and its inverse, is a direct consequence of Lemma~\ref{lemma:Y}; the
complexity estimate follows from the super-linearity of
the function~$\sM$.

\begin{lemma}\label{lemma:Y2}
  Given $\vP$, for $\u$ in $\F^m$, one can compute 
  $\Y_{\vP}\,\u$ and $\Y_{\vP}^{-1}\,\u$
  in $O(\sM(m))$.
\end{lemma}
Another straightforward result is the following extension of
Lemma~\ref{lemma:transp}.
\begin{lemma}\label{lemma:transp2}
  The relation $\Y_{\vP}\, \CC_{\vP}^t = \CC_{\vP}\, \Y_{\vP}$ holds.
\end{lemma}

Next, we discuss computations with Krylov matrices derived from
$\CC_\vP$ and its transpose. We will only need some special cases, for
which very simple formulas are available (and for which invertibility
will be easy to prove). Recall that $\W_\vP\in \F^{m\times m}$ denotes
the matrix of the map $\red_\vP$ of multiple reduction modulo
$P_1,\dots,P_d$; thus, it is obtained by stacking the matrices
$\W_{P_1,m},\dots,\W_{P_d,m}$.

\begin{lemma}\label{lemma:Krylov}
  Writing $\e_{m,i}$ to denote the $i$th unit vector in $\F^m$,
  the following holds:
  \begin{itemize}
  \item $\K(\CC_\vP, \v, m)=\W_\vP$ with $\v = [\e_{m_1,1}^t |\cdots | \e_{m_d,1}^t]^t \in \F^m$, and
  \item $\K(\CC_\vP^t, \w, m)=\Y_\vP^{-1}\, \W_\vP$ with $\w = [\e_{m_1,m_1}^t |\cdots | \e_{m_d,m_d}^t]^t \in \F^m$.
  \end{itemize}
\end{lemma}
\begin{proof}
The block structure of $\CC_\vP$ and the definition of $\v$ imply that 
$\K(\CC_\vP, \v, m)$ is obtained by stacking the matrices $\K(\CC_{P_i}, \e_{m_i,1}, m)$ for $i=1,\ldots,d$.
Since $\pol(\e_{m_i,1}) = 1 \in \F[x]_{m_i}$, Lemma~\ref{lemma:2} gives
$\K(\CC_{P_i}, \e_{m_i,1}, m)=\W_{P_i,m}$ for all~$i$; this proves the first case.
In the second case, $\K(\CC_\vP^t, \w, m)$ is decomposed into blocks $\K(\CC_{P_i}^t, \e_{m_i,m_i}, m)$,
which, by Lemma~\ref{lemma:2} and since $\pol(\Y_{P_i}\e_{m_i,m_i}) = 1 \in \F[x]_{m_i}$,
are equal to $\Y_{P_i}^{-1}\, \W_{P_i,m}$.
\end{proof}

\medskip

Finally, the following lemma discusses computations involving
two families of monic polynomials $\vP$ and $\vQ$.
\begin{lemma}\label{lemma:iQ}
  Let $\vP=P_1,\dots,P_d$ and $\vQ=Q_1,\dots,Q_e$, such that
  $\mathscr{H}_\vP$ and $\mathscr{H}_\vQ$ hold, and let $P=P_1\cdots
  P_d$ and $Q=Q_1 \cdots Q_e$. Then, the following holds:
  \begin{itemize}
  \item If $\gcd(P,Q)=1$, we can compute
  all $1/Q \bmod P_i$, for $i=1,\ldots,d$, 
  in time $O(\sD(\vP,\vQ)+\sC(\vP,\vQ))$.
  \item Let $n=\deg(Q)$ and $\widetilde{Q}=\rev(Q,n)$.  If
    $\gcd(P,\widetilde{Q})=1$, we can compute all $1/\widetilde{Q}
    \bmod P_i$, for $i=1,\ldots,d$, in time $O(\sD(\vP,\vQ)+\sC(\vP,\vQ))$.
  \end{itemize}
\end{lemma}
\begin{proof} 
  Let $p=\max(m,n)$. We first discuss a particular case, where $d=e=1$
  and $P=x^m-\varphi$ and $Q=x^n-\psi$.
  In this case, we can compute $1/Q \bmod P$ in linear
  time $O(p)$~\cite[Lemma~3]{Sergeev10}, provided it is well-defined;
  the same holds for $1/\widetilde{Q} \bmod P$. Thus, our claim holds,
  since in this case, $\sD(\vP,\vQ)=m+n$.

  If $P$ is not as above, we compute $Q$ in time
  $O(\sC(\vQ)+\sD(\vQ))$ using Lemma~\ref{lemma:WX}, then $Q \bmod P$
  in time $O(\sM(p))$, which is $O(\sC(\vP,\vQ))$. Using
  Lemma~\ref{lemma:WX} again, we deduce all $Q \bmod P_i$, for $i \le
  d$, in time $O(\sC(\vP))$; finally, we invert all remainders using
  fast extended GCD with the $P_i$'s in time $O(\sum_{i\le d}
  \sM(m_i)\lg(m_i))$; under our assumption on $P$, this is
  $O(\sD(\vP))$. The total cost is $O(\sD(\vP,\vQ)+\sC(\vP,\vQ))$, so
  our claim for the inverses of $Q$ modulo the $P_i$'s holds.  
  We proceed similarly for $\widetilde{Q}$.

  The last case to consider is when $d=1$ and $P=x^m-\varphi$, but
  with $Q$ not of the form $x^n-\psi$. We proceed in the converse
  direction: we first reduce and invert $P$ modulo all~$Q_j$, for all
  $j \le e$; this takes time $O(\sD(\vP,\vQ)+\sC(\vP,\vQ))$. By Chinese
  Remaindering, we obtain $R=1/P \bmod Q$, without increasing the
  cost. Knowing $R$, we deduce $S=1/Q \bmod P$, since $R$
  and $S$ satisfy $RP+SQ =1$; this costs an extra $O(\sM(p))$, which
  is $O(\sC(\vP,\vQ))$, so our claim is proved. We proceed similarly
  for $\widetilde{Q}$.\,
\end{proof}

\section{Inverting displacement operators} \label{sec:invert}
Let us consider $\vP=P_1,\dots,P_d$ and $\vQ=Q_1,\dots,Q_e$ 
tuples of monic polynomials in $\F[x]$,
and let $m_1,\dots,m_d$, $n_1,\dots,n_e$, $m$, $n$, $P$ and $Q$ be as before.
We assume the coprimality conditions $\mathscr{H}_\vP$ and
$\mathscr{H}_\vQ$. In this section, we establish inversion
formulas for the two basic operators of Sylvester and Stein types
associated with $(\vP,\vQ)$.  In other words, given two matrices
$(\mG,\mH)$ in $\F^{m\times \alpha}\times \F^{n\times \alpha}$, we
show how to reconstruct the matrices $\mA$ and $\mA'$ in $\F^{m\times n}$ such that
$$\nabla_{\CC_\vP,\CC_\vQ^t}(\mA) = \mG \mH^t \quad\text{and}\quad \Delta_{\CC_\vP,\CC_\vQ^t}(\mA') = \mG \mH^t,$$
provided the corresponding operators are invertible.
We will use the following objects.
\begin{itemize}
\item We denote by $\g_1,\dots,\g_\alpha$ and $\h_1,\dots,\h_\alpha$
  the columns of $\mG$ and $\mH$. Corresponding to the partition of
  $\CC_\vP$ into blocks, we partition each $\g_k$, $k \le \alpha$, into
  smaller column vectors $\g_{1,k},\dots,\g_{d,k}$; 
  similarly, we partition each $\h_k$ into $\h_{1,k},\dots,\h_{e,k}$ 
  according to the block structure of $\CC_\vQ^t$.
  The matrices
  $\mG$ and $\mH$ themselves are partitioned into matrices
  $\mG_1,\dots,\mG_d$ and $\mH_1,\dots,\mH_e$; $\mG_i$ has size $m_i
  \times \alpha$ and columns $\g_{i,1},\dots,\g_{i,\alpha}$, whereas
  $\mH_j$ has size $n_j \times \alpha$ and columns
  $\h_{j,1},\dots,\h_{j,\alpha}$.

\item For $i=1,\dots,d$ and $k=1,\dots,\alpha$, we let
  $g_{i,k}=\pol(\g_{i,k})$, so that $g_{i,k}$ is in $\F[x]_{m_i}$.
  Similarly, for $j=1,\dots,e$, we define
  $h_{j,k}=\pol(\h_{j,k})\in\F[x]_{n_j}$.

  We  further let $\gamma_k$ be the unique
  polynomial in $\F[x]_m$ such that $\gamma_k \bmod P_i =
  g_{i,k}$ holds for $i=1,\dots,d$; in other words, we have
  $$\gamma_k = \crt_\vP(g_{1,k},\dots,g_{d,k}).$$ The polynomial
  $\eta_k$ is defined similarly, replacing $g_{i,k}$ and $P_i$ by
  $h_{j,k}$ and $Q_j$.
\end{itemize}

In the following theorem, recall that the necessary and sufficient
condition for the Sylvester operator to be invertible is that
$\gcd(P,Q)=1$; for the Stein operator, the condition is that
$\gcd(P,\widetilde{Q})=1$, with $\widetilde{Q}=\rev(Q,n)$.
\begin{theorem}\label{theo:rec}
  The following holds:
  \begin{itemize}
  \item Suppose that $\gcd(P,Q)=1$. Then, the unique matrix $\mA \in \F^{m \times n}$ 
    such that $\nabla_{\CC_\vP,\CC_\vQ^t}(\mA)=\mG \mH^t$ is given by
    $$\mA = \V_{\vP,\vQ}\, \W_\vP \left ( \sum_{k \le \alpha} \CC_{\gamma_k,P,n}\,
    \CC_{\eta_k, Q} \right ) \XX_\vQ\, \Y_\vQ,$$ where $\V_{\vP,\vQ}\in\F^{m \times m}$
    is the block-diagonal matrix with blocks $\CC_{Q^{-1},P_i}$ for
    $i=1,\dots,d$, where $\CC_{Q^{-1},P_i}$ denotes $\CC_{Q^{-1}\bmod P_i,P_i}$.
  \item Suppose that $\gcd(P,\widetilde{Q})=1$. Then, the unique matrix $\mA' \in \F^{m \times n}$ 
    such that $\Delta_{\CC_\vP,\CC_\vQ^t}(\mA')=\mG \mH^t$ is given by
    $$\mA' = \V'_{\vP,\vQ}\, \W_\vP \left ( \sum_{k \le \alpha}
    \CC_{\gamma_k,P,n}\, \J_n\, \CC_{\eta_k, Q} \right ) \XX_\vQ\,
    \Y_\vQ,$$ where $\V'_{\vP,\vQ}\in\F^{m \times m}$ is the
    block-diagonal matrix with blocks $\CC_{\widetilde{Q}^{-1},P_i}$
    for $i=1,\dots,d$, where $\CC_{\widetilde{Q}^{-1},P_i}$ denotes
    $\CC_{\widetilde{Q}^{-1}\bmod P_i,P_i}$.
  \end{itemize}
\end{theorem}

The following corollary on the cost of matrix-vector multiplication
follows directly; the more difficult case of matrix-matrix
multiplication will be handled in
Section~\ref{sec:multiplication_alogrithms}.
\begin{corollary}\label{coro:mul1}
  We can take 
  $$ \sMM(\nabla_{\CC_\vP,\CC_\vQ^t},\alpha,1) = O(\sD(\vP,\vQ) + \alpha \sC(\vP,\vQ))$$
  and
  $$ \sMM(\Delta_{\CC_\vP,\CC_\vQ^t},\alpha,1) = O(\sD(\vP,\vQ) + \alpha \sC(\vP,\vQ)).$$
\end{corollary}
\begin{proof}
  The proof is mostly the same in both cases; it amounts to estimating
  first the cost of computing the polynomials that define the matrices involved in  Theorem~\ref{theo:rec},
  and then the cost of multiplying each of these matrices by a single vector.
  
  Given the families of polynomials $\vP$ and $\vQ$, we can compute
  the products $P$ and $Q$ in time $O(\sC(\vP,\vQ)+\sD(\vP,\vQ))$, by
  Lemma~\ref{lemma:WX}. Moreover, since both assumptions $\mathscr{H}_\vP$
  and $\mathscr{H}_\vQ$ hold and depending on whether $\gcd(P,Q)=1$ or
  $\gcd(P,\widetilde{Q})=1$, we deduce from Lemma~\ref{lemma:iQ} that
  we can compute the inverses of $Q$, or of $\widetilde{Q}$, modulo
  all $P_i$'s in time $O(\sD(\vP,\vQ) + \sC(\vP,\vQ))$.  Finally,
  given further the generator $(\mG,\mH)$ in $\F^{m\times\alpha}
  \times \F^{n\times\alpha}$, it follows from Lemma~\ref{lemma:WX}
  that the polynomials $\gamma_1,\ldots,\gamma_\alpha$ and
  $\eta_1,\ldots,\eta_\alpha$ can be obtained in time $O(\sD(\vP) +
  \alpha \sC(\vP))$ and $O(\sD(\vQ) + \alpha \sC(\vQ))$, respectively.
  Thus, the overall cost of deducing the needed polynomials
  from $\vP$, $\vQ$, $\mG$, $\mH$ is in $O(\sD(\vP,\vQ) + \alpha
  \sC(\vP,\vQ))$.

  We now turn to the cost of the matrix-vector products performed
  when applying Theorem~\ref{theo:rec} to the evaluation of $\mA\u$ or
  $\mA'\u$ for some $\u$ in $\F^n$.  Lemma~\ref{lemma:Y2} shows that
  the cost of multiplying $\Y_\vQ$ by $\u$ is $O(\sM(n))$, which is
  $O(\sM(p))$.  By Lemma~\ref{lemma:W}, the cost for $\XX_\vQ$ is
  $O(\sD(\vQ) + \sC(\vQ))$.  The inner sum amounts to $O(\alpha)$
  multiplications modulo $P$ or $Q$, for a total of $O(\alpha
  \sM(p))$. Lemma~\ref{lemma:W} also shows that the cost for $\W_\vP$
  is $O(\sD(\vP)+\sC(\vP))$.  Finally, multiplying by $\V_{\vP,\vQ}$
  or $\V'_{\vP,\vQ}$ amounts to performing modular multiplications
  modulo all $P_i$'s, and this can be done in time $O(\sM(m)) \subset
  O(\sM(p))$.  Hence the cost of all these matrix-vector
  products is in $O(\sD(\vP,\vQ) + \sC(\vP,\vQ) + \alpha \sM(p))$;
  this fits the announced bound, since $\sM(p) \le
  \sC(\vP,\vQ)$.
\end{proof}

\medskip

The rest of this section is devoted to proving Theorem~\ref{theo:rec} above.

\subsection{Sylvester operator $\nabla_{\CC_\vP,\CC_\vQ^t}$} \label{subsec:inverting-sylvester}
Assuming that $\gcd(P,Q)=1$, we show how to
reconstruct the unique matrix $\mA\in\ \F^{m\times n}$ such that
$\nabla_{\CC_\vP,\CC_\vQ^t}(\mA)=\mG \mH^t$. We first show how to
reconstruct all blocks in an ad-hoc block decomposition of $\mA$, then
use Chinese remaindering.

\paragraph{Step 1: reconstruction of the blocks of $\mA$} 
Partitioning the matrix $\mA$ into blocks $\mA_{i,j}$ conformally
with the block-diagonal structure of the matrices $\CC_\vP$ and $\CC_\vQ$,
we have $\nabla_{\CC_{P_i},\CC_{Q_j}^t}(\mA_{i,j})=\mG_i\mH_j^t$ for all $i,j$. 
In this paragraph, we prove a reconstruction formula for each block $\mA_{i,j}$.

\begin{lemma}\label{lemma:Aij}
  For all $i\le d$ and $j\le e$, we have
$$\mA_{i,j} = \sum_{k \le \alpha} \CC_{g_{i,k}, P_i}\,\CC_{Q_j^{-1},P_i,n_j}\, \CC_{h_{j,k},Q_j} \, \Y_{Q_j}.$$
\end{lemma}
\begin{proof}
Fix $i$ and $j$. Note that $\CC_{P_i}$ and $\CC_{Q_j}$ cannot be
simultaneously singular, since $P_i$ and $Q_j$ are coprime.  Then, for
all $\ell \ge 1$, a slight variation of~\cite[Theorem~4.8]{PanWan03}
(with both cases ``$\CC_{P_i}$ nonsingular'' and ``$\CC_{Q_j}$ nonsingular''
covered by a single identity) gives
$$\CC_{P_i}^\ell\, \mA_{i,j} -\mA_{i,j}\, (\CC_{Q_j}^t)^\ell = \sum_{k \le
  \alpha}\K(\CC_{P_i}, \g_{i,k},\ell)\,\J_\ell\, \K(\CC_{Q_j},
\h_{j,k},\ell)^t,$$ where $\K(\CC_{P_i},\g_{i,k},\ell)\in\F^{m_i
  \times \ell}$ denotes the Krylov matrix of column length $\ell$ associated with
the matrix $\CC_{P_i}$ and the vector $\g_{i,k}$.

Writing $Q_j=q_{j,0} + \cdots +q_{j,n_j} x^{n_j}$, 
we obtain, for $1 \le \ell \le n_j$,
\begin{align*}
q_{j,\ell}\, \CC_{P_i}^\ell\, \mA_{i,j} -\mA_{i,j}\, q_{j,\ell}\, (\CC_{Q_j}^t)^\ell &=
 \sum_{k \le \alpha}\K(\CC_{P_i},\g_{i,k},\ell)\, q_{j,\ell}\,\J_\ell\, \K(\CC_{Q_j}, \h_{j,k},\ell)^t\\
&=  \sum_{k \le \alpha} \K(\CC_{P_i}, \g_{i,k},n_j)\, q_{j,\ell}\,\J_{\ell,n_j}\, \K(\CC_{Q_j}, \h_{j,k},n_j)^t,
\end{align*}
since we can inflate all matrices $\J_\ell$ to size $n_j \times n_j$,
replacing them by $\J_{\ell,n_j}$. Note that the above equality then 
also holds for $\ell=0$, since $\J_{0,n_j} = 0$ by convention.
Summing over all $\ell=0,\dots,n_j$ and using the fact that
$Q_j(\CC_{Q_j}^t)=Q_j(\CC_{Q_j})^t=0$, we deduce 
\[
Q_j(\CC_{P_i})\, \mA_{i,j} = \sum_{k \le \alpha}\K(\CC_{P_i},\g_{i,k},n_j)\,
\Y_{Q_j}\, \K(\CC_{Q_j}, \h_{j,k},n_j)^t.
\]
Using the first part of Lemma~\ref{lemma:2} together with the fact that 
for $\gcd(P_i,Q_j)=1$ the matrix $Q_j(\CC_{P_i}) = \CC_{Q_j, P_i}$ is invertible
with inverse $\CC_{Q_j, P_i}^{-1}=\CC_{Q_j^{-1}, P_i}$, we obtain
\[
\mA_{i,j} = \sum_{k \le \alpha} \CC_{Q_j^{-1}, P_i}\,\CC_{g_{i,k},P_i,n_j}\, \Y_{Q_j} \, \CC_{h_{j,k},Q_j}^t.
\]
Hence, applying Lemma~\ref{lemma:transp} to $h_{j,k}$ and since multiplication matrices mod $P_i$ commute,
\[
\mA_{i,j} =
\sum_{k \le \alpha} \CC_{g_{i,k}, P_i}\,\CC_{Q_j^{-1},P_i,n_j}\, \CC_{h_{j,k},Q_j} \, \Y_{Q_j}.
\]
\end{proof}

\paragraph{Step 2: a first reconstruction formula for $\mA$}
Let $\B_{\vP,\vQ}$ be the  $m \times n$ block matrix
$$\B_{\vP,\vQ}=\begin{bmatrix} \CC_{Q_j^{-1},P_i,n_j} \end{bmatrix}_{1 \le i
  \le d \atop 1 \le j \le e}.$$ Similarly, for $k \le \alpha$, we
define $\mC_k$ and $\mD_k$ as the block diagonal matrices,
with respectively the multiplication matrices $(\CC_{g_{i,k},P_i})_{i
  \le d}$ and $(\CC_{h_{j,k},Q_j})_{j \le e}$ on the diagonal.

Putting all blocks $\mA_{i,j}$ together, we deduce the following
reconstruction formula for $\mA$; the proof is a straightforward
application of the previous lemma.

\begin{lemma}\label{lemma:A}
We have
$$\mA = \left ( \sum_{k \le \alpha} \mC_k\,\B_{\vP,\vQ}\, \mD_k \right) \, \Y_\vQ.$$
\end{lemma}

\paragraph{Step 3: using a factorization of $\B_{\vP,\vQ}$} Next, we use the fact
that $\B_{\vP,\vQ}$ has a block structure similar to a Cauchy matrix to factor
it explicitly. For this, we introduce a matrix related to~$\W_\vP$: we
let $\W_{\vP,n}$ be the $m \times n$ matrix of multiple reduction
modulo $P_1,\dots,P_d$, taking as input a polynomial of degree less than
$n$, instead of $m$ for $\W_\vP$.

\begin{lemma}\label{lemma:B}
  The equality $\B_{\vP,\vQ}= \V_{\vP,\vQ} \, \W_{\vP,n}\, \XX_\vQ$ holds, where
  $\V_{\vP,\vQ}\in\F^{m \times m}$ is the block-diagonal matrix with blocks
  $\CC_{Q^{-1},P_i}$, for $i=1,\dots,d$.
\end{lemma}
\begin{proof}
  Observe that, given the coefficient vectors of polynomials
  $F_1,\dots,F_e$, with $F_j \in \F[x]_{n_j}$ for all $j$,
  the matrix $\B_{\vP,\vQ}$ returns the coefficients of 
  \begin{align*}
  G_i &= \frac {F_1}{Q_1} + \cdots + \frac {F_e}{Q_e} \bmod P_i \\[1mm]
   &= \frac {F_1 Q_2\cdots Q_e + \cdots + Q_1 \cdots Q_{e-1}F_e}{Q_1 \cdots Q_e} \bmod P_i, 
  \end{align*}
  for $i=1,\dots,d$. Computing the polynomials $G_i$ can be done as follows:
  \begin{itemize}
  \item compute $H=F_1 Q_2\cdots Q_e + \cdots + Q_1 \cdots Q_{e-1}F_e$ by calling $\comb_\vP$;
  \item compute $H_i = H \bmod P_i$, for $i=1,\dots,d$;
  \item deduce $G_i = H_i/Q \bmod P_i$, for $i=1,\dots,d$.
  \end{itemize}
This proves the requested factorization of $\B_{\vP,\vQ}$.
\end{proof}

\medskip

As a result, we deduce the following equalities for $\mA$:
\begin{align*}
\mA &= \left ( \sum_{k \le \alpha} \mC_k\, \V_{\vP,\vQ} \, \W_{\vP,n}\, \XX_\vQ\, \mD_k \right) \, \Y_\vQ\\
 &= \V_{\vP,\vQ} \,\left ( \sum_{k \le \alpha} \mC_k\,  \W_{\vP,n}\, \XX_\vQ\, \mD_k \right) \, \Y_\vQ.
\end{align*}
The last equality is due to the fact that the block-diagonal matrices
$\mC_k$ and $\V_{\vP,\vQ}$ commute (since all pairwise
corresponding blocks are multiplication matrices modulo the same
polynomials $P_i$).

\paragraph{Step 4: using Chinese remaindering}  The next step will allow us
to take $\W_{\vP,n}$ and $\XX_\vP$ out of the inner sum. For any polynomial
$H$, any $i \le d$ and any $k \le \alpha$, it is equivalent to {\it
  (i)} reduce $H$ modulo $P_i$ and multiply it by $g_{i,k}$ modulo
$P_i$ and {\it (ii)} multiply $H$ by the polynomial $\gamma_k$
(defined above) modulo $P$, and reduce it modulo~$P_i$.  In other
words, we have the commutation relation
$\mC_k\,  \W_{\vP,n} = \W_\vP \, \CC_{\gamma_k, P, n}$.
Similarly, $\XX_\vQ\, \mD_k=\CC_{\eta_k, Q} \,\XX_\vQ$
and this concludes the proof of the first part of Theorem~\ref{theo:rec}.

\subsection{Stein operator $\Delta_{\CC_\vP,\CC_\vQ^t}$}

The proof for Stein operator case follows the same steps as for the
Sylvester case. Let $\mA'$ be such that $\Delta_{\CC_\vP,\CC_\vQ^t}(\mA')=\mG
\mH^t$. Just like we decomposed $\mA$ into blocks $\mA_{i,j}$, we
decompose $\mA'$ into blocks $\mA'_{i,j}$, such that
$\Delta_{\CC_{P_i},\CC_{Q_j}^t}(\mA'_{i,j})=\mG_i \mH_j^t$. Extending the
notation used above, we write $\widetilde{Q_j}=\rev(Q_j,n_j)$, so that
we have $\widetilde{Q}=\widetilde{Q_1}\cdots \widetilde{Q_e}$.
The following lemma is then the analogue of Lemma~\ref{lemma:Aij} for
Stein operators and, 
as detailed in Appendix~\ref{app:proof-lem:inverse-for-stein}, 
can be proved in the same way using~\cite[Theorem~4.7]{PanWan03}.
 
\begin{lemma} \label{lem:inverse-for-stein}
  For all $i\le d$ and $j\le e$, we have
$$\mA'_{i,j} = \sum_{k \le \alpha} \CC_{g_{i,k}, P_i}\,\CC_{\widetilde{Q_j}^{-1},P_i,n_j}\,\J_{n_j}\, \CC_{h_{j,k},Q_j} \, \Y_{Q_j}.$$
\end{lemma}

Mimicking the construction in the previous section, we introduce the $m
\times n$ block matrix $\BB_{\vP,\vQ}'$ given by
$$\BB_{\vP,\vQ}'=\begin{bmatrix}
\CC_{\widetilde{Q_j}^{-1},P_i,n_j} \end{bmatrix}_{1 \le i \le d \atop 1
  \le j \le e}.$$ We will use again the matrices $\mC_k$ and $\mD_k$
introduced before, as well as the block-diagonal matrix $\D(\J_{n_j})$
having blocks $\J_{n_j}$ on the diagonal. This leads us to the
following analogue of Lemma~\ref{lemma:A}, whose proof is
straightforward.
\begin{lemma}\label{lemma:Ap}
We have
$$\mA' = \left ( \sum_{k \le \alpha} \mC_k\,\BB_{\vP,\vQ}'\, \D(\J_{n_j})\,\mD_k \right) \, \Y_\vQ.$$
\end{lemma}

The next step is to use the following factorization of
$\BB_{\vP,\vQ}'$, or more precisely of $\BB_{\vP,\vQ}'\,
\D(\J_{n_j})$. The proof is the same as that of Lemma~\ref{lemma:B},
up to taking into account the reversals induced by the matrices
$\J_{n_j}$.
\begin{lemma}\label{lemma:Bp}
  The equality $\BB_{\vP,\vQ}'\,\D(\J_{n_j}) = \V'_{\vP,\vQ} \,
  \W_{\vP,n}\, \J_n\, \XX_\vQ$ holds, where $\V'_{\vP,\vQ}\in\F^{m
    \times m}$ is the block-diagonal matrix with blocks
  $\CC_{\widetilde{Q}^{-1},P_i}$, for $i=1,\dots,d$.
\end{lemma}
We conclude the proof of our theorem as before, using the relations
\begin{align*}
\mC_k\,\V'_{\vP,\vQ} &=\V'_{\vP,\vQ}\,\mC_k,  \\[1mm]
\mC_k\, \W_{\vP,n} &= \W_\vP \, \CC_{\gamma_k, P, n},\\[1mm]
\XX_\vQ\, \mD_k&=\CC_{\eta_k, Q} \,\XX_\vQ.
\end{align*}

\section{Using operator equivalences} \label{sec:equiv}

Let $\vP,\vQ$ be as in Theorems~\ref{theo:mainSylv}
and~\ref{theo:mainSylv2}, and let as before $p=\max(m,n)$.  We are now
going to extend the complexity estimates for matrix-vector
multiplication given in Corollary~\ref{coro:mul1} to more operators
(not only the basic ones), by providing reductions to the Hankel
operator $\nabla_{\ZZ_{m,0},\ZZ_{n,1}^t}$.

\begin{theorem}\label{theo:equiv}
  Suppose that $\mathscr{H}_\vP$ and $\mathscr{H}_\vQ$ hold.
  Then for any displacement operator $\mcL$ associated with $(\vP,\vQ)$, we can 
  take
\begin{align*}
 \sMM(\mcL, \alpha, \beta) &\le  \sMM(\nabla_{\ZZ_{m,0},\ZZ_{n,1}^t}, \alpha+2, \beta) 
    +O\big (\sD(\vP,\vQ) + (\alpha+\beta)\sC(\vP,\vQ)\big ), \\[1mm]
 \sMS(\mcL, \alpha) &\le \sMS(\nabla_{\ZZ_{m,0},\ZZ_{n,1}^t}, \alpha+2)
    +O\big(\sD(\vP,\vQ) + \alpha\sC(\vP,\vQ) \big ),\\[1mm]
 \sMI(\mcL, \alpha) &\le \sMI(\nabla_{\ZZ_{m,0},\ZZ_{m,1}^t}, \alpha+2)
 +O\big(\sD(\vP,\vQ) + \alpha\sC(\vP,\vQ) + \alpha^{\omega-1} m\big ).
\end{align*}
\end{theorem}

The rest of this section is devoted to the proof of this theorem.  
In what follows, we write several formulas involving some matrices $\mG$ and~$\mH$.
In these formulas, $\mG$ and $\mH$ will always be taken in
$\F^{m\times \alpha}$ and $\F^{n\times \alpha}$, respectively, for some
$\alpha \ge 1$ (and with $n = m$ when dealing with matrix inversion).

\subsection{Reduction to basic operators} 

Table~\ref{tab4} shows that if a matrix $\mA$ is structured for one of
the operators associated with $(\vP,\vQ)$, then simple pre- and
post-multiplications transform $\mA$ into a matrix which is structured
for one of the two {\em basic} operators associated with $(\vP,\vQ)$.  All
equivalences in this table follow directly from 
the identities $\Y_{\vP}\, \CC_{\vP}^t = \CC_{\vP}\, \Y_{\vP}$ and
$\Y_{\vQ}\, \CC_{\vQ}^t = \CC_{\vQ}\, \Y_{\vQ}$ (Lemma~\ref{lemma:transp2})
and from $\Y_\vP$ and $\Y_\vQ$ being invertible and symmetric. 
Combining these equivalences with Lemma~\ref{lemma:Y2} will lead to 
the cost bounds in Lemma~\ref{lemma:redBasic} below, 
expressed in terms of basic operators.

\begin{table}[h]
\caption{Reduction to the basic operators $\nabla_{\CC_\vP,\CC_\vQ^t}$ and $\Delta_{\CC_\vP,\CC_\vQ^t}$.}\label{tab4}
\vspace*{-0.5cm}
\begin{align*}
\nabla_{\CC_\vP,\CC_\vQ}(\mA) = \mG\, \mH^t & \iff \nabla_{\CC_\vP,\CC_\vQ^t}(\mA\, \Y_\vQ) = \mG\, (\Y_\vQ\mH)^t  \\
\nabla_{\CC_\vP^t,\CC_\vQ^t}(\mA) = \mG\, \mH^t & \iff  \nabla_{\CC_\vP,\CC_\vQ^t}(\Y_\vP\, \mA) = (\Y_\vP\mG)\, \mH^t  \\
\nabla_{\CC_\vP^t,\CC_\vQ}(\mA) = \mG\, \mH^t & \iff  \nabla_{\CC_\vP,\CC_\vQ^t}(\Y_\vP\, \mA\, \Y_\vQ) = (\Y_\vP\mG) (\Y_\vQ\mH)^t  \\[3mm]
\Delta_{\CC_\vP,\CC_\vQ}(\mA) = \mG\, \mH^t & \iff  \Delta_{\CC_\vP,\CC_\vQ^t}(\mA \,\Y_\vQ) = \mG\, (\Y_\vQ\mH)^t  \\
\Delta_{\CC_\vP^t,\CC_\vQ^t}(\mA) = \mG\, \mH^t & \iff  \Delta_{\CC_\vP,\CC_\vQ^t}(\Y_\vP\, \mA) = (\Y_\vP\mG)\, \mH^t  \\
\Delta_{\CC_\vP^t,\CC_\vQ}(\mA) = \mG\, \mH^t & \iff  \Delta_{\CC_\vP,\CC_\vQ^t}(\Y_\vP\, \mA\, \Y_\vQ) = (\Y_\vP\mG) (\Y_\vQ\mH)^t.
\end{align*}
\end{table}

\begin{lemma}\label{lemma:redBasic}
  Let $\mcL$ be a displacement operator associated with $(\vP,\vQ)$. Then:
  \begin{itemize}
  \item If $\mcL$ is a Sylvester operator, we can take
     \begin{align*}
   \sMM(\mcL,\alpha,\beta)& \le \sMM(\nabla_{\CC_\vP,\CC_\vQ^t},\alpha,\beta)+O((\alpha+\beta) \sM(p)),\\[1mm]
  \sMS(\mcL,\alpha)& \le \sMS(\nabla_{\CC_\vP,\CC_\vQ^t},\alpha)+O(\alpha \sM(p)),\\[1mm]
  \sMI(\mcL,\alpha)& \le \sMI(\nabla_{\CC_\vP,\CC_\vQ^t},\alpha)+O(\alpha \sM(m));
    \end{align*}
  \item If $\mcL$ is a Stein operator, we can take
     \begin{align*}
   \sMM(\mcL,\alpha,\beta)& \le \sMM(\Delta_{\CC_\vP,\CC_\vQ^t},\alpha,\beta)+O((\alpha+\beta) \sM(p)),\\[1mm]
  \sMS(\mcL,\alpha)& \le \sMS(\Delta_{\CC_\vP,\CC_\vQ^t},\alpha)+O(\alpha \sM(p)),\\[1mm]
  \sMI(\mcL,\alpha)& \le \sMI(\Delta_{\CC_\vP,\CC_\vQ^t},\alpha)+O(\alpha \sM(m)).
    \end{align*}
  \end{itemize}
\end{lemma}
\begin{proof}
  We give the proof for the Sylvester operator
  $\mcL=\nabla_{\CC_\vP^t,\CC_\vQ}$, which appears in the third row of
  Table~\ref{tab4}; the other five cases in that table can be
  handled similarly.

  Suppose first that, given an $\mcL$-generator $(\mG,\mH)$ for $\mA
  \in \F^{m\times n}$ as well as a matrix $\mB \in \F^{n\times
    \beta}$, we want to compute $\mA \mB$.  We begin by computing a
  $\nabla_{\CC_\vP,\CC_\vQ^t}(\mA')$-generator $(\Y_\vP\mG,\Y_\vQ\mH)$
  for $\mA'=\Y_\vP\, \mA\, \Y_\vQ$
  (cf.~the third row of Table~\ref{tab4}).  Lemma~\ref{lemma:Y2}
  implies that this takes time $O(\alpha \sM(p))$.  Then, we evaluate
  $\mA \mB = \Y_\vP^{-1}\,\mA' \,\Y_\vQ^{-1}\,\mB$ from right to
  left. The products by $\Y_\vP^{-1}$ and $\Y_\vQ^{-1}$ take time
  $O(\beta \sM(p))$, and the product by $\mA'$ takes time
  $\sMM(\nabla_{\CC_\vP,\CC_\vQ^t},\alpha,\beta)$; thus, the first
  claim is proved.

  Suppose now that we are given a vector $\b \in \K^m$ in addition to the matrix $\mA$.
In order to solve the system $\mA \x =
  \b$, we solve $\mA' \x' = \b'$, with $\mA'$ as defined above and
  $\b' = \Y_\vP \b$. The latter system admits a solution if and only
  if the former one does; if $\x'$ is a solution of $\mA' \x' = \b'$, then $\x=
  \Y_\vQ \x'$ is a solution of $\mA \x = \b$.
  Hence, as before, we set up an
  $\nabla_{\CC_\vP,\CC_\vQ^t}$-generator $(\Y_\vP\mG,\Y_\vQ\mH)$ for
  $\mA'$, using $O(\alpha \sM(p))$ operations in~$\F$, and compute
  $\b'$ in time $O(\sM(m))$. Then, in time
  $\sMS(\nabla_{\CC_\vP,\CC_\vQ^t},\alpha)$ we either assert that the
  system $\mA' \x' = \b'$ has no solution, or find such a solution
  $\x'$. Finally, we recover $\x$ in time $O(\sM(n))$. This
  proves the second claim.

  Finally, assume that $m=n$ and consider the question of inverting $\mA$.  This matrix
  is invertible if and only if the matrix $\mA'$ defined above is
  invertible. Again, we set up a
  $\nabla_{\CC_\vP,\CC_\vQ^t}$-generator $(\Y_\vP\mG,\Y_\vQ\mH)$ for
  $\mA'$, using $O(\alpha \sM(m))$ operations in $\F$.  Then, in time
  $\sMI(\nabla_{\CC_\vP,\CC_\vQ^t},\alpha)$ we either assert that
  $\mA'$ is not invertible or deduce a
  $\nabla_{\CC_\vQ^t,\CC_\vP}$-generator for $\mA'^{-1}$, say
  $(\mG_{\rm inv}',\mH_{\rm inv}')$. Finally, if $\mA'$ is invertible
  then, using $\CC_{\vP}^t = \Y_{\vP}^{-1}\, \CC_{\vP}\, \Y_{\vP}$ and
  $\CC_{\vQ} = \Y_{\vQ}\, \CC_{\vQ}^t \Y_{\vQ}^{-1}$, we obtain
  a $\nabla_{\CC_\vQ,\CC_\vP^t}$-generator $(\Y_\vQ\mG_{\rm inv}',
  \Y_\vP\mH_{\rm inv}')$ for $\mA^{-1}$ in time $O(\alpha
  \sM(m))$. This proves the last claim.
\end{proof}

\subsection{Reduction to the Hankel case} 

Our second reduction is less straightforward: we use Pan's idea of
multiplicative transformation of operators~\cite{Pan90} to reduce
{\it basic} operators to an operator of Hankel type. There is some
flexibility in the choice of the target operator, here
the Sylvester operator $\nabla_{\ZZ_{m,0},\ZZ_{n,1}^t}$; the only (natural) requirement is
that this target operator remains invertible.

The following proposition summarizes the transformation
process. Although the formulas are long, they describe simple
processes: for an operation such as multiplication, inversion or system solving, 
this amounts to e.g.~the analogue operation for an operator of Hankel type,
several products with simple matrices derived from $\vP$ and~$\vQ$,
and $O(1)$ matrix-vector products with the input matrix or its
transpose.

\begin{proposition}\label{prop:redHank}
  Let $\mcL \in \{\nabla_{\CC_\vP,\CC_\vQ^t}, \Delta_{\CC_\vP,\CC_\vQ^t}\}$,
  and suppose that $\mathscr{H}_\vP$ and $\mathscr{H}_\vQ$ hold.
  Then
  \begin{align}
\sMM(\mcL, \alpha, \beta) & \le 
    \sMM(\nabla_{\ZZ_{m,0},\ZZ_{n,1}^t}, \alpha+2, \beta) 
    +O\big (\sD(\vP,\vQ) + (\alpha+\beta)\sC(\vP,\vQ)\big ),\label{prop:i}\\[1mm]
\sMS(\mcL, \alpha) & \le
    \sMS(\nabla_{\ZZ_{m,0},\ZZ_{n,1}^t}, \alpha+2)
    +O\big(\sD(\vP,\vQ) + \alpha\sC(\vP,\vQ) \big ),\label{prop:ii}\\[1mm]
\sMI(\mcL, \alpha) & \le
    \sMI(\nabla_{\ZZ_{m,0},\ZZ_{m,1}^t}, \alpha+2)
    +O\big(\sD(\vP,\vQ) + \alpha\sC(\vP,\vQ) + \alpha^{\omega-1} m\big )\label{prop:iii}.
  \end{align}
\end{proposition}

The proof of Proposition~\ref{prop:redHank} will occupy the rest of
this subsection.  Before that, though, we point out that
Theorem~\ref{theo:equiv} follows directly from this proposition and
Lemma~\ref{lemma:redBasic}.  In particular, since $p \le \sM(p) \le
\sM(m)+\sM(n) \le \sC(\vP,\vQ)$ for $p = \max(m,n)$, overheads such as
$O(\alpha \sM(m))$ or $O(\alpha \sM(p))$ that appear in that lemma are
absorbed into terms such as $O(\alpha \sC(\vP,\vQ))$ that appear in
the proposition.

\paragraph{Preliminaries}
To establish each of the claims in the proposition above, it will be
useful to have simple matrices $\mL$ and $\mR$ for
pre-/post-multiplying $\mA$ and whose displacement rank with respect
to $\nabla_{\ZZ_{m,0},\CC_\vP}$ and $\nabla_{\CC_\vQ^t,\ZZ_{n,1}^t}$,
respectively, is small.  Lemma~\ref{lemma:LR} below shows that a
possible choice, leading to displacement ranks at most $1$, is
\[
\mL = \J_m  \, \W_\vP^t \, \Y_\vP^{-1} \quad \text{and} \quad \mR = \Y_\vQ^{-1}\, \W_\vQ\,\J_n.
\]
In order to define generators for such matrices, we first rewrite the
products $P = P_1\cdots P_d$ and $Q=Q_1\cdots Q_e$ as $P =
p_0+\cdots+p_{m-1}x^{m-1}+x^m$ and $Q=q_0+\cdots+q_{n-1}x^{n-1}+x^n$,
and let $\m = [p_0~\cdots~p_{m-1}]^t$ and $\n =
[q_0~\cdots~q_{n-1}]^t$ be the coefficient vectors of the polynomials
$P-x^m$ and $Q-x^n$.  Then, we let
\[
\t = \left [ \begin{matrix} 1 \\ 0 \\ \vdots \\ 0 
  \end{matrix} \right ]\in \F^m, \quad
\u = \Y^{-1}_\vP\,\W_\vP\,\m \in \F^m, \quad
\s = \left [ \begin{matrix} 1 \\ 0 \\ \vdots \\ 0 \end{matrix} \right ]\in \F^n, \quad
\r = -\Y_\vQ^{-1}\,\W_\vQ\, (\n + \s) \in \F^n.
\]

\begin{lemma}\label{lemma:LR}
  The relations $\nabla_{\ZZ_{m,0},\CC_\vP}(\mL) = \t\, \u^t$ and
  $\nabla_{\CC_\vQ^t,\ZZ_{n,1}^t}(\mR) = \r\, \s^t$ hold.
\end{lemma}
\begin{proof}
According to the second part of Lemma~\ref{lemma:Krylov},
$\Y_\vP^{-1}\,\W_\vP$ is equal to the Krylov matrix $\K(\CC_\vP^t, \w, m)$,
so that the columns of $\mL^t = \Y_\vP^{-1}\,\W_\vP\,\J_m$ are
$(\CC_\vP^t)^{m-1}\w,\dots,$ $ \CC_\vP^t \w, \w$. 
Hence only the first column of $\CC_\vP^t \,\mL^t - \mL^t \,\ZZ_{m,0}^t$ is nonzero,
and it is equal to $\tilde \u := (\CC_\vP^t)^m\,\w$.
After transposition, this leads to $\nabla_{\ZZ_{m,0},\CC_\vP}(\mL) = -\t\, {\tilde\u}^t$.
We can now check that $\tilde \u = -\u$ as follows.
First, from Lemma~\ref{lemma:transp2} and since $\Y_\vP$ is invertible,
we see that $\tilde\u = \Y_\vP^{-1}\,\CC_\vP^m\,\Y_\vP\,\w$.
Then, the shape of $\w$ implies that $\Y_\vP\,\w$ is the vector $\v$
defined in the first part of Lemma~\ref{lemma:Krylov},
so that $\tilde\u = \Y_\vP^{-1}\,\CC_\vP^m\,\v$.
Now, the special shape of $\v$ implies that $\CC_\vP^m\,\v$ 
is the vector whose $i$th subvector of length $m_i$ contains the coefficients of
$x^m\bmod P_i = - (P-x^m) \bmod P_i$.
In other words, $\CC_\vP^m\,\v = -\W_\vP\,\m$.
This shows that $\tilde\u =  -\Y_\vP^{-1}\,\W_\vP\,\m= -\u$, 
which concludes the proof of the first relation.

For the second relation, the proof is the same, taking into account
that we are now considering the operator
$\nabla_{\CC_\vQ^t,\ZZ_{n,1}^t}$
and that $\ZZ_{n,1}^t = \ZZ_{n,0}^t + \J_n \,\s\, \s^t$.
\end{proof}

\paragraph{Proof of~\eqref{prop:i} for $\mcL=\nabla_{\CC_\vP,\CC_\vQ^t}$}
Let $\mA$ be in $\F^{m\times n}$ and let $(\mG,\mH)$ be a
$\nabla_{\CC_\vP,\CC_\vQ^t}$-generator for $\mA$. With
$\mL,\mR,\r,\s,\t,\u$ as in the previous paragraph, define further 
$${\mG'} =\big [\, \t\ |\ \mL\,\mG\ |\ \mL\,\mA\,\r\, \big]
\quad\text{and}\quad 
 {\mH'} =\big [\, \mR^t \,\mA^t\,\u\ |\ \mR^t\,\mH\ |\ \s\, \big].$$
\begin{lemma}\label{lemma:multPan}
  The matrix $\mA'=\mL\,\mA\,\mR$ satisfies $\nabla_{\ZZ_{m,0},\ZZ_{n,1}^t}(\mA') = {\mG'} {\mH'}^t$.
\end{lemma}
\begin{proof}
Applying the general formula 
\begin{equation} \label{eq:general-formula-ABC}
\nabla_{\mM,\mN}(\mA\mB\mC) =
\nabla_{\mM,\mQ}(\mA)\mB\mC + \mA \nabla_{\mQ,\mR}(\mB)\mC +
\mA\mB\nabla_{\mR,\mN}(\mC)
\end{equation}
(which follows directly
from~\cite[Theorem~1.5.4]{Pan01}) and then using Lemma~\ref{lemma:LR},
we obtain
\begin{align*}
\nabla_{\ZZ_{m,0},\ZZ_{n,1}^t}(\mL\,\mA\,\mR) &=
\nabla_{\ZZ_{m,0},\CC_\vP}(\mL)\,\mA\,\mR +
\mL\,\nabla_{\CC_\vP,\CC_\vQ^t}(\mA)\,\mR +
\mL\,\mA\,\nabla_{\CC_\vQ^t,\ZZ_{n,1}^t}(\mR)\\
&= \t\, \u^t\,\mA\,\mR + \mL\,\mG\,\mH^t\,\mR + \mL\,\mA\, \r\, \s^t,
\end{align*}
which is the announced equality.
\end{proof}

\medskip

To compute a product of the form $\mA\mB$ with
$\mB\in\F^{n\times\beta}$, we first compute the matrices $\mG'$ and
$\mH'$ described above.  To obtain $\mG'$, it suffices to set up the
vectors $\r$ and $\mA\r$ and to compute $\alpha+1$ matrix-vector
products by~$\mL$.  Given $(\mG,\mH)$ and $\r$,
Corollary~\ref{coro:mul1} implies that $\mA\r$ is obtained in time
\[ 
O(\sD(\vP,\vQ) + \alpha \sC(\vP,\vQ)).\] On the other hand, it follows
from Lemmas~\ref{lemma:W} and~\ref{lemma:Y2} that 
the vector $\r$ is obtained in time $O(\sC(\vQ) + \sD(\vQ))$
and, after some
precomputation of time $O(\sC(\vP) + \sD(\vP))$, that
the $\alpha+1$ products by $\mL$ are obtained in time $O(\alpha
\sC(\vP))$.  Thus, overall, $\mG'$ is obtained in time $O(\sD(\vP,\vQ)
+ \alpha \sC(\vP,\vQ))$.

We proceed similarly for $\mH'$. The only differences are that we have
to do matrix-vector products involving $\mR^t$ and $\mA^t$ instead of
$\mL$ and $\mA$. For the former, Lemmas~\ref{lemma:W}
and~\ref{lemma:Y2} show that after a precomputation of cost
$O(\sC(\vQ) + \sD(\vQ))$, the cost of one such product is
$O(\sC(\vQ))$.  For the latter, recall that, as pointed out in the
introduction, from the given $\nabla_{\CC_\vP,\CC_\vQ^t}$-generator of
$\mA$, we can deduce in negligible time a
$\nabla_{\CC_\vQ,\CC_\vP^t}$-generator of $\mA^t$ of the same length
$\alpha$.  This allows us to do matrix-vector products with $\mA^t$
with the same asymptotic cost 
$$O(\sD(\vQ,\vP) + \alpha \sC(\vQ,\vP)) =O(\sD(\vP,\vQ) + \alpha
\sC(\vP,\vQ))$$ as for $\mA$. Thus, the overall cost for
computing $\mH'$ is asymptotically the same as for $\mG'$.

Let us now bound the cost of deducing the product $\mA\mB$ from
$\mG',\mH',\mB$.  Note first that under the coprimality assumptions
$\mathscr{H}_\vP$ and $\mathscr{H}_\vQ$ the matrices $\W_\vP$ and
$\W_\vQ$ are invertible, and so are $\mL$ and $\mR$.  Consequently,
$\mA\mB = \mL^{-1} \mA' \,\mR^{-1} \mB$ and it suffices to bound the
cost of each of the three products $\mB' := \mR^{-1}\mB$, $\mB'' :=
\mA' \mB'$, and $\mL^{-1} \mB''$.  By reusing the same precomputation
as for $\mH'$ and applying again Lemmas~\ref{lemma:W}
and~\ref{lemma:Y2}, we obtain $\mB'$ for an additional cost of
$O(\beta \sC(\vQ))$, via $\beta$ matrix-vector products by~$\mR^{-1}$.
Then, Lemma~\ref{lemma:multPan} says that $\mA'$ is a Hankel-like
matrix for which a $\nabla_{\ZZ_{m,0},\ZZ_{n,1}^t}$-generator of length
$\alpha+2$ is $(\mG',\mH')$.  Hence the cost for deducing $\mB''$
from $\mG',\mH',\mB'$ is at most $\sMM(\nabla_{\ZZ_{m,0},\ZZ_{n,1}^t},
\alpha+2, \beta)$.  Finally, $\mL^{-1} \mB''$ is obtained for an
additional cost of $O(\beta \sC(\vP))$, by reusing the same
precomputation as for $\mG'$ and by performing $\beta$ matrix-vector
products by $\mL^{-1}$.

To summarize, we have shown that $(\mG',\mH')$ can be obtained in time
$O(\sD(\vP,\vQ) + \alpha \sC(\vP,\vQ))$ and that $\mA\mB$ can be
deduced from $\mG',\mH',\mB$ in time
$\sMM(\nabla_{\ZZ_{m,0},\ZZ_{n,1}^t}, \alpha+2, \beta) + O(\beta
\sC(\vP,\vQ))$.  Adding these two costs thus proves the bound~\eqref{prop:i}
in Proposition~\ref{prop:redHank} for  $\mcL=\nabla_{\CC_\vP,\CC_\vQ^t}$.

\paragraph{Proof of~\eqref{prop:ii} for $\mcL=\nabla_{\CC_\vP,\CC_\vQ^t}$}
For $\mA \in \F^{m\times n}$ given by
some $\nabla_{\CC_\vP,\CC_\vQ^t}$-generator $(\mG,\mH)$ of length
$\alpha$ and given $\b$ in $\F^m$, we now want to find a (nontrivial) solution of $\mA \x
= \b$, or determine that no solution exists. 

Define $\b'=\mL \b$. Because $\mL$ and $\mR$ are invertible matrices,
solving the system $\mA' \x'= \b'$ is equivalent to solving $\mA \x
=\b$, with then $\x= \mR \x'$. As in the previous paragraph, we can
compute a generator $(\mG',\mH')$ of $\mA'$ for the operator
$\nabla_{\ZZ_{m,0},\ZZ_{n,1}^t}$ in time $O(\sD(\vP,\vQ) + \alpha
\sC(\vP,\vQ))$. The cost of computing $\b'$ from $\b$ (and $\x$ from
$\x'$) fits into the same bound, and solving the new system $\mA' \x'=
\b'$ takes time $\sMS(\nabla_{\ZZ_{m,0},\ZZ_{n,1}^t}, \alpha+2)$.
Summing these costs, we prove the second item in the proposition.

\paragraph{Proof of~\eqref{prop:iii} for $\mcL=\nabla_{\CC_\vP,\CC_\vQ^t}$}
Assume now that $m=n$ and that $\mA \in \F^{m\times m}$ is given by
a $\nabla_{\CC_\vP,\CC_\vQ^t}$-generator $(\mG,\mH)$ of length
$\alpha$.  As before, $\mL$ and $\mR$ are invertible by assumption, so
that $\mA$ is invertible if and only if $\mA' = \mL \, \mA \, \mR$ is
invertible.  By Lemma~\ref{lemma:multPan} this matrix $\mA'$ has
displacement rank at most $\alpha+2$ for
$\nabla_{\ZZ_{m,0},\ZZ_{m,1}^t}$.  Thus, as explained in the
introduction, if $\mA'$ is invertible then its inverse has
displacement rank at most $\alpha+2$ for
$\nabla_{\ZZ_{m,1}^t,\ZZ_{m,0}}$.  The next lemma shows that if
$(\mG_{\rm inv}', \mH_{\rm inv}')$ is a 
$\nabla_{\ZZ_{m,1}^t,\ZZ_{m,0}}$-generator for the inverse of $\mA'$,
then a $\nabla_{\CC_\vQ^t,\CC_\vP}$-generator for the inverse of
$\mA$ is given by the matrices
\[
\mG_{\rm inv} = \big [\, \r\ |\ \mR\,\mG_{\rm inv}'\ |\ \mR\,{\mA'}^{-1}\,\t\ \big]
\quad\text{and}\quad 
\mH_{\rm inv} = \big [\ \mL^t \,{\mA'}^{-t}\,\s\ |\ \mL^t\,\mH_{\rm inv}'\ |\ \u\ \big].
\]

\begin{lemma} \label{lemma:generation_of_the_inverse_of_A}
  The matrix $\mA^{-1}$ satisfies $\nabla_{\CC_\vQ^t,\CC_\vP}(\mA^{-1}) = \mG_{\rm inv} \mH_{\rm inv}^t$.
\end{lemma}
\begin{proof}
As for Lemma~\ref{lemma:multPan} the proof 
follows from~\cite[Theorem~1.5.4]{Pan01} and Lemma~\ref{lemma:LR}:
\begin{align*}
\nabla_{\CC_\vQ^t,\CC_\vP}(\mR\,{\mA'}^{-1}\,\mL) &=
\nabla_{\CC_\vQ^t,\ZZ_{m,1}^t}(\mR)\,{\mA'}^{-1}\,\mL +
\mR\,\nabla_{\ZZ_{m,1}^t,\ZZ_{m,0}}({\mA'}^{-1})\,\mL \\ 
& \hspace*{6cm} + \mR\,{\mA'}^{-1}\,\nabla_{\ZZ_{m,0},\CC_\vP}(\mL)\\
&= \r\, \s^t\,{\mA'}^{-1}\,\mL + \mR\,\mG_{\rm inv}'\,(\mH_{\rm inv}')^t\,\mL + \mR\,{\mA'}^{-1}\, \t\, \u^t.
\end{align*}
To conclude, we remark that $\mA^{-1}=\mR\,{\mA'}^{-1}\,\mL$.
\end{proof}

\medskip

The bound in~\eqref{prop:iii} can now be established as follows.  We
begin by computing $(\mG',\mH')$, which by Lemma~\ref{lemma:multPan}
is a $\nabla_{\ZZ_{m,0},\ZZ_{m,1}^t}$-generator of length $\alpha+2$
of~$\mA'$; as shown in the previous paragraph, this is done in time
$O(\sD(\vP,\vQ) + \alpha \sC(\vP,\vQ))$.  Then, in time
$\sMI(\nabla_{\ZZ_{m,0},\ZZ_{m,1}^t}, \alpha+2)$ we either conclude
that $\mA'$ is singular, or produce a
$\nabla_{\ZZ_{m,1}^t,\ZZ_{m,0}}$-generator $(\mG_{\rm inv}',\mH_{\rm
  inv}')$ of length $\alpha+2$ for the inverse of $\mA'$.
Finally, if $\mA'$ is invertible we proceed in two steps: First, using
the same amount of time as for $\mG'$ and $\mH'$, we set up the
matrices $\mG_{\rm inv}$ and $\mH_{\rm inv}$ introduced before
Lemma~\ref{lemma:generation_of_the_inverse_of_A}.  Then we reduce each
of these two matrices to arrive at a generator of length
$\alpha$; this generator compression step can be done in time
$O(\alpha^{\omega-1} m)$, in view of~\cite[Remark~4.6.7]{Pan01}.
Overall, these costs add up to the result reported in~\eqref{prop:iii}.

\paragraph{Proof of~\eqref{prop:i}--~\eqref{prop:iii} for $\mcL=\Delta_{\CC_\vP,\CC_\vQ^t}$}
We only sketch the proof in these cases. 
This time, we rely on the (easily verified) general formula
$$\Delta_{\mM,\mN}(\mA\mB\mC) 
= -\nabla_{\mM,\mP}(\mA)\mB\mQ\mC + \mA \Delta_{\mP,\mQ}(\mB)\mC + \mM\mA\mB\nabla_{\mQ,\mN}(\mC).$$
We use it to write
$$\Delta_{\ZZ_{m,0},\ZZ_{n,1}^t}(\mL\,\mA\,\mR) =
- \nabla_{\ZZ_{m,0},\CC_\vP}(\mL)\,\mA\,\CC_\vQ^t\,\mR
+\mL\, \Delta_{\CC_\vP,\CC_\vQ^t}(\mA)\, \mR
+ \ZZ_{m,0}\,\mL\,\mA\nabla_{\CC_\vQ^t,\ZZ_{n,1}^t}(\mR);
$$
this allows us to reduce questions (multiplication, linear system solving, inversion) related to $\mA$ to
the same questions for $\mA' = \mL\,\mA\,\mR$, where $\mA'$ is 
given through a $\Delta_{\ZZ_{m,0},\ZZ_{n,1}^t}$-generator of length $\alpha+2$.
Then, the relations $\ZZ_{n,1}^t \J_n = \J_n \ZZ_{n,1}$ and 
$\ZZ_{n,1}^t = \ZZ_{n,1}^{-1}$ imply 
$$\nabla_{\ZZ_{m,0},\ZZ_{n,1}^t}(\mA'\,\J_n) = - \Delta_{\ZZ_{m,0},\ZZ_{n,1}^t}(\mA')\, \ZZ_{n,1}\, \J_n;$$
this allows us to reduce our problems to computations with the matrix
$\mA'\,\J_n$ given by means of a $\nabla_{\ZZ_{m,0},\ZZ_{n,1}^t}$-generator of length $\alpha+2$.

For inversion (where $m=n$),
inverting $\mA'\,\J_n$ leads to a $\nabla_{\ZZ_{m,1}^t,\ZZ_{m,0}}$-generator
of $(\mA'\,\J_m)^{-1}=\J_m {\mA'}^{-1}$. 
Then, using the relations on $\ZZ_{m,1}$ given above,
we obtain
$$\Delta_{\ZZ_{m,1}^t, \ZZ_{m,0}}( {\mA'}^{-1} ) = 
 \ZZ_{m,1}^t\, \J_m \, \nabla_{\ZZ_{m,1}^t, \ZZ_{m,0}} (\J_m {\mA'}^{-1})$$
and thus a $\Delta_{\ZZ_{m,1}^t, \ZZ_{m,0}}$-generator of ${\mA'}^{-1}$. 
The general formula above further leads to
\begin{align*}
\Delta_{\CC_\vQ^t,\CC_\vP}(\mR\,{\mA'}^{-1}\,\mL) &= -
\nabla_{\CC_\vQ^t,\ZZ_{m,1}^t}(\mR)\,{\mA'}^{-1}\,\ZZ_{m,0}\,\mL
+\mR\, \Delta_{\ZZ_{m,1}^t,\ZZ_{m,0}}({\mA'}^{-1})\, \mL 
\\ &
\hspace*{6cm} + \CC_\vQ^t\,\mR\,{\mA'}^{-1} \nabla_{\ZZ_{m,0},\CC_\vP}(\mL);
\end{align*}
since $\mR\,{\mA'}^{-1}\,\mL = \mA^{-1}$, this gives
a $\Delta_{\CC_\vQ^t,\CC_\vP}$-generator of $\mA^{-1}$ whose length 
$\alpha+2$ is finally reduced to $\alpha$.

In all cases Lemma~\ref{lemma:LR} can be reused, and the requested
complexity bounds are then derived in the same way as for the three
cases above, up to a minor difference: we also need to take into
account the cost of multiplication by $\CC_\vQ$ or its transpose with
a vector. However, due to the sparse nature of this matrix, this cost
is easily seen to be $O(n)$, so it does not impact the asymptotic
estimate.

\section{Multiplication algorithms} \label{sec:multiplication_alogrithms}

In this section, we prove bounds on the cost of the product
$\mA\mB$, where $\mA$ is a structured matrix in $\F^{m\times n}$ and
$\mB$ is an arbitrary matrix in $\F^{n\times \beta}$; this will prove
the claims in Theorems~\ref{theo:mainSylv} and~\ref{theo:mainSylv2}
regarding multiplication.

\begin{theorem}\label{theo:mul}
  For any invertible operator $\mcL$ associated with $(\vP,\vQ)$, we can take
  $$\sMM(\mcL, \alpha, \beta)=O\left
  (\frac{\beta'}{\alpha'}\sMMa'\left(\frac{p}{\alpha'},\alpha'\right)
  + \sD(\vP,\vQ)+(\alpha+\beta)\sC(\vP,\vQ) \right),$$
  with $p=\max(m,n)$,  $\alpha'=\min(\alpha,\beta)$, and $\beta'=\max(\alpha,\beta)$.  
\end{theorem}
Using Lemma~\ref{lemma:redBasic}, it will be sufficient to consider
multiplication for the basic operators associated with
$(\vP,\vQ)$. Thus, throughout this section we let $\mcL:\F^{m\times n}
\to \F^{m\times n}$ be one of the two operators
$ \nabla_{\CC_\vP,\CC_\vQ^t}$ and $\Delta_{\CC_\vP,\CC_\vQ^t}$
and,
given $(\mG,\mH)\in \F^{m\times \alpha}\times \F^{n\times \alpha}$ and
$\mB \in \F^{n\times \beta}$, we show how to compute the product $\mA
\mB \in \F^{m\times \beta}$ with $\mA$ the
unique matrix in $\F^{m\times n}$ such that $\mcL(\mA)=\mG \mH^t$.

In Subsection~\ref{ssec:mulSS} we deal with the cases covered by Theorem~\ref{theo:rec}; 
a key step of the proof (corresponding to the case $Q = x^n$) 
is deferred to Subsection~\ref{section:R} due to its length. 
We then extend our results to all remaining cases in Subsection~\ref{ssec:Rgen}.

In view of Theorem~\ref{theo:equiv}, it would actually be sufficient
to assume that $\mcL$ is the Hankel operator $\nabla_{\ZZ_{m,0},\ZZ_{n,1}^t}$.
This would allow us to bypass most of the derivation in
Subsection~\ref{ssec:Rgen}; however, such a restriction does not seem
to lead to a better cost estimate than our direct approach, which we
include for completeness.

\subsection{Applying Theorem~\ref{theo:rec}}\label{ssec:mulSS}

Let us first assume that $\mcL=\nabla_{\CC_\vP,\CC_\vQ^t}$. Given
$\vP$, $\vQ$ and the generator $(\mG,\mH)$, and a matrix $\mB \in
\F^{n \times \beta}$, Theorem~\ref{theo:rec} shows that $\mA \mB$ can
be computed as follows:
\begin{enumerate}
\item\label{s0} compute the polynomials $P$, $Q$, $\gamma_k$, $\eta_k$, $Q^{-1}\bmod P_i$
 for $k\le\alpha$ and $i\le d$;
\item\label{s1} compute the $n \times \beta$ matrix $\mB' = \XX_\vQ\,\Y_\vQ \, \mB$;
\item\label{s2} compute the $m \times \beta$ matrix $\mC$ given by
  $$\mC=\sum_{k \le \alpha} \CC_{\gamma_k,P,n}\, \CC_{\eta_k, Q}\, \mB';$$
\item\label{s3} return the $m \times \beta$ matrix $\V_{\vP,\vQ}\,\W_\vP\,\mC$. 
\end{enumerate}
Proceeding as in the proof of Corollary~\ref{coro:mul1}, we see that
Steps~\ref{s0}, \ref{s1}, \ref{s3} can be completed in time
$O(\sD(\vP,\vQ)+(\alpha+\beta)\sC(\vP,\vQ))$. We are thus left with analyzing
Step~\ref{s2}.

Let $\b_1,\dots,\b_\beta$ be the columns of $\mB'$, and for
$i=1,\dots,\beta$ let $B_i=\pol(\b_i) \in \F[x]_n$ be the polynomial
whose coefficients are the entries of $\b_i$. In polynomial terms,
given $B_1,\dots,B_\beta$ in $\F[x]_n$, we want to compute
$C_1,\dots,C_\beta$ in $\F[x]_m$ such that
$$C_i=\sum_{k \le \alpha} \gamma_k (\eta_k B_i \bmod Q) \bmod P$$
for $i \le \beta$;
then, the $m$ coefficients of $C_i$ will make up the $i$th column
of the matrix $\mC$.

To compute $C_i$, we can compute the sum first, and then reduce it modulo
$P$. The core problem is thus to compute
$$R_i=\sum_{k \le \alpha} \gamma_k (\eta_k B_i \bmod Q) $$ 
in $\F[x]_{m+n-1}$
for all $i \le \beta$; the cost is given in Theorem~\ref{theo:Rgen}
of Subsection~\ref{ssec:Rgen} below:
writing $\alpha'=\min(\alpha,\beta)$ and $\beta'=\max(\alpha,\beta)$,
one can compute $R_1,\dots,R_\beta$ in time
\begin{equation} \label{eq:cos-of-thm31}
O\left(\frac{\beta'}{\alpha'}\sMMa'\left(\frac{p}{\alpha'},\alpha'\right)\right).
\end{equation}
Finally, 
since $R_i$ has degree less than $m+n-1$ and since $P$ has degree $m$,
computing each $C_i = R_i \bmod P$ takes time $O(\sM(m)+\sM(n)) \subset O(\sM(p))$,
so we can deduce $C_1,\ldots,C_\beta$ from $R_1,\ldots,R_\beta$ 
for a negligible total cost of $O(\alpha \sM(p))$.

In the case $\mcL=\Delta_{\CC_\vP,\CC_\vQ^t}$, the approach is almost
the same. A minor difference is that we replace $\V_{\vP,\vQ}$ by
$\V'_{\vP,\vQ}$: this has no influence on the cost. The other
difference is the matrix~$\J_n$ appearing in the inversion formula for
$\Delta_{\CC_\vP,\CC_\vQ^t}$ (second part of Theorem~\ref{theo:rec}); in polynomial terms, this now leads us
to compute the sums
$$R'_i=\sum_{k \le \alpha} \gamma_k \,\rev(\eta_k B_i \bmod Q, n-1)$$
for $i=1,\ldots,\beta$, and then to reduce each of them modulo $P$.

This problem is actually an instance of the question treated 
above, for we have
$$\rev(R'_i, m+n-2) = \sum_{k \le \alpha} \rev(\gamma_k,m-1) (\eta_k
B_i \bmod Q).$$ Since knowing $\rev(R'_i, m+n-2)$ gives us $R'_i$ for
free, we can appeal here as well to Theorem~\ref{theo:Rgen}
to compute $R'_i$, using $\rev(\gamma_k,m-1)$ instead of $\gamma_k$.
The cost estimate is the same as before, namely,~(\ref{eq:cos-of-thm31}),
thus leading to Theorem~\ref{theo:mul}
regarding a product of the form $\mA \mB$ for the operator $\mcL=\Delta_{\CC_\vP,\CC_\vQ^t}$.

\subsection{Computing the $R_i$'s when $Q=x^n$}\label{section:R}

The previous subsection has shown that for both Sylvester's and
Stein's displacement operators, computing the product $\mA\mB$ essentially
reduces to the following basic problem: given as input
\begin{itemize}
\item a monic polynomial $Q$ of degree $n$,
\item polynomials
  $U_1,\dots,U_\alpha$ in $\F[x]_m$
  and
  $V_1,\dots,V_\alpha, W_1,\dots,W_\beta$ in $\F[x]_n$, 
\end{itemize}
compute $$R_i=\sum_{k \le \alpha} U_k (V_k W_i \bmod
Q) \qquad \mbox{for $i=1,\dots,\beta$.}$$ 
The naive algorithm computes all sums
independently; writing $p=\max(m,n)$, its cost is $O(\alpha\beta\sM(p))$.
In this subsection we establish the following improved complexity estimate 
in the special case where $Q=x^n$.
(This result will be further extended to the case of a general $Q$ of degree $n$ in Subsection~\ref{ssec:Rgen}.)

\begin{theorem}\label{theo:R}
  Assume $Q=x^n$ and $\alpha \le n$, and let $p = \max(m,n)$, 
  $\alpha'=\min(\alpha,\beta)$, and $\beta'=\max(\alpha,\beta)$.  Then
  one can compute the polynomials $R_1,\dots,R_\beta$ defined above in time
  $$O\left(\frac{\beta'}{\alpha'}\sMMa'\left(\frac{p}{\alpha'},\alpha'\right)\right).$$
\end{theorem}
To prove Theorem~\ref{theo:R}, we first rephrase our problem in polynomial matrix terms. 
Let $$\mU \in \F[x]_m^{\alpha\times 1},\quad \mV \in \F[x]_n^{\alpha\times
  1},\quad \mW \in \F[x]_n^{\beta\times 1}$$ be the polynomial vectors
with respective entries $(U_i)$, $(V_i)$ and $(W_i)$. Then, our
problem amounts to computing the polynomial row vector $\mR \in
\F[x]_{m+n-1}^{1\times\beta}$ such that 
\[ 
\mR = \mU^t(\mV\mW^t \bmod x^n). 
\]
We first solve that problem in the special case where $\alpha=\beta$
(this is where most of the difficulty takes place); then, we reduce
the other cases of arbitrary $\alpha,\beta$ to this case.

\medskip

\paragraph{Case $\alpha=\beta$}
The following lemma
is the basis of our algorithm in this case; it is taken from the proof
of~\cite[Lemma~12]{BoJeSc08}.
\begin{lemma} \label{lem:recursive-formula-for-R}
Let $\alpha,\gamma,\nu$ be in $\NN_{>0}$ with $\nu$ even. 
Let $\mV,\mW$ be in $\F[x]_\nu^{\alpha\times\gamma}$ and define
\[
\mV_0 = \mV \bmod x^{\nu/2},\quad \mV_1 = \mV \bdiv x^{\nu/2},
\]
\[
\mW_0 = \mW \bmod x^{\nu/2},\quad \mW_1 = \mW \bdiv x^{\nu/2}.
\]
Then the matrices $[\mV_0 \,\, \mV_1]$ and $[\mW_1 \,\, \mW_0]$ are
in $\F[x]_{\nu/2}^{\alpha \times 2\gamma}$, and we have
\[
\mV\mW^t \bmod x^\nu = \mV_0\mW_0^t + x^{\nu/2} \big([\mV_0 \,\, \mV_1][\mW_1 \,\, \mW_0]^t \bmod x^{\nu/2}\big).
\]
\end{lemma}
We will want to apply this lemma recursively, starting from $\nu=n$
and $\gamma=1$.  It will be convenient to have both $n$ and $\alpha$
be powers of two.
\begin{itemize}
\item If $n$ is not a power of two, we define ${\overline n} =
  2^{\lceil \log n \rceil}$ and $\delta = {\overline n}-n$. One may then
  check that $a \bmod b = x^{-\delta} \big((x^\delta a) \bmod
  (x^\delta b)\big)$ for any $a$ and nonzero $b$ in $\F[x]$; applying
  this identity componentwise to the definition of $\mR$ gives
  \begin{equation} \label{eq:R-via-scaling}
    \mR = x^{-\delta}\, \mU^t \big( \mV (x^\delta\,\mW)^t \bmod x^{\overline n}\big),
  \end{equation}
  where $\mV$ and $x^{\delta}\mW$ have degree less than $\overline n$. 
\item If $\alpha$ is not a power of two, we define $\overline \alpha =
  2^{\lceil \log \alpha \rceil}$ and $\mu =\overline \alpha-\alpha$. Then,
  we can introduce $\mu$ dummy polynomials $(U_i)$, $(V_i)$ and
  $(W_i)$, all equal to zero, without affecting the value of
  $R_1,\dots,R_\alpha$.  
\end{itemize}
The resulting algorithm is given in Figure~\ref{fig1}, with a
top-level procedure {\sf mul} that handles the cases when $n$ or
$\alpha$ are not powers of two by defining $\overline n$ and
$\overline \alpha$ as above, and a main recursive procedure {\sf
  mul\_rec}.  We initially set $\nu=\overline n$ and $\gamma=1$;
through each recursive call, the width $\gamma$ of the matrices
$\mV$ and $\mW$ is doubled, while their degree $\nu$ is divided by
two, so that the invariant $\gamma \nu =\overline n$ is maintained.

\begin{figure}[htbp]
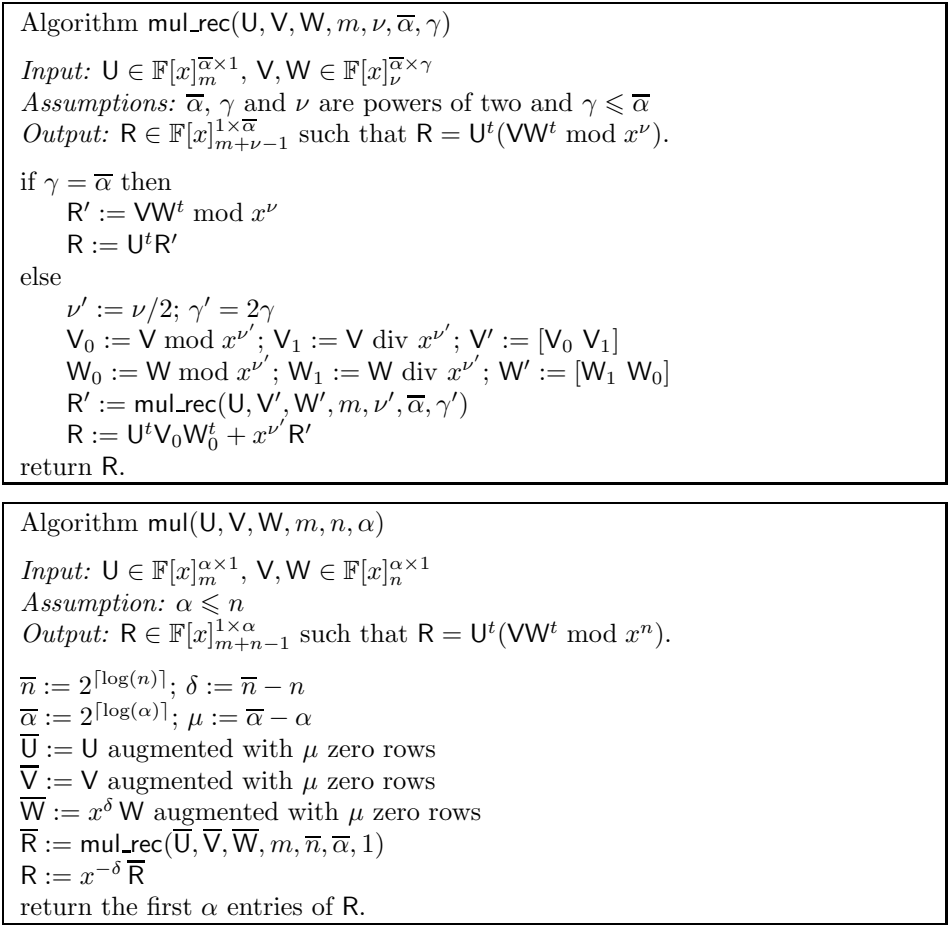

\begin{center}
\fbox{
\begin{minipage}{12.2cm}
Algorithm {\sf mul\_rec}($\mU,\mV,\mW,m,\nu,\overline \alpha,\gamma$)

\medskip

{\it Input:} $\mU\in\F[x]_m^{\overline \alpha\times 1}$,
             $\mV,\mW\in\F[x]_\nu^{\overline \alpha\times \gamma}$ \\
{\it Assumptions:} $\overline \alpha$, $\gamma$ and $\nu$ are powers of two and
                   $\gamma\le \overline \alpha$\\
{\it Output:} $\mR\in\F[x]_{m+\nu-1}^{1\times \overline \alpha}$ such that $\mR = \mU^t (\mV \mW^t \bmod x^\nu)$.

\medskip

if $\gamma = \overline \alpha$ then \\
\hspace*{0.5cm} $\mR' := \mV \mW^t \bmod x^\nu$ \\
\hspace*{0.5cm} $\mR := \mU^t \mR'$ \\
else\\
\hspace*{0.5cm} $\nu':=\nu/2$; $\gamma'=2\gamma$\\
\hspace*{0.5cm} $\mV_0 := \mV \bmod x^{\nu'};\, \mV_1 := \mV \bdiv x^{\nu'};\,\mV' := [\mV_0 \,\,\mV_1]$\\
\hspace*{0.5cm} $\mW_0 := \mW \bmod x^{\nu'};\, \mW_1 := \mW \bdiv x^{\nu'};\,\mW' := [\mW_1 \,\,\mW_0]$ \\
\hspace*{0.5cm} $\mR' := {\sf mul\_rec}(\mU, \mV', \mW',m,\nu',\overline \alpha,\gamma')$\\
\hspace*{0.5cm} $\mR := \mU^t\mV_0\mW_0^t + x^{\nu'}\mR' $\\
return $\mR$.
\end{minipage}
}
\end{center}

\medskip

\begin{center}
\fbox{
\begin{minipage}{12.2cm}
Algorithm {\sf mul}($\mU,\mV,\mW,m,n,\alpha$)
\medskip

{\it Input:} $\mU\in\F[x]_m^{\alpha\times 1}$,
              $\mV,\mW\in\F[x]_n^{\alpha\times 1}$ \\
{\it Assumption:} $\alpha \le n$ \\
{\it Output:} $\mR\in\F[x]_{m+n-1}^{1\times\alpha}$ such that $\mR = \mU^t (\mV \mW^t \bmod x^n)$.

 \medskip

$\overline n := 2^{\lceil\log(n)\rceil}$; $\delta := \overline n -n$\\
$\overline \alpha := 2^{\lceil\log(\alpha)\rceil}$; $\mu := \overline \alpha -\alpha$\\
${\overline\mU} := \mU$ augmented with $\mu$ zero rows \\
${\overline\mV} := \mV$ augmented with $\mu$ zero rows \\
${\overline\mW} := x^\delta \, \mW$ augmented with $\mu$ zero rows \\
${\overline\mR} := {\sf mul\_rec}(\overline \mU,\overline\mV,{\overline\mW},m, \overline n , \overline \alpha, 1)$\\
$\mR := x^{-\delta} \, {\overline\mR}$ \\
return the first $\alpha$ entries of $\mR$.
\end{minipage}
}
\caption{Algorithms {\sf mul} and {\sf mul\_rec}.}\label{fig1}
\end{center}
\end{figure}

\begin{lemma} \label{prop:recursive-routine}
Let $p = \max(m, n)$. 
Algorithm {\sf mul} of Figure~\ref{fig1}
works correctly in time $$O\left(\sMMa'\left(\frac{p}{\alpha},\alpha\right)\right).$$
\end{lemma}
\begin{proof}
  Correctness is a direct consequence of
  Lemma~\ref{lem:recursive-formula-for-R} and the discussion that
  followed it.  For the cost analysis, we let $C(k)$ denote the cost
  of {\sf mul\_rec} called upon parameters $\overline \alpha$ and
  $\gamma$ such that $\overline \alpha = 2^k \gamma$. Note that the
  total cost of {\sf mul} will then be at most $C(\kappa)$, with $\kappa =
  \lceil \log \alpha \rceil$.  In particular, below, we always have
  $k \le \kappa$. Another useful remark is that we always have $\nu \le
  2p/\gamma$, since $\nu=\overline n/\gamma$ and $\overline n \le 2n \le 2p$.

  If $k=0$, we have $\gamma = \overline\alpha$, so we are in the base
  case of the algorithm, where $\mV$ and $\mW$ are in
  $\F[x]^{\overline \alpha\times \overline \alpha}_\nu$. We first
  compute $\mR'=\mV\mW^t \bmod x^\nu$ in time $\sMMa(\nu,\overline
  \alpha)$.  Then, given $\mU$ in $\F[x]_m^{\overline \alpha\times
    1}$ and $\mR'$ in
  $\F[x]_\nu^{\overline\alpha\times\overline\alpha}$, we obtain $\mR=
  \mU^t \mR'$ in two steps as follows: rewriting $\mU^t$ in the form
  $\mU^t = [1~x^c~x^{2c}~\cdots~x^{(\overline\alpha-1)c}]\mU'$ with $c =
  \lfloor m/\overline\alpha \rfloor$ and
  $\mU'\in\F[x]_c^{\overline\alpha\times\overline\alpha}$, we compute
  first $\mQ = \mU'\mR'$ in time
  $\sMMa(\max\{c,\nu\},\overline\alpha)$, and then deduce $\mR$ from
  $\mQ$ using at most $\overline\alpha^2 (c+\nu-1)$ additions in $\F$.  In summary,
  \begin{align*}
   C(0) &\le \sMMa(\nu,\overline\alpha) + \sMMa(\max\{c,\nu\},\overline\alpha) + \overline\alpha^2(c+\nu-1)   \\
   & \le 2 \sMMa \!\left( \frac{2p}{\overline\alpha}, \overline\alpha \right) + 3 \, \overline\alpha \, p,
  \end{align*}
  using $c \le m/\overline\alpha \le p/\overline\alpha$ and $\nu = \overline n / \overline \alpha \le 2p/\overline\alpha$.
  Therefore, there is a constant $c_0$ with
  \begin{equation} \label{eq:C0-detailed}
  C(0) \le c_0 \cdot \sMMa\!\left (\frac{p}{\overline \alpha},\overline \alpha\right).
  \end{equation}

  Let us now bound $C(k)$ when $k \ge 1$.  Given
  $\mV,\mW\in\F[x]_\nu^{\overline \alpha\times\gamma}$, we can compute
  $\mV_0$, $\mV_1$, $\mW_0$, $\mW_1$ for free.  
  Then $\mR'$ is computed
  recursively in time $C(k-1)$, and it remains to bound the cost of
  producing the result as $\mR = \mQ +x^{\nu'} \mR'$, with
  $\mQ=\mU^t\mV_0\mW_0^t$.

  For now let us write $D(k)$ for the cost of computing $\mQ$.  Given
  $\mR'$ in $\F[x]_{m+\nu'-1}^{1\times\overline\alpha}$ with $\nu' = \nu/2$ and $\mQ$ in
  $\F[x]_{m+\nu-1}^{1\times\overline\alpha}$, we can add them together
  using at most $\overline \alpha(m+\nu'-2)$ additions in $\F$.  Since $\nu' = \overline
  n/2\gamma \le n/\gamma \le n$, we deduce that for $k\ge 1$, 
  \begin{align}\label{eq:Ck}
    C(k) \le  C(k-1) + D(k) + 2 \, \overline \alpha \, p.
  \end{align}
  In order to bound $D(k)$ let us rewrite as before $\mU^t$ as $\mU^t
  = [1~x^c~x^{2c}~\cdots~x^{(\gamma-1)c}]\mU'$ with $c = \lfloor
  m/\gamma \rfloor$ and $\mU'\in\F[x]_c^{\gamma \times \overline\alpha}$;
  we compute $\mQ$ as $\mQ=[1,x^c,x^{2c},\ldots,x^{(\gamma-1)c}]
  (\mU' \mV_0 \mW_0^t)$.  The product $\mU' \mV_0 \mW_0^t$
  involves three polynomial matrices of respective dimensions $\gamma
  \times\overline \alpha$, $\overline\alpha \times \gamma$, $\gamma
  \times\overline \alpha$, with entries of degree less than
  $\max\{c,\nu'\}$. 
  Furthermore, since $c \le m/\gamma$ and, as seen above, $\nu' \le n/\gamma$, 
  we have $\max\{c,\nu'\} \le p/\gamma$.
  Since $\gamma \le \overline\alpha$, we can proceed
  as in the second case considered in the proof of~\cite[Lemma~7]{BoJeSc08}
  to show that $\mU' \mV_0 \mW_0^t$ can be evaluated as $(\mU' \mV_0) \mW_0^t $ in time
\[
\frac{\overline\alpha}{\gamma} \, \sMMa\left(\frac{p}{\gamma},\gamma\right) + 2\,\overline\alpha\,p 
+ \frac{\overline\alpha}{\gamma} \, \sMMa\left(\frac{2p}{\gamma},\gamma\right).
\]
 Since $\mU' \mV_0 \mW_0^t$ has dimensions $\gamma\times \overline\alpha$ and degree less than $3p/\gamma$,
 reconstructing $\mQ$ from that product requires at most $\gamma \overline\alpha \cdot 3p/\gamma = 3\,\overline\alpha\,p$
 additions in $\F$. Thus, using $\overline \alpha = 2^k \gamma$, 
\[
 D(k) \le  2^k \left( \sMMa\left(\frac{2^k p}{\overline\alpha},\frac{\overline\alpha}{2^k}\right) 
+ \sMMa\left(\frac{2^{k+1}p}{\overline\alpha},\frac{\overline\alpha}{2^k}\right)\right)+ 5\,\overline \alpha\,p.
\]
Combining this bound with~(\ref{eq:Ck}) 
and the fact that $\sMMa(2d,n) = O(\sMMa(d,n))$,
we deduce that there exists a constant $c_1$ such that
for all $k\ge 1$,
\begin{equation}\label{eq:Ck-detailed}
 C(k) \le C(k-1) + c_1\cdot 2^k \sMMa\left(\frac{2^k p}{\overline\alpha},\frac{\overline\alpha}{2^k}\right).
\end{equation}

 Taking all $k=0,\dots,\kappa$ into account, we deduce from~(\ref{eq:C0-detailed}) and~(\ref{eq:Ck-detailed}) 
 that the total time of {\sf mul} is bounded as
 \begin{align*}
 C(\kappa) &\le \max \{c_0,c_1\} \cdot  
 \sum_{k=0}^{\kappa} 2^k \sMMa \left (\frac{2^k p}{\overline \alpha},\frac{\overline \alpha}{2^k}\right)
 \in O \left(\sMMa'\left (\frac{p}{\overline \alpha},\overline \alpha\right )\right ).
 \end{align*}
 The conclusion follows from the facts that $p/\overline \alpha
 \le p /\alpha$ and that $\overline \alpha < 2 \alpha$.
\end{proof}

\bigskip
 
\paragraph{Case $\alpha < \beta$} In this situation, we split the 
vector $\mW$ into vectors $\mW_1,\dots,\mW_c$, with each $\mW_i$ in 
$\F[x]_n^{\alpha \times 1}$ and $c =\lceil \beta /\alpha\rceil$ (so 
$\mW_\c$ may be padded with zeros). 
Then, algorithm {\sf mul} is applied $c$ times (namely, to the $(U,V,W_i)$)
and by Lemma~\ref{prop:recursive-routine} we obtain $R_1,\ldots,R_\beta$ in time
$$O\left(\frac{\beta}{\alpha}\sMMa'\left(\frac{p}{\alpha},\alpha\right)\right).$$

\smallskip

\paragraph{Case $\beta < \alpha$} In this case, we split 
the vectors $\mU$ and $\mV$ into $\mU_1,\dots,\mU_c$ and
$\mV_1,\dots,\mV_c$, with each $\mU_i$ and $\mV_i$ in $\F[x]_n^{\beta
  \times 1}$ and $c =\lceil \alpha /\beta\rceil$. 
As before, algorithm {\sf mul} is applied $c$ times, but now to the $(\mU_i,\mV_i,\mW)$,
for a cost of
$$O\left(\frac{\alpha}{\beta}\sMMa'\left(\frac{p}{\beta},\beta\right)\right);$$
the $c$ results thus produced are then added, for a negligible cost of $O(\alpha p)$. 

\medskip

In the next subsection, we denote by {\sf
  mul}($\mU,\mV,\mW,m,n,\alpha,\beta$) the algorithm that handles all
possible cases for 
$\alpha \le n$, 
following the discussion in
the above paragraphs.

\subsection{Computing the $R_i$'s in the general case} \label{ssec:Rgen}

We now address the case of an arbitrary $Q$. Given such a $Q$ of
degree $n$ in $\F[x]$, as well as  $U_1,\dots,U_\alpha$ in
$\F[x]_m$, $V_1,\dots,V_\alpha$ in $\F[x]_n$ and $W_1,\dots,W_\beta$
in $\F[x]_n$, we recall that our goal is to compute $$R_i=\sum_{k \le
  \alpha} U_k (V_k W_i \bmod Q),\quad i=1,\dots,\beta.$$

\begin{theorem}\label{theo:Rgen}
  Assume $\alpha \le n$, and let $p = \max(m,n)$, 
  $\alpha'=\min(\alpha,\beta)$, and $\beta'=\max(\alpha,\beta)$.  
  Then
  one can compute the polynomials $R_1,\dots,R_\beta$ in time
  $$O\left(\frac{\beta'}{\alpha'}\sMMa'\left(\frac{p}{\alpha'},\alpha'\right)\right).$$
\end{theorem}

To prove Theorem~\ref{theo:Rgen}, 
the idea is to use the well-known fact that Euclidean division reduces
to power series inversion and multiplication;
in this way, the proof of Theorem~\ref{theo:Rgen} will reduce to that of Theorem~\ref{theo:R}.

The first step is to replace remainders by quotients. 
By applying the Euclidean division equality $V_k W_i = (V_k W_i \bdiv Q)Q + (V_k W_i \bmod Q)$
and defining
$$
S_i = \sum_{k \le \alpha} U_k (V_k W_i \bdiv Q)
\quad \text{for $i \le \beta$ \,\, and} \quad
T =  \sum_{k \le \alpha} U_k V_k,
$$
we easily deduce the following lemma.
\begin{lemma}
$R_i = T W_i - Q S_i$ for $i \le \beta$.
\end{lemma}

The main issue in our algorithm will be the computation of
$S_1,\dots,S_\beta$; we will actually compute their reverse
polynomials. Observe that $U_i$ has degree at most $m-1$, $V_k W_i
\bdiv Q$ has degree at most $n-2$, and $S_i$ has degree at most
$m+n-3$. 
We thus introduce the polynomials $\widetilde{S_i}=\rev(S_i,m+n-3)$;
knowing those, we can recover the polynomials $S_i$ for free.
For $k \le\alpha$ and $i \le \beta$, let us also define
$$\widetilde{U_k} = \rev(U_k,m-1), \,\,
  \widetilde{V_k} = \frac{\rev(V_k,n-1)}{\rev(Q,n)} \bmod x^{n-1}, \,\,
  \widetilde{W_i} = \rev(W_i,n-1) \bmod x^{n-1}.$$
All these polynomials are related by the following formula.
\begin{lemma}
  $\widetilde {S_i} = \sum_{k \le \alpha} \widetilde {U_k} (\widetilde {V_k} \widetilde{W_i} \bmod x^{n-1})$ for $i \le \beta$.
\end{lemma}
\begin{proof}
Using that $\rev(ab,r+s) = \rev(a,r)\rev(b,s)$ for any polynomials 
$a,b\in\F[x]$ of respective degrees $r$ and $s$, 
we deduce from the definitions of $S_i$ and $\widetilde S_i$ that
$$\widetilde S_i = \sum_{k \le \alpha} \widetilde{U_k} \rev( V_k W_i \bdiv Q,n-2).$$
Following~\cite[p.~258]{GaGe13}, we conclude with the equalities
\begin{align*}
\rev(V_k W_i \bdiv Q,n-2) &= \frac{\rev(V_k W_i,2n-2)}{\rev(Q,n)} \bmod x^{n-1}  \\
&= \frac{\rev(V_k,n-1)\,\rev(W_i,n-1)}{\rev(Q,n)} \bmod x^{n-1}\\
&= \widetilde{V_k} \widetilde{W_i}  \bmod x^{n-1}.
\end{align*}

\end{proof}

Figure~\ref{fig2} details the algorithm we just sketched.
To prove Theorem~\ref{theo:Rgen}, it suffices to observe that the cost
boils down to one call to ${\sf mul}$ plus $O((\alpha+\beta) \sM(p))$ for all
other operations, and then to apply Theorem~\ref{theo:R}.

\begin{figure}[htbp]
\begin{center}
\fbox{
\begin{minipage}{12.2cm}
Algorithm {\sf mulQ}($\mU,\mV,\mW,m,n,\alpha,\beta,Q$)
\medskip

{\it Input:} $\mU\in\F[x]_m^{\alpha\times 1}$,
              $\mV\in\F[x]_n^{\alpha\times 1}$, $\mW\in\F[x]_n^{\beta\times 1}$\\ 
\textcolor{white}{{\it Input:}} $Q$ monic of degree $n$ in $\F[x]$ \\
{\it Assumption:} 
$\alpha \le n$\\
{\it Output:} $\mR\in\F[x]_{m+n-1}^{1\times\beta}$ such that $\mR = \mU^t (\mV \mW^t \bmod Q)$.

 \medskip

$\widetilde{\mU} := [\rev(U_k,m-1), k=1,\dots,\alpha]$\\
$\widetilde{\mV} := [\rev(V_k,n-1)/\rev(Q,n) \bmod x^{n-1}, k=1,\dots,\alpha]$\\
$\widetilde{\mW} := [\rev(W_i,n-1) \bmod x^{n-1}, i=1,\dots,\beta]$\\
$\widetilde{\mS} := {\sf mul}(\widetilde{\mU},\widetilde{\mV},\widetilde{\mW},m,n-1,\alpha,\beta)$\\
$T := \sum_{k=1,\dots,\alpha} U_k V_k$\\
$\mR := [T W_i - Q \rev(\widetilde{S_i},m+n-3), i=1,\dots,\beta]$\\
return $\mR$.
\end{minipage}
}
\caption{Algorithm {\sf mulQ}.}\label{fig2}
\end{center}
\end{figure}

\section{Algorithms for inversion and system solving} \label{sec:inverse}
We conclude in this section by establishing the cost bounds announced in Theorem~\ref{theo:mainSylv}
for structured inversion and structured system solving.
We rely on Theorem~\ref{theo:equiv} from
Section~\ref{sec:equiv} to reduce our task to the case of the Hankel operator $\nabla_{\ZZ_{m,0},\ZZ_{n,1}^t}$, 
then use the fast multiplication algorithm of
Section~\ref{sec:multiplication_alogrithms} to speed up the Morf \!/\!
Bitmead-Anderson (MBA) approach.

\begin{theorem}\label{theo:inv}
  For any invertible operator $\mcL$ associated with $(\vP,\vQ)$, we can take
  $$\sMI(\mcL, \alpha) = O\left (\sMMa''\left(\frac{m}{\alpha},\alpha\right) \right)$$
and
$$\sMS(\mcL, \alpha)\,=\,O\left (\sMMa''\left(\frac{p}{\alpha},\alpha\right) \right),
\qquad p = \max(m,n).$$
\end{theorem}

Recall that Theorem~\ref{theo:equiv} in Section~\ref{sec:equiv} proves that we can take
$$\sMI(\mcL, \alpha) \le \sMI(\nabla_{\ZZ_{m,0},\ZZ_{m,1}^t}, \alpha+2)
+O\big(\sD(\vP,\vQ) + \alpha\sC(\vP,\vQ) + \alpha^{\omega-1} m\big )$$
and
$$\sMS(\mcL, \alpha) \le \sMS(\nabla_{\ZZ_{m,0},\ZZ_{n,1}^t}, \alpha+2)
+O\big(\sD(\vP,\vQ) + \alpha\sC(\vP,\vQ)\big ),$$
so we are left with estimating the cost of inversion and system solving
for $\nabla_{\ZZ_{m,0},\ZZ_{n,1}^t}$. For this, we shall use 
the MBA approach with preconditioning introduced in~\cite{Kaltofen95}
together with the ``compression-free'' formulas from~\cite{JeaMou10} for generating
the matrix inverse.
Such formulas being expressed with products of the form 
(structured matrix) $\times$ (unstructured matrix),
this provides a way to exploit our results~for fast structured matrix
multiplication from Section~\ref{sec:multiplication_alogrithms}.
Specifically, 
starting from a $\nabla_{\ZZ_{m,0},\ZZ_{n,1}^t}$-generator $(\mG,\mH)$ of length $\alpha$ of $\mA \in \F^{m\times n}$,
we proceed in three steps as~follows.

\medskip

{\it Preconditioning.}
From $(\mG,\mH)$ we begin by deducing a generator of the preconditioned matrix 
$$\widetilde\mA = \U(\v_1)\mA\U(\v_2)^t,$$ 
where each $\U(\v_i)$ is a unit upper triangular Toeplitz matrix
defined by some random vector~$\v_i$;
as shown in~\cite{KaSa91}, such a preconditioning ensures that, with high probability,
$\widetilde\mA$ has {\it generic rank profile}, that is, 
denoting by $\widetilde\mA_k$ the leading principal submatrix 
of $\widetilde\mA$ of dimensions $k \times k$, 
$$\det\!\big(\widetilde\mA_k\big) \ne 0 \quad \mbox{for $k=1,\ldots,\rank(\widetilde\mA)$.}$$

In our case,
preconditioning will be done in such a way that 
the generator of $\widetilde\mA$ corresponds not to $\nabla_{\ZZ_{m,0},\ZZ_{n,1}^t}$ 
but to another Hankel operator, namely $\nabla_{\ZZ_{m,0},\ZZ_{n,0}^t}$
(algorithm {\sf precond}, Figure~\ref{fig:precond}).
The latter operator being singular but still {\it partly regular}~\cite[\S4.5]{Pan01}, 
such a generator will in fact be a triple $(\widetilde\mG,\widetilde\mH,\widetilde\u)$,
with $(\widetilde\mG,\widetilde\mH)$ of length $O(\alpha)$ and $\widetilde\u$ carrying
the last row of $\widetilde\mA$.
Also, we will guarantee that the matrices $\widetilde\mG$ and $\widetilde\mH$ have the
following special shape: their first $\alpha$ columns are precisely $\U(\v_1)\mG$ and $\U(\v_2)\mH$.
These two requirements will be key to 
exploit the ``compression-free'' inversion scheme of~\cite[\S4.3]{JeaMou10}
and then to move back from, say, $\widetilde\mA^{-1}$ to $\mA^{-1}$.

\medskip

{\it Computation of the leading principal inverse.}
Assuming that $\widetilde\mA$ has generic rank profile 
and given $(\widetilde\mG,\widetilde\mH,\widetilde\u)$,
we now consider computing a compact representation of the inverse of its largest nonsingular leading principal
submatrix, that is, of $\widetilde\mA_r^{-1}$ with $$r = \rank(\widetilde\mA) = \rank(\mA).$$
Since $\widetilde\mA_r$ is structured with respect to $\nabla_{\ZZ_{r,0},\ZZ_{r,0}^t}$,
its inverse is structured with respect to $\nabla_{\ZZ_{r,0}^t,\ZZ_{r,0}}$ and, 
as recalled in~(\ref{eq:mcLP}),
$$\rank\big(\nabla_{\ZZ_{r,0}^t,\ZZ_{r,0}}(\widetilde\mA_r^{-1})\big) 
= 
\rank\big(\nabla_{\ZZ_{r,0},\ZZ_{r,0}^t}(\widetilde\mA)\big) =: \rho.$$
Therefore, our second step (detailed in Section~\ref{ssec:leading_principal_inverse}) 
will be to compute a $\nabla_{\ZZ_{r,0}^t,\ZZ_{r,0}}$-generator of $\widetilde\mA_r^{-1}$
of length $\rho$. 

In fact, following~\cite[\S5]{Pan01}, 
we give an algorithm (\lpinv, Figure~(\ref{fig:lpinv})) that does {\it not}
assume $\widetilde\mA$ has generic rank profile but discovers whether this is the case or not.
It calls first an auxiliary routine, called \largest, 
which returns the size $\ell$ of the largest leading principal submatrix $\widetilde\mA_\ell$ having generic rank profile
together with a generator of the inverse $\widetilde\mA_\ell^{-1}$.
Then, from such a generator, one can check efficiently whether $\ell$ equals $\rank(\widetilde\mA)$,
which is a condition equivalent to $\widetilde\mA$ having generic rank profile.
Thus, overall, algorithm \lpinv\ generates $\widetilde\mA_r^{-1}$ 
if and only if $\widetilde\mA$ has generic rank profile (and reports 'failure' otherwise,
since this provides a way of certifying whether preconditioning $\mA$ was successful or not).

Algorithm \largest\ calls a core recursive routine \largestrec,
which can be seen as a combination of Kaltofen's algorithm Leading Principal Inverse~\cite{Kaltofen95}
and~\cite[algorithm~{\sf GenInvHL}]{JeaMou10},
which thus relies on products of the form (structured matrix) $\times$ (unstructured matrix).
Also, the generating matrices $\widetilde\mY_\ell$ and $\widetilde\mZ_\ell$ produced by algorithm \largest\
are specified to have the following shape: 
$$\widetilde\mY_\ell = -\widetilde\mA_\ell^{-1} \widetilde\mG_\ell, 
\qquad \widetilde\mZ_\ell = \widetilde\mA_\ell^{-t} \widetilde\mH_\ell,$$
where $\widetilde\mG_\ell$ is made from the first $\ell$ rows of $\widetilde\mG$, and similarly for $\widetilde\mH_\ell$.
When $\ell = r$,
the same specification is inherited by the generator $(\widetilde\mY_r, \widetilde\mZ_r,\widetilde\v)$ produced by algorithm \lpinv,
whose matrices satisfy
$$\widetilde\mY_r = -\widetilde\mA_r^{-1} \widetilde\mG_r, 
\qquad \widetilde\mZ_r = \widetilde\mA_r^{-t} \widetilde\mH_r$$
where $\widetilde\mG_r$ and $\widetilde\mH_r$ are the first $r$ rows of $\widetilde\mG$ and~$\widetilde\mH$, and
$\widetilde\v^t$ is the first row of $\widetilde\mA_r^{-1}$.

\medskip

{\it Generating inverses and solving linear systems.}
Given the rank $r$ of $\widetilde\mA$ and 
the $\nabla_{\ZZ_{r,0}^t,\ZZ_{r,0}}$-generator $(\widetilde\mY_r, \widetilde\mZ_r,\widetilde\v)$
of $\mA_r^{-1}$ as above,
it is immediate to decide whether $\widetilde\mA$ and $\mA$ are invertible and, if so, 
to deduce the $\nabla_{\ZZ_{m,0}^t,\ZZ_{m,1}}$-generator of $\mA^{-1}$ given by
$$\mY = -\mA^{-1}\mG, \qquad \mZ = \mA^{-t}\mH.$$
This corresponds to algorithm {\sf inv} in Figure~\ref{fig:algo_inv}.
Now, given an additional vector $\b \in \F^m$, 
we reduce the study of $\mA\x = \b$ to that of the equivalent 
linear system $\widetilde\mA\widetilde\x = \widetilde\b$,
where $\widetilde\b = \U(\v_1)\b$ and for which 
an algorithm can be derived directly from the one in~\cite[\S4]{KaSa91},
with the additional guarantee that if $\b = 0$ and the column rank of $\mA$ is not full, then
a nonzero solution $\x$ is obtained. This corresponds to algorithm {\sf solve} in Figure~\ref{fig:algo_solve}.
(Clearly, both {\sf inv} and {\sf solve} are Las Vegas algorithms---``always correct, probably fast'',
thanks to the specification of algorithm \lpinv.)

\medskip

The next three sections provide detailed descriptions of the algorithms mentioned above, 
namely {\sf precond}, \largestrec, \largest, \lpinv, {\sf inv}, {\sf solve},
together with their correctness and complexity proofs,
thereby establishing Theorem~\ref{theo:inv}.

\subsection{Preconditioning} \label{ssec:preconditioning}
We precondition our structured matrix $\mA$ as shown in Figure~\ref{fig:precond} below.
Here and hereafter, 
$\e_{n,i}$ denotes the $i$th unit vector in $\K^n$
and, for any given $m\times n$ matrix $\mM$ such that $n \ge \alpha$, 
we write $\M^{\mapsto \alpha}$ for the $m\times \alpha$ matrix 
obtained by keeping only the first $\alpha$ columns of $\M$.

\begin{figure}[htbp]
\begin{center}
\fbox{
\begin{minipage}{10.9cm}
Algorithm {\sf precond}($\mG,\mH,\v_1,\v_2$)
\medskip

{\it Input:} $(\mG,\mH)\in\F^{m\times\alpha} \times \F^{n\times\alpha}$ 
                  such that $\nabla_{\ZZ_{m,0},\ZZ_{n,1}^t}(\mA) = \mG\mH^t$,\\ 
\textcolor{white}{{\it Input:}} $\v_1\in\F^m$ and $\v_2\in\F^n$. \\[2mm]
{\it Output:} $(\widetilde\mG, \widetilde\mH) \in\F^{m\times(\alpha+4)}\times \F^{n\times(\alpha+4)}$ 
and $\widetilde\u \in \F^n$ such that 
\begin{itemize}
\item $\nabla_{\ZZ_{m,0},\ZZ_{n,0}^t}(\widetilde\mA) = \widetilde\mG\widetilde\mH^t$ for 
          $\widetilde\mA = \U(\v_1)\,\mA\,\U(\v_2)^t$; 
 \item $\widetilde\u^t$ is the last row of $\widetilde\mA$;
\vspace*{0.1cm}
 \item $\widetilde\mG^{\mapsto\alpha} = \U(\v_1)\mG$ and $\widetilde\mH^{\mapsto\alpha} = \U(\v_2)\mH$.
\end{itemize}

\medskip

$\mG_1 := \big[ \ZZ_{m,0}\J_m \v_1\,|\, -\e_{m,1}\big]$; 
$\mH_1 := \big[\e_{m,m} \,|\, \ZZ_{m,0}^t \v_1\big]$ \\[1mm]
$\mG_2 := \big[ \ZZ_{n,1}^t \v_2 \,|\, -\e_{n,n} \big]$; 
$\mH_2 := \big[ \e_{n,1} \,|\, \ZZ_{n,0}\J_n \v_2 \big]$ \\[1mm]
$\widetilde\mG := \big[\U(\v_1)\mG \,|\, \mG_1 \,|\, \U(\v_1)\mA\mG_2 \big]$;
$\widetilde\mH := \big[\U(\v_2)\mH \,|\, \U(\v_2)\mA^t\mH_1 \,|\, \mH_2\big]$\\[1mm]
$\widetilde\u := \U(\v_2)\mA^t \U(\v_1)^t \e_{m,m}$\\[1mm] 
return $(\widetilde\mG, \widetilde\mH, \widetilde\u)$.
\end{minipage}
}
\caption{Algorithm {\sf precond}.}\label{fig:precond}
\end{center}
\end{figure}

\begin{lemma} \label{lem:preconditioning}
Algorithm {\sf precond} works correctly in time $O(\alpha \M(p))$.
Furthermore, if the vectors $\v_1\in\F^m$ and $\v_2\in\F^n$ have 
their first entry equal to $1$ and their remaining $m+n-2$ entries 
chosen uniformly at random from a finite subset $S \subset \F$, then 
the matrix $\widetilde\mA = \U(\v_1) \mA \U(\v_2)^t$ has generic rank profile 
with probability at least $$1-r(r+1)/|S|,$$
where $r = \rank\big(\widetilde\mA\big) = \rank(\mA)$ and $|S|$ is the cardinality of $S$.
\end{lemma}
\begin{proof}
We start by checking $\nabla_{\ZZ_{m,0},\ZZ_{n,0}^t}(\widetilde\mA) = \widetilde\mG {\widetilde \mH}^t$,
from which correctness follows.
Writing $\mcL_1 = \nabla_{\ZZ_{m,0},\ZZ_{m,0}}(\U(\v_1))$ 
and $\mcL_2 = \nabla_{\ZZ_{n,1}^t,\ZZ_{n,0}^t}(\U(\v_2)^t)$
and applying~(\ref{eq:general-formula-ABC}) gives 
$$
\nabla_{\ZZ_{m,0},\ZZ_{n,0}^t}(\widetilde\mA)
= \mcL_1 \mA \U(\v_2)^t + \U(\v_1) \mG \mH^t \U(\v_2)^t + \U(\v_1)\mA \mcL_2,
$$
and it remains to check that $\mcL_1 = \mG_1\mH_1^t$ and $\mcL_2 = \mG_2\mH_2^t$.
Since $\U(\v_1)$ is upper triangular Toeplitz, the matrix $\mcL_1 = \ZZ_{m,0}\U(\v_1) - \U(\v_1)\ZZ_{m,0}$ 
is zero everywhere except on its last column, which is equal to $\ZZ_{m,0}\J_m\v_1$,
and on its first row, which is equal to $-\v_1^t\ZZ_{m,0}$. 
Hence $\mcL_1 = \ZZ_{m,0}\J_m\v_1 \, \e_{m,m}^t - \e_{m,1} \, (\ZZ_{m,0}^t \v_1)^t = \mG_1\mH_1^t$.
For $\mcL_2$, we proceed similarly, by noting that $\mcL_2 = \ZZ_{n,1}^t \U(\v_2)^t - \U(\v_2)^t \ZZ_{n,0}^t$  
is zero everywhere but on its first column and last row, 
equal to $\ZZ_{n,1}^t \v_2$ and $-\ZZ_{n,0} \J_n \v_2$, respectively.

Let us now bound the cost of deducing $\widetilde\mG$, $\widetilde\mH$, $\widetilde\u$ from $\mG$, $\mH$, $\v_1$, $\v_2$.
First, we set up the $m\times 2$ matrices $\mG_1$ and $\mH_1$ and the $n \times 2$ matrices $\mG_2$ and $\mH_2$
in time $O(p)$.
Then, since $\mA$ satisfies $\nabla_{\ZZ_{m,0},\ZZ_{n,1}^t}(\mA) = \mG\mH^t$,
multiplying $\mA$ or $\mA^t$ be a single vector can be done in time~$O(\alpha \M(p))$, using for example
the reconstruction formula in~\cite[(21b)]{JeaMou10}. 
Hence we obtain the products $\mA \mG_2$, $\mA^t \mH_1$, and $\mA^t \U(\v_1)^t \e_{m,m}$ in time $O(\alpha \M(p))$.
Finally, since $\mG$ and $\mH$ have $\alpha$ columns each, it remains to multiply $O(\alpha)$ vectors by the triangular
Toeplitz matrices $\U(\v_1)$ and $\U(\v_2)$, and this can be done in time $O(\alpha \M(p))$.
Overall, the cost of the algorithm is thus bounded by $O(\alpha \M(p))$.

The probability analysis is due to Kaltofen and Saunders~\cite[Theorem~2]{KaSa91}.
\end{proof}

\subsection{Computation of the leading principal inverse} \label{ssec:leading_principal_inverse}
In order to generate the inverse of the largest nonsingular leading principal submatrix of~$\mA$, 
we start with the generation of the largest leading principal submatrix
having generic rank profile. 
This is done by algorithms \largestrec\ and \largest\ displayed in Figures~\ref{fig:algo_largestrec} and~\ref{fig:algo_largest}.
In algorithm \largestrec\ the matrices $\mA_{ij}$, $\mG_i$, $\mH_j$ for $1 \le i,j \le 2$
have dimensions $m_i \times m_j$, $m_i \times \alpha$, $m_j \times \alpha$, respectively, 
and correspond to the block partitions
$$
\mA  = \begin{bmatrix} \mA_{11} & \mA_{12} \\ 
\mA_{21} & \mA_{22} \end{bmatrix}\!, 
\quad \mG = \begin{bmatrix} \mG_1 \\ \mG_2 \end{bmatrix}\!,
\quad \mH = \begin{bmatrix} \mH_1 \\ \mH_2 \end{bmatrix}.
$$
Similarly, $\u_{ij}^t$ denotes the last row of $\mA_{ij}$.

\begin{figure}[htbp]
\begin{center}
\fbox{
\begin{minipage}{12.2cm}
Algorithm \largestrec($\mG,\mH,\u$)

\medskip
{\it Input:} $(\mG,\mH,\u)\in\F^{m \times\alpha} \times \F^{n\times\alpha} \times \F^n$ 
                  such that $\alpha \le \min(m,n)$ and \\
\textcolor{white}{{\it Input:}} $\nabla_{\ZZ_{m,0},\ZZ_{n,0}^t}(\mA) = \mG\mH^t$ 
and $\u^t=\e_{m,m}^t \mA$ (the last row of $\mA$). \\ 
{\it Output:} $(\ell,\mY,\mZ,\v)
\in \NN \times \F^{\ell\times\alpha} \times \F^{\ell\times\alpha} \times \F^\ell$ 
such that $\ell$ is the order of the largest\\ \hspace*{1.375cm}
leading principal submatrix $\mA_\ell$ of $\mA$ having
generic rank profile,\\ \hspace*{1.375cm}
$\mY = -\mA_\ell^{-1}\mG_\ell$, 
$\mZ = \mA_\ell^{-t}\mH_\ell$
and $\v = \mA_\ell^{-t}\e_{\ell,1}$ (the first row of $\mA_\ell^{-1}$). \\

if $\min(m,n) < 2\alpha$ then \\
\hspace*{0.5cm} compute $\mA$ explicitly and then deduce $\ell$, $\mY, \mZ, \v$.\\
else \\
\hspace*{0.5cm} $m_1 := \lceil m/2 \rceil$; $m_2 := \lfloor m/2 \rfloor$; $n_1 := \lceil n/2 \rceil$; $n_2 := \lfloor n/2 \rfloor$ \\
\hspace*{0.5cm} $(\ell_{11},\mY_{11},\mZ_{11},\v_{11}) := \largestrec(\mG_1,\mH_1,\u_{11})$ \\
\hspace*{0.5cm} if $\ell_{11} < \min(m_1,n_1)$ then \\
\hspace*{1cm} $(\ell,\mY,\mZ,\v) := (\ell_{11},\mY_{11},\mZ_{11},\v_{11})$ \\
\hspace*{0.5cm} else \\
\hspace*{1cm} $\mG_\mS := \mG_2+\mA_{21} \mY_{11}$; $\mH_\mS := \mH_2 - \mA_{12}^t\,\mZ_{11}$ \\
\hspace*{1cm} $\u_\mS := \u_{22} - \mA_{12}^t \mA_{11}^{-t} \u_{21}$ \\
\hspace*{1cm} if the $(1,1)$ element of $\mS$ is zero then \\
\hspace*{1.5cm} $(\ell,\mY,\mZ,\v) := (\ell_{11},\mY_{11},\mZ_{11},\v_{11})$ \\
\hspace*{1cm} else \\
\hspace*{1.5cm} $(\ell_\mS,\mY_\mS,\mZ_\mS,\v_\mS) := \largestrec(\mG_\mS,\mH_\mS,\u_\mS)$ \\
\hspace*{1.5cm} $\ell := \ell_{11} + \ell_\mS$ \\[2mm]
\hspace*{1.5cm} $\mY := \left [\begin{smallmatrix} 
\mY_{11} - (\mA_{11}^{-1}\mA_{12})\mY_\mS\\[1mm]
\mY_\mS
\end{smallmatrix} \right ]$; $ \mZ := \left [\begin{smallmatrix} 
\mZ_{11} - (\mA_{11}^{-T}\mA_{21}^T)\mZ_\mS\\[1mm] \mZ_\mS \end{smallmatrix} \right ]$;\\[2mm]
\hspace*{1.5cm} $\w := -\mS^{-T} \mA_{12}^T \v_{11}$; $\v := \left [\begin{smallmatrix} 
\v_{11} - \mA_{11}^{-T} \mA_{21}^T \w\\[1mm]
\w
\end{smallmatrix} \right ]$;\\
return $(\ell,\mY,\mZ,\v)$.
\end{minipage}
}
\caption{Algorithm \largestrec.
Here $\mA_\ell \in \F^{\ell \times \ell}$ denotes the largest leading principal submatrix of $\mA$ whose rank profile is generic,
and $\mG_\ell$ and $\mH_\ell$ are the $\ell\times\alpha$ submatrices
consisting of the first $\ell$ rows of $\mG$ and $\mH$, respectively.} 
\label{fig:algo_largestrec}
\end{center}
\end{figure}

\begin{lemma} \label{lem:largest_rec}
Algorithm \largestrec\ is correct and
if $m=n$ is an integer power of two, then its cost is 
$$ O \! \left (\sMMa''\left(\frac{m}{\alpha},\alpha\right) \right)\!.$$
\end{lemma}
\begin{proof}
Correctness follows directly from combining the analysis in~\cite[pp.~801--803]{Kaltofen95}
with the formulas for $\mG_\mS$, $\mH_\mS$, $\u_\mS$, $\mY$, $\mZ$, $\w$
given in~\cite[p.~287]{JeaMou10}.

Assuming $m=n$ is a power of two, let $C(m,\alpha)$ denote the cost of algorithm \largestrec.
We begin by showing that we can take
$$C(m,\alpha) = O(\alpha^\omega) \qquad \mbox{if $\alpha \le m < 2\alpha$.}$$
Given $\mG$ and $\mH$ we compute the product $\mG\mH^t$ in time $O(\alpha^\omega)$.
Then, starting from the last row of $\mA$ and using the fact that for $i,j > 1$ 
the $(i,j)$ entry of $\mG\mH^t$ equals $a_{i-1,j} - a_{i,j-1}$, we deduce all the entries
of $\mA$ in time $O(\alpha^2)$.

Let us now bound the cost when $m  \ge 2\alpha$.
Proceeding as in the proof of~\cite[Lemma~5]{JeaMou10},
it is easily seen that 
one can compute some generators of length at most $\alpha+2$ for the matrices 
$\mA_{21}$, $\mA_{12}^t$, $\mA_{11}^{-1}\mA_{12}$, $\mA_{22}^{-t} \mA_{21}^t$
as well as the vectors $\u_{11}$, $\u_{21}$, $\u_{22}$, $\u_\mS$
for a total time in $O(\alpha\,\M(m))$.
The number of additions is in $O(\alpha\,m)$ and we also have to perform
the four matrix products 
$\mA_{21}\mY_{11}$, $\mA_{12}^t\,\mZ_{11}$, 
$(\mA_{11}^{-1}\mA_{12})\mY_\mS$,
and $(\mA_{11}^{-T}\mA_{21}^T)\mZ_\mS$.
For example, $\mA_{12}$ has displacement rank at most $\alpha+2$ with respect to
the operator $\nabla_{\ZZ_{m/2,0},\ZZ_{m/2,1}^t}$, so that using a decomposition of the form
$\mA_{12} =\mA_{12}' + \mA_{12}''$ with $\mA_{12}'$ and $\mA_{12}''$
of displacement ranks at most $\alpha$ and $2$, respectively, 
we can evaluate the product $\mA_{21}\mY_{11}$ in time 
$\sMM(\nabla_{\ZZ_{m/2,0},\ZZ_{m/2,1}^t}, \alpha,\alpha) + O(\alpha\,\M(m))$;
by Theorem~\ref{theo:mainSylv2}, this is in $O\big(\frac{m}{2\alpha} \sMMa'(\frac{m}{2\alpha},\alpha)\big)$.
Adding up all these costs thus leads to
$$C(m,\alpha) = 2C\!\left(\frac{m}{2}, \alpha\right) 
+ O\!\left(\sMMa'\!\left(\frac{m}{2\alpha},\alpha\right)\right) \qquad \mbox{if $m \ge 2\alpha$.}$$
Hence, for some constant $c_0$,
$$C(m,\alpha) \le c_0 \left( \sum_{i=0}^{i_0-1} 2^i \sMMa'\!\left(\frac{m}{2^{i+1}\alpha},\alpha\right) 
+ 2^{i_0}\alpha^\omega\right)\!$$
with $i_0\in\NN$ such that $\alpha \le m/2^{i_0} < 2\alpha$,
that is, $i_0 = \lfloor \log(m/\alpha)\rfloor$.
Defining $\overline{\alpha} = 2^{\lceil\log(\alpha)\rceil}$,
we have $\alpha \in (\overline{\alpha}/2,\overline{\alpha}]$, 
which implies $m/\alpha \in [m/\overline{\alpha}, 2m/\overline{\alpha})$ and thus, 
since $m$ is an integer power of two, $$2^{i_0} = m / \overline{\alpha}.$$
Besides, $m/ (2^{i+1}\alpha) < m / (2^i \overline{\alpha})$,
which by assumption on $\sMMa'$ implies
$$\sMMa'\!\left(\frac{m}{2^{i+1}\alpha},\alpha\right) = O\!\left(\sMMa'\!\left(\frac{m}{2^i \overline{\alpha}},\alpha\right)\right).$$
Third, $\alpha^\omega = O\big(\sMMa'(1,\alpha)\big)$.
Consequently, for some constant $c_1$,
$$C(m,\alpha) \le c_1 \cdot \sum_{i=0}^{\log(m/\overline{\alpha})} 2^i \sMMa'\!\left(\frac{m}{2^i \overline{\alpha}},\alpha\right).$$
Now, $\alpha \le \overline{\alpha}$ implies $m/\overline{\alpha} \le m/\alpha \le \overline{m/\alpha}$,
where $ \overline{m/\alpha}$ denotes the smallest integer power of two greater than or equal to $m/\alpha$.
Hence $\log(m/\overline{\alpha}) \le \log\!\big(\overline{m/\alpha}\,\big)$ and
$\sMMa' (2^{-i} m/ \overline{\alpha}) = O\big(\sMMa'\big(2^{-i} \overline{m/\alpha}\,\big)\big)$ by assumption on $\sMMa'$,
so that the sum in the above bound on $C(m,\alpha)$ is in $O(\sMMa''(m/\alpha,\alpha))$,
as announced.
\end{proof}

\begin{figure}[htbp]
\begin{center}
\fbox{
\begin{minipage}{12.2cm}
Algorithm \largest($\mG,\mH,\u$)
\medskip

{\it Input:} $(\mG,\mH,\u)\in\F^{m \times\alpha} \times \F^{n\times\alpha} \times \F^n$ 
                  such that $\alpha \le \min(m,n)$ and \\
\textcolor{white}{{\it Input:}} $\nabla_{\ZZ_{m,0},\ZZ_{n,0}^t}(\mA) = \mG\mH^t$ 
and $\u^t=\e_{m,m}^t \mA$ (the last row of $\mA$). \\ 
{\it Output:} $(\ell,\mY,\mZ,\v)
\in \NN \times \F^{\ell\times\alpha} \times \F^{\ell\times\alpha} \times \F^\ell$ 
such that $\ell$ is the order of the  largest \\ \hspace*{1.375cm}
leading principal submatrix $\mA_\ell$ of $\mA$ having
generic rank profile, \\ \hspace*{1.375cm}
$\mY = -\mA_\ell^{-1}\mG_\ell$, 
$\mZ = \mA_\ell^{-t}\mH_\ell$
and $\v = \mA_\ell^{-t}\e_{\ell,1}$ (the first row of $\mA_\ell^{-1}$). \\

$\overline{p} := 2^{\lceil \log_2(\max(m,n)) \rceil}$ \\
compute $\overline{\alpha}\in\NN$,
$\overline{\mG},\overline{\mH} \in \F^{\overline{p} \times \overline{\alpha}}$,
and $\overline{\u} \in \F^{\overline{p}}$ such that: \\[-3mm]
\begin{itemize}
\item $(\overline{\mG},\overline{\mH}, \overline{\u})$ is a $\nabla_{\ZZ_{\overline{p},0},\ZZ_{\overline{p},0}^t}$-generator 
of length $\overline{\alpha}$ of 
$\overline{\mA} = \left [\begin{smallmatrix} \mA&0\\0&0\end{smallmatrix} \right ]\in \F^{\overline{p} \times \overline{p}}$;
\vspace*{1mm}
\item $\alpha \le \overline{\alpha} \le \min(\alpha+2,\overline{p})$;
\vspace*{1mm}
\item $\overline{\mG} = \left [\begin{smallmatrix} \mG & *\\ {*}&*\end{smallmatrix} \right ]$
and $\overline{\mH} = \left [\begin{smallmatrix} \mH &*\\ {*}&*\end{smallmatrix} \right ]$
\vspace*{1mm}
\item $\overline{\u}^t$ is the last row of $\overline{\mA}$
\end{itemize}
\vspace*{1mm}
$(\ell, \overline{\mY}, \overline{\mZ}, \v)
  := \largestrec(\overline{\mG},\overline{\mH},\overline{\u})$ \\[1mm]
$\mY := \overline{\mY}\,^{\mapsto\alpha}$; $\mZ := \overline{\mZ}\,^{\mapsto\alpha}$ \\[1mm]
return $(\ell,\mY,\mZ,\v)$.
\end{minipage}
}
\caption{Algorithm \largest.
Here $\mA_\ell \in \F^{\ell \times \ell}$ denotes the largest leading principal submatrix of $\mA$ whose rank profile is generic,
and $\mG_\ell$ and $\mH_\ell$ are the $\ell\times\alpha$ submatrices
consisting of the first $\ell$ rows of $\mG$ and $\mH$, respectively.} 
\label{fig:algo_largest}
\end{center}
\end{figure}

\begin{lemma} \label{lem:largest}
Let $p = \max(m,n)$. 
Algorithm \largest\ works correctly in time  
$$ O \! \left (\sMMa''\left(\frac{p}{\alpha},\alpha\right) \right)\!.$$
\end{lemma}
\begin{proof}
Let us show first how to generate the augmented matrix $\overline{\mA}$.
If $\overline{p} = m = n$, then $\overline{\mA} = \mA$ and it suffices to take 
$\overline{\alpha}=\alpha$, $\overline{\mG}=\mG$, and $\overline{\mH} = \mH$.
If $\overline{p} = n > m$ then $\nabla_{\ZZ_{m,0},\ZZ_{n,0}^t}(\mA) = \mG\mH^t$ leads to the following generator of
length $\overline{\alpha} := \alpha + 1$:
$$
\nabla_{\ZZ_{\overline{p},0},\ZZ_{\overline{p},0}^t}\big( \overline{\mA} \big)
= \overline{\mG}\, \overline{\mH}^t, \qquad
\overline{\mG} = \begin{bmatrix} \mG & \\ & \e_{\overline{p}-m,1} \end{bmatrix}\!, \qquad
\overline{\mH} = \big [\, \mH \,|\, \u \,\big];
$$
the range constraint on $\overline{\alpha}$ is satisfied, since $\alpha \le \min(m,n) = m < n = \overline{p}$,
and, on the other hand, the upper left corners of $\overline{\mG}$ and $\overline{\mH}$ are $\mG$ and $\mH$, respectively.
Proceeding similarly when $\overline{p} = m > n$, we take
$$
\overline{\alpha} = \alpha + 1, \qquad \overline{\mG} = \big [\, \mG \,|\, -\u' \,\big], \qquad
\overline{\mH} = \begin{bmatrix} \mH & \\ & \e_{\overline{p}-n,1} \end{bmatrix}\!,
$$
where $\u'$ is the last column of $\mA$ and can be obtained in time $O(\alpha \M(\overline{p}))$.
It remains to handle the case $\overline{p} > \max(m,n)$, 
where $\mA$ is bordered by both some zero rows and some zero columns.
In this case, we deduce from $\nabla_{\ZZ_{m,0},\ZZ_{n,0}^t}(\mA) = \mG\mH^t$ that
a $\nabla_{\ZZ_{\overline{p},0},\ZZ_{\overline{p},0}^t}$-generator of length $\alpha+2$ of $\overline{\mA}$ is given by
$$
\overline{\mG} = \begin{bmatrix} \mG & & -\u' \\ & \e_{\overline{p}-m,1}& \end{bmatrix}, \qquad 
\overline{\mH} = \begin{bmatrix} \mH & \u & \\ & & \e_{\overline{p}-n,1}\end{bmatrix}.
$$
Again, it is clear that the upper left corners of $\overline{\mG}$ and $\overline{\mH}$ 
are $\mG$ and $\mH$, respectively.
Furthermore, if $\overline{p} \ge \alpha+2$, 
we can take $\overline{\alpha} = \alpha+2$;
otherwise, since $\overline{p} > \max(m,n) \ge \min(m,n) \ge \alpha$,
we are in the situation where $\overline{p} = \alpha + 1$ and $m=n=\alpha$.
In this case we shall proceed as follows to reduce the generator length from $\alpha+2$ to $\alpha+1$
while ensuring that $\mG$ and $\mH$ are in the upper left corners of the new generator matrices:
the matrices $\overline{\mG}$ and $\overline{\mH}$ 
defined above have dimensions $(\alpha+1) \times (\alpha+2)$ and are given~by
$$
\overline{\mG} = \begin{bmatrix} \mG & & -\u' \\ & 1 & \end{bmatrix} 
\qquad\mbox{and}\qquad
\overline{\mH}^t = \begin{bmatrix} \mH^t & \\ \u^t & \\ & 1 \end{bmatrix}\!;
$$
since $\mG$ is $\alpha \times \alpha$, we can deduce from $\mG$ and $\u'$ a vector $\u'' \in \F^\alpha$ 
such that $\mG \u'' = \u'$ in time $O(\alpha^\omega)$; 
then, using the fact that
$$\mE = \begin{bmatrix} \I_\alpha & & \u'' \\ & 1 & \\ & & 1\end{bmatrix}
\qquad \Longrightarrow \qquad
\overline{\mG} \mE = \begin{bmatrix} \mG & & 0 \\ & 1 & 0 \end{bmatrix}
\quad\mbox{and}\quad
\mE^{-1} \overline{\mH}^t = \begin{bmatrix} \mH^t & -\u'' \\ \u^t &  \\  & 1 \end{bmatrix}\!,
$$
we conclude that 
a suitable $\nabla_{\ZZ_{\overline{p},0},\ZZ_{\overline{p},0}^t}$-generator of $\overline{\mA}$
is obtained by taking the first $\alpha+1 =: \overline{\alpha}$ columns of $\overline{\mG} \mE$ and $\overline{\mH} \mE^{-t}$.
Finally, obtaining the vector $\overline{\u}$ is free, since the last row of 
$\overline{\mA}$ is either the last row of $\mA$ (which is part of the input)
or the zero row.
To summarize, we can always find a suitable generator $(\overline{\mG},\overline{\mH},\overline{\u})$
of the augmented matrix $\overline{\mA}$ in time $O(\alpha\M(\overline{p}) +\alpha^\omega)$,
that is, since $\overline{p} < 2p$ and $\alpha \le p$, 
$$O(\alpha\,\M(p) + \alpha^{\omega-1}p).$$

Since $\overline{\alpha} \le \overline{p}$ and $\overline{p}$ is a power of two,
Lemma~\ref{lem:largest_rec} implies that
the output $(\ell,\mY,\mZ,\v)$ of \largestrec\ is obtained
in time $$O(\sMMa''(\overline{p} / \overline{\alpha}, \overline{\alpha}))$$
and satisfies the following:
$\ell$ is the size of the largest leading principal submatrix of~$\mA$ having generic rank profile;
furthermore, denoting this matrix by $\mA_\ell$, we have
$\v^t$ equal to the first row of $\mA_\ell^{-1}$ and
$\overline{\mY} = -\mA_\ell^{-1}\overline{\mG}_\ell$ and
$\overline{\mZ} = -\mA_\ell^{-t}\overline{\mH}_\ell$
with $\overline{\mG}_\ell$ and $\overline{\mH}_\ell$ the $\ell \times \overline{\alpha}$ submatrices
consisting of the first $\ell$ rows of $\overline{\mG}$ and $\overline{\mH}$.
Since $\mG\in\F^{m\times\alpha}$ and $\mH\in\F^{n\times\alpha}$ 
are in the upper left corners of $\overline{\mG}$ and $\overline{\mH}$
and since $\ell \le \min(m,n)$, we deduce that 
$$\overline{\mG}_\ell = [\mG_\ell \,\, *] \qquad \mbox{and} \qquad \overline{\mH}_\ell = [\mH_\ell \,\,*].$$
Consequently, by keeping only the first $\alpha$ columns of each of $\overline{\mY}$ and $\overline{\mZ}$,
we obtain the matrix pair $(\mY,\mZ)$ such that
$$\mY = -\mA_\ell^{-1}\mG_\ell \qquad\mbox{and}\qquad \mZ = -\mA_\ell^{-t}\mH_\ell.$$
This concludes the proof of correctness.
For the cost, note that the smallest integer power of two greater than or equal to 
$\overline{p} / \overline{\alpha}$ is less than $4p/\alpha$ and, on the other hand, 
and that $\overline{\alpha} = O(\alpha)$. 
Hence $\sMMa''(\overline{p} / \overline{\alpha}, \overline{\alpha}) = O\big(\sMMa''(p/\alpha,\alpha)\big)$,
and the latter expression dominates over the term in $O(\alpha\,\M(p) + \alpha^{\omega-1}p)$.
\end{proof}

\begin{figure}[htbp]
\begin{center}
\fbox{
\begin{minipage}{12.2cm}
Algorithm \lpinv($\mG,\mH,\u$)
\medskip

{\it Input:} $(\mG,\mH)\in\F^{m\times\alpha} \times \F^{n\times\alpha}$ 
                  such that $\alpha \le \min(m,n)$ and $\nabla_{\ZZ_{m,0},\ZZ_{n,0}^t}(\mA) = \mG\mH^t$;\\ 
\textcolor{white}{{\it Input:}} $\u\in\F^n$ such that $\u^t=\e_{m,m}^t \mA$ (the last row of $\mA$). \\ 
{\it Output:} $(r,\mY,\mZ,\v)$ such that $r =\rank(\mA)$, $(\mY,\mZ) = (-\mA_r^{-1}\mG_r, \mA_r^{-t}\mH_r)$, and \\
\textcolor{white}{{\it Output:}} $\v = \mA_r^{-t}\e_{r,1}$
(the first row of $\mA_r^{-1}$) if $\mA$ has generic rank profile, \\
\textcolor{white}{{\it Output:}} and $r =$ ``failure'' otherwise.

\bigskip

$(\ell, \mY_\ell, \mZ_\ell, \v_\ell) := \largest(\mG, \mH, \u)$ \\
$\mG_\mS := \mG_2+\mA_{21} \mY_\ell$; $\mH_\mS := \mH_2 - \mA_{12}^t\,\mZ_\ell$\\
$\u_\mS := \u_{22} - \mA_{12}^t \mA_\ell^{-t} \u_{21}$ \\
$r_0 :=$ the rank of $[\mG_\mS | \e_{m-\ell,1}] \, [\mH_\mS |\u_\mS]^t$\\
if $r_0 = 0$ then \\
\hspace*{0.5cm} $(r,\mY,\mZ,\v) := (\ell, \mY_\ell, \mZ_\ell, \v_\ell)$ \\
else \\
\hspace*{0.5cm} $(r,\mY,\mZ,\v) := (\mbox{``failure''}, *, *, *)$ \\
return $(r,\mY,\mZ,\v)$.
\end{minipage}
}
\caption{Algorithm \lpinv.
Here $r$ denotes the (unknown) rank of $\mA$, and
$\mA_r$ is the $r\times r$ matrix consisting of the first $r$ rows and columns of $\mA$.
Similarly, $\mG_r$ and $\mH_r$ are the $r\times\alpha$ matrices
consisting of the first $r$ rows of $\mG$ and $\mH$, respectively.} 
\label{fig:lpinv}
\end{center}
\end{figure}

\begin{lemma} \label{lem:lp_inv}
Let $p = \max(m, n)$. 
Algorithm \lpinv\ works correctly in time $$O\left(\sMMa''\left(\frac{p}{\alpha},\alpha\right)\right).$$
\end{lemma}
\begin{proof}
By Lemma~\ref{lem:largest}, 
the call to \largest\ yields $(\ell,\mY_\ell,\mZ_\ell,\v_\ell)$ such that
$\ell$ is the order of the largest leading principal submatrix of $\mA$ having generic rank profile,
$\mY_\ell = -\mA_\ell^{-1}\mG_\ell$, $\mZ_\ell = \mA_\ell^{-t}\mH_\ell$, and $\v_\ell = \mA_\ell^{-t} \e_{\ell,1}$.
Thanks to the special shape of the matrices $\mY_\ell$ and $\mZ_\ell$, 
we can check as in~\cite{JeaMou10} that 
the expressions for $\mG_\mS$,  $\mH_\mS$, $\u_\mS$
lead to a $\nabla_{\ZZ_{m-\ell,1},\ZZ_{n-\ell,0}^t}$-generator of length $\alpha+1$ 
of the Schur complement $\mS = \mA_{22} - \mA_{21} \mA_\ell^{-1} \mA_{12}$:
$$\nabla_{\ZZ_{m-\ell,1},\ZZ_{n-\ell,0}^t}(\mS) =  [\mG_\mS | \e_{m-\ell,1}] [\mH_\mS |\u_\mS]^t.$$
Now, since both $\nabla_{\ZZ_{m-\ell,1},\ZZ_{n-\ell,0}^t}$ and $\mA_\ell$ are invertible,
$$r_0 := \rank\big(\nabla_{\ZZ_{m-\ell,1},\ZZ_{n-\ell,0}^t}(\mS)\big) = 0 
\qquad\Longleftrightarrow\qquad \mS = 0 
\qquad\Longleftrightarrow\qquad \ell = \rank(\mA).$$
Correctness then follows, since $\ell = \rank(\mA)$ if and only if $\mA$ has generic rank profile.

To bound the cost, recall first from Lemma~\ref{lem:largest} than 
$\ell$, $\mY_\ell$, $\mZ_\ell$, $\v_\ell$ are computed in time $O(\sMMa''(p/\alpha,\alpha))$.
Then we can obtain $\mG_\mS$ in time $O(\sMMa'(p/\alpha,\alpha))$
by evaluating the product $\mA_{21}\mY_\ell$ as follows.
Defining $\overline{\mA}_{21} = \left [ \begin{smallmatrix} 0 & 0 \\ 0 & \mA_{21} \end{smallmatrix}\right] \in \F^{m\times n}$
and $\overline{\mY}_\ell = \left [ \begin{smallmatrix} 0 \\ \mY_\ell \end{smallmatrix}\right] \in \F^{n \times \alpha}$,
we can deduce from $\mG,\mH,\u$ a $\nabla_{\ZZ_{m,1},\ZZ_{n,0}}$-generator  
$(\overline{\mG}_{21},\overline{\mH}_{21})$
of length $\alpha + 2$ of $\overline{\mA}_{21}$ in time $O(\alpha\,\M(p))$.
Writing $\overline{\mA}_{21} = \overline{\mA}_{21}' + \overline{\mA}_{21}''$
with $\overline{\mA}_{21}'$ of displacement rank $\alpha \le \min(m,n)$ and $\overline{\mA}_{21}''$ of displacement rank~$2$,
we can evaluate $\overline{\mA}_{21} \overline{\mY}_\ell$
in time $\sMM(\nabla_{\ZZ_{m,1},\ZZ_{n,0}^t}, \alpha,\alpha) + O(\alpha\,\M(p))$,
which by Theorem~\ref{theo:mainSylv2} is in $O(\sMMa'(p/\alpha,\alpha) + \alpha \, \M(p))$.
It remains to extract $\mA_{21} \mY_\ell$  
from $\overline{\mA}_{21} \overline{\mY}_\ell = \left [ \begin{smallmatrix} 0 \\ \mA_{21} \mY_\ell \end{smallmatrix}\right]$
and to add it to $\mG_2$, for an overhead of $O(\alpha p)$. 
Since $\alpha \, \M(p) = O(\sMMa'(p/\alpha,\alpha))$, we have thus obtained $\mG_\mS$ in time $O(\sMMa'(p/\alpha,\alpha))$.
The same cost bound can be derived for the matrix $\mH_\mS$ and,
on the other hand, the cost bound $O(\alpha \, \M(p))$ applies to the computation of $\u_\mS$
and follows from computing a generator of length $O(\alpha)$ of  $\mA_{12}^t \mA_\ell^{-t}$ and then
multiplying by and subtracting from a vector.
Finally, given $\mG_\mS$ $\mH_\mS$, $\u_\mS$, the rank $r_0$ can be deduced in time 
$O(\alpha^{\omega-1} p)$.
The conclusion for the cost then follows from the fact that $\alpha^{\omega-1} p$
and $\sMMa'(p/\alpha,\alpha)$ are both in $O(\sMMa''(p/\alpha,\alpha))$.
\end{proof}

\subsection{Generating inverses and solving linear systems} \label{ssec:algo_inv}
We conclude by showing how to apply algorithm \lpinv\ from the previous section 
to our initial problems ${\sf inv}(\mcL,\alpha)$
and ${\sf solve}(\mcL,\alpha)$. 
This corresponds to algorithms {\sf inv} and {\sf solve} given in Figures~\ref{fig:algo_inv} and~\ref{fig:algo_solve}.

\begin{figure}[htbp]
\begin{center}
\fbox{
\begin{minipage}{12.2cm}
Algorithm {\sf inv}($\mG,\mH,S$)
\medskip

{\it Input:} $(\mG,\mH)\in\F^{m\times \alpha}\times\F^{m\times \alpha}$ 
                  such that $\alpha\le m$ and $\nabla_{\ZZ_{m,0},\ZZ_{m,1}^t}(\mA)=\mG\mH^t$, \\
\textcolor{white}{{\it Input:}} and a finite subset $S\subset \F$. \\
{\it Output:} if not ``failure'', then $(\mY,\mZ) = (-\mA^{-1}\mG,\mA^{-t}\mH)$ or ``$\mA$ is singular''.\\

choose the entries of $\r_1, \r_2 \in\F^{m-1}$ uniformly at random from $S$ \\
$\v_1 := \big[1 \,|\, \r_1^t\big]^t$; $\v_2 := \big[1 \,|\, \r_2^t\big]^t$\\
$(\widetilde\mG,\widetilde\mH,\widetilde\u) := {\sf precond}(\mG,\mH,\v_1,\v_2)$\\
$(r,\widetilde\mY,\widetilde\mZ,*) := \lpinv(\widetilde\mG,\widetilde\mH,\widetilde\u)$\\
if $r \not\in\{0,1,\ldots,m\}$ then \\
\hspace*{0.5cm} return ``failure'' \\
else \\
\hspace*{0.5cm} if $r=m$ then \\
\hspace*{1cm} $\mY := \U(\v_2)^t \, {\widetilde\mY}^{\mapsto\alpha}$; 
                            $\mZ := \U(\v_1)^t\,{\widetilde\mZ}^{\mapsto\alpha}$\\
\hspace*{1cm} return $(\mY,\mZ)$ \\
\hspace*{0.5cm} else \\
\hspace*{1cm} return ``$\mA$ is singular'' \\
\end{minipage}
}
\caption{Algorithm {\sf inv}.}\label{fig:algo_inv}
\end{center}
\end{figure}

\begin{theorem} \label{thm:algo_inverse}
Algorithm {\sf inv} works correctly in time 
$O\left(\sMMa''\left(\frac{m}{\alpha},\alpha\right)\right).$
It makes $2m-2$ random choices in $\F$
and fails with probability less than~$1/2$ if $|S| \ge 2m(m+1)$.
\end{theorem}
\begin{proof}
By applying Lemma~\ref{lem:preconditioning} to the case $m=n$, 
we see that $\widetilde\u^t$ is the last row of $\widetilde\mA = \U(\v_1)\mA\U(\v_2)^t$
and that 
$\widetilde\mG$ and $\widetilde\mH$ are $m\times (\alpha+4)$ matrices such that
$$
\ZZ_{m,0}\,\widetilde\mA - \widetilde\mA\,\ZZ_{m,0}^t = \widetilde\mG \widetilde\mH^t,
\qquad
\widetilde\mG^{\mapsto\alpha} = \U(\v_1)\mG, \qquad \widetilde\mH^{\mapsto\alpha} = \U(\v_2)\mH.$$
If $r \not\in\{0,1,\ldots,m\}$ then $r =$ ``failure'', and this is what we return.
Assume now that $r \in \{0,1,\ldots,m\}$. 
In this case, preconditioning has ensured that $\widetilde\mA$ has generic rank profile.
The number $r$ produced by \lpinv\ thus satisfies $r = \rank(\widetilde\mA) = \rank(\mA)$,
and $\mA$ is singular if and only if $r\ne m$.
When $r=m$, we have $\widetilde\mA = \widetilde\mA_r$, so 
the matrices $\widetilde\mY$ and $\widetilde\mZ$ produced by \lpinv\ satisfy
$\widetilde\mY = -\widetilde\mA^{-1}\widetilde\mG$
and
$\widetilde\mZ = \widetilde\mA^{-t}\widetilde\mH$.
Using the special shape of the first $\alpha$ columns of $\widetilde\mG$ and $\widetilde\mH$,
we conclude that the returned matrices $\mY$ and $\mZ$ are $-\mA^{-1}\mG$ and $\mA^{-t}\mH$, as wanted.
(Note that although it is produced by \lpinv,
the first row of the inverse of $\widetilde\mA_r$ is not needed here and denoted by $*$; 
we will need it, however, when solving linear systems at the end of this section.)

Applying Lemmas~\ref{lem:preconditioning} and~\ref{lem:lp_inv} with $p=m$ shows that
the cost of calling {\sf precond} and \lpinv\ is $O(\alpha \, \M(m) + \sMMa''(m/\alpha,\alpha))$;
on the other hand, the Toeplitz structure of $\U(\v_1)$ and $\U(\v_2)$ implies that 
$\mY$ and $\mZ$ can be deduced in time $O(\alpha \, \M(m))$ 
from  $\v_1$, $\v_2$, $\widetilde\mY$, and $\widetilde\mZ$.
Writing $\overline{\alpha}$ and $\overline{m/\alpha}$ 
for the smallest integer powers of two greater than or equal to 
$\alpha$ and $m/\alpha$,
we have
$\sMMa''(m/\alpha,\alpha) \ge \sMMa'\big(\overline{m/\alpha},\alpha\big)
\ge \overline{\alpha}\, \sMMa\big(\overline{\alpha} \cdot \overline{m/\alpha},1\big)
= \overline{\alpha} \, \M\big(\overline{\alpha}\cdot \overline{m/\alpha}\big) \ge \alpha \, \M(m)$,
from which the claimed cost bound follows.

Finally, by Lemma~\ref{lem:preconditioning}, 
the preconditioned matrix $\widetilde\mA$ has generic rank profile with probability 
$\mathcal{P} \ge 1 - \rank(\mA) \cdot (\rank(\mA)+1)/|S|$.
Since $\rank(\mA) \le m$, we have $\mathcal{P} \ge 1 -m(m+1)/|S|$, so that $|S| \ge 2m(m+1)$ implies $\mathcal{P} \ge 1/2$.
\end{proof}

For solving $\mA\,\x = \b$, we work on the equivalent (preconditioned)
system $\widetilde\mA \, \widetilde\x = \widetilde\b$ such that
$\widetilde\mA = \U(\v_1) \mA \U(\v_2)^t$ and $\widetilde\b = \U(\v_1) \b$,
for which any solution $\widetilde\x$ yields a solution $\x = \U(\v_2)^t \, \widetilde\x $.
Algorithm~{\sf solve} in Figure~\ref{fig:algo_solve} uses the following notation:
writing as before $r$ for the rank of $\mA$, we partition $\widetilde\mA$ and $\widetilde\b$ into blocks as 
$$\widetilde \mA = \begin{bmatrix} \widetilde\mA_r & \widetilde\mA_{12} \\ 
\widetilde\mA_{21} & \widetilde\mA_{21}\,\widetilde\mA_r^{-1}\,\widetilde\mA_{12} \end{bmatrix},
\qquad 
\widetilde\b = \begin{bmatrix} \widetilde\b_1 \\ \widetilde\b_2 \end{bmatrix}
$$
with 
$\mA_{12} \in \F^{r\times (n-r)}$, 
$\widetilde\mA_{21}\in \F^{(m-r)\times r}$,
$\widetilde\b_1\in\F^r$ and
$\widetilde\b_2\in\F^{m-r}$.
\begin{figure}[htbp]
\begin{center}
\fbox{
\begin{minipage}{12.2cm}
Algorithm {\sf solve}($\mG,\mH,\b,S$)

\medskip
 
{\it Input:} $(\mG,\mH)\in\F^{m\times\alpha} \times \F^{n\times\alpha}$ 
                  such that $\alpha \le \min(m,n)$ and $\nabla_{\ZZ_{m,0},\ZZ_{n,1}^t}(\mA) = \mG\mH^t$, \\
\textcolor{white}{{\it Input:}} $\b\in\F^m$, and a finite subset $S\subset \F$. \\
{\it Output:} if not ``failure'', then a nontrivial solution $\x\in\F^n$ to $\mA\,\x=\b$ \\
\hspace*{1.2cm} or ``no solution exists.''\\

choose the entries of $\r_1\in\F^{m-1}$ and $\r_2\in\F^{n-1}$ uniformly at random from $S$ \\
$\v_1 := \big[1 \,|\, \r_1^t\big]^t$; $\v_2 := \big[1 \,|\, \r_2^t\big]^t$\\
$(\widetilde\mG,\widetilde\mH,\widetilde\u) := {\sf precond}(\mG,\mH,\v_1,\v_2)$\\
$(r,\widetilde\mY,\widetilde\mZ,\widetilde\v) := \lpinv(\widetilde\mG,\widetilde\mH,\widetilde\u)$\\ [1mm]
if $r \not\in\{0,1,\ldots,m\}$ then \\
\hspace*{0.5cm} return ``failure'' \\
else \\[1mm]
\hspace*{0.4cm} $\left [\begin{smallmatrix} \widetilde\b_1 \\ \widetilde\b_2 \end{smallmatrix} \right ] := \U(\v_1) \b$ \\
\hspace*{0.5cm} $\widetilde \x_1 := \widetilde\mA_r^{-1} \widetilde\b_1$ \\
\hspace*{0.5cm} if $\widetilde\mA_{21}  \widetilde \x_1 = \widetilde{\b}_2$ then \\
\hspace*{1cm} $\widetilde \x := \left [\begin{smallmatrix} \widetilde \x_1+\widetilde\mA_r^{-1}\widetilde\mA_{12}\,\e_{n-r,1} \\ 
-\e_{n-r,1}\end{smallmatrix} \right ]$ \\
\hspace*{1cm} $\x := \U(\v_2)^t \,\widetilde\x$\\
\hspace*{1cm} return $\x$\\
\hspace*{0.5cm} else\\
\hspace*{1cm} return ``no solution exists'' \\
\end{minipage}
}
\caption{Algorithm {\sf solve}.}\label{fig:algo_solve}
\end{center}
\end{figure}

\begin{theorem} \label{thm:algo_solve}
Algorithm {\sf solve} works correctly in time $O\left(\sMMa''\left(\frac{p}{\alpha},\alpha\right)\right)$ with $p = \max(m,n)$.
It makes $m+n-2$ random choices in $\F$
and fails with probability less than~$1/2$ if $|S| \ge 2q(q+1)$ with $q = \min(m,n)$.
\end{theorem}
\begin{proof}
Recalling that $\mA \x = \b$ is equivalent to $\widetilde \mA \, \widetilde\x = \widetilde\b$
with $\widetilde\mA = \U(\v_1)\mA\U(\v_2)^t$, $\widetilde \x = \U(\v_2)^{-t}\x$ and $\widetilde\b = \U(\v_1)\b$,
we see that the algorithm begins by generating the matrix $\widetilde\mA$.
If $r \in\{0,1,\ldots,m\}$, then $\widetilde\mA$ has generic rank profile and thus $r = \rank(\widetilde \mA) = \rank(\mA)$.
This implies $\widetilde \mA_{22} - \widetilde\mA_{21}\widetilde\mA_r^{-1}\widetilde\mA_{12} = 0$
and, partitioning $\widetilde\b$ conformally with $\widetilde\mA$, 
we deduce that the linear system $\widetilde\mA \, \widetilde\x = \widetilde\b$ is equivalent to 
$$
\begin{bmatrix} \widetilde\mA_{r} & \widetilde\mA_{12} \\[1mm] 0 & 0 \end{bmatrix}
\widetilde\x 
= \begin{bmatrix} \I_r & 0 \\[1mm] -\widetilde\mA_{21} \widetilde\mA_r^{-1}  & \I_{m-r} \end{bmatrix}
\!\!
\begin{bmatrix} \widetilde \b_1 \\[1mm] \widetilde \b_2 \end{bmatrix}\!.
$$
If $\widetilde\b_2$ is such that $-\widetilde\mA_{21} \widetilde\x_1 + \widetilde\b_2 \ne 0$
for $\widetilde \x_1 = \widetilde\mA_r^{-1} \widetilde\b_1$,
then no solution exists.
Else, it is easily checked that any vector of the form
$$\widetilde\x(\v) = \begin{bmatrix} \widetilde\x_1 + \widetilde\mA_r^{-1} \widetilde\mA_{12}\v \\[1mm] -\v\end{bmatrix}
\quad \mbox{with} \quad \v \in \F^{n-r}$$
is a solution to $\widetilde\mA \, \widetilde\x = \widetilde \b$;
furthermore, taking $\v = e_{n-r,1} \ne 0$ when $n-r > 0$ ensures that this solution is nontrivial.
Pre-multiplying this solution by $\U(\v_2)^t$ then yields a nontrivial solution to the original system,
so correctness follows.

By Lemmas~\ref{lem:preconditioning} and~\ref{lem:lp_inv}, 
calling {\sf precond} and \lpinv\ uses $O(\alpha \, \M(p) + \sMMa''(p/\alpha,\alpha))$ operations in $\F$.
Then, since $\U(\v_1)$ and $\U(\v_2)$ are Toeplitz matrices,
we can deduce $\widetilde\b = \U(\v_1) \b$ and $\x= \U(\v_2)^t \,\widetilde\x$ 
from $\v_1$, $\v_2$, $\b$, $\widetilde\x$ in time $O(\M(p))$.
Finally, given the generators $(\widetilde\mG,\widetilde\mH,\widetilde\v)$ and
$(\widetilde\mY,\widetilde\mZ,\widetilde\v)$ of $\widetilde\mA$ and $\widetilde\mA_r^{-1}$,
we can generate the submatrices $\widetilde\mA_{12}$ and $\widetilde\mA_{21}$
and perform the matrix-vector products 
$\widetilde\mA_r^{-1}\widetilde\b_1$, $\widetilde\mA_{21} \widetilde \x_1$, 
and $\widetilde\mA_r^{-1}\widetilde\mA_{12}\,\e_{n-r,1}$ in time $O(\alpha\,\M(p))$.
Recalling that $\alpha \M(p) \le \sMMa''(p/\alpha,\alpha)$, 
we conclude that the total cost is thus in $O(\sMMa''(p/\alpha,\alpha))$.

The probability analysis is the same as for inversion.
\end{proof}

{\it Remark.}
If one wants 
to sample uniformly the solution manifold of $\mA \, \x = \b$,
then it suffices to replace the unit vector $\e_{n-r,1}$ in algorithm {\sf solve} 
by a vector $\r_3\in\F^{n-r}$ whose entries are selected uniformly at random from the subset $S$.
This way of constructing $\x$ corresponds to the technique introduced 
by Kaltofen and Saunders in~\cite[Theorem~4]{KaSa91}, but requires
only $n-r$ additional random choices instead of $n$.

\bibliographystyle{siamplain}

\begin{thebibliography}{10}

\bibitem{AhStUl75}
{\sc A.~V. Aho, K.~Steiglitz, and J.~D. Ullman}, {\em Evaluating polynomials at
  fixed sets of points}, SIAM J. Comput., 4 (1975), pp.~533--539.

\bibitem{BiPa94}
{\sc D.~Bini and V.~Y. Pan}, {\em Polynomial and Matrix Computations, volume 1:
  Fundamental Algorithms}, Birkh\"auser, 1994.

\bibitem{BA80}
{\sc R.~R. Bitmead and B.~D.~O. Anderson}, {\em Asymptotically fast solution of
  {T}oeplitz and related systems of linear equations}, Linear Algebra Appl., 34
  (1980), pp.~103--116.

\bibitem{BoJeSc08}
{\sc A.~Bostan, C.-P. Jeannerod, and {\'E}.~Schost}, {\em Solving structured
  linear systems with large displacement rank}, Theoret. Comput. Sci., 407
  (2008), pp.~155--181.

\bibitem{BoLeSaScWi04}
{\sc A.~Bostan, G.~Lecerf, B.~Salvy, {\'E}.~Schost, and B.~Wiebelt}, {\em
  Complexity issues in bivariate polynomial factorization}, Proceedings of
  ISSAC'04, ACM, 2004, pp.~42--49.

\bibitem{BoLeSc03}
{\sc A.~Bostan, G.~Lecerf, and {\'E}.~Schost}, {\em Tellegen's principle into
  practice}, Proceedings of ISSAC'03, ACM, 2003, pp.~37--44.

\bibitem{BoSc05}
{\sc A.~Bostan and {\'E}.~Schost}, {\em Polynomial evaluation and interpolation
  on special sets of points}, J. Complexity, 21 (2005), pp.~420--446.

\bibitem{BuClSh97}
{\sc P.~B{\"u}rgisser, M.~Clausen, and A.~Shokrollahi}, {\em Algebraic
  Complexity Theory}, Springer, 1997.

\bibitem{CaKa91}
{\sc D.~G. Cantor and E.~Kaltofen}, {\em On fast multiplication of polynomials
  over arbitrary algebras}, Acta Inform., 28 (1991), pp.~693--701.

\bibitem{ChJeNeScVi15}
{\sc M.~F.~I. Chowdhury, C.-P. Jeannerod, V.~Neiger, {\'E}.~Schost, and
  G.~Villard}, {\em Faster algorithms for multivariate interpolation with
  multiplicities and simultaneous polynomial approximations}, IEEE Trans.
  Inform. Theory, 61 (2015), pp.~2370--2387.

\bibitem{FrMoKaLj79}
{\sc B.~Friedlander, M.~Morf, T.~Kailath, and L.~Ljung}, {\em New inversion
  formulas for matrices classified in terms of their distance from {T}oeplitz
  matrices}, Linear Algebra Appl., 27 (1979), pp.~31--60.

\bibitem{GaGe13}
{\sc J.~\gathen{von zur} Gathen and J.~Gerhard}, {\em Modern computer algebra},
  Cambridge University Press, third~ed., 2013.

\bibitem{GoOl94}
{\sc I.~Gohberg and V.~Olshevsky}, {\em Complexity of multiplication with
  vectors for structured matrices}, Linear Algebra Appl., 202 (1994),
  pp.~163--192.

\bibitem{GoOl94b}
{\sc I.~Gohberg and V.~Olshevsky}, {\em Fast algorithms with preprocessing for
  matrix-vector multiplication problems}, J. Complexity, 10 (1994),
  pp.~411--427.

\bibitem{HaQuZi02}
{\sc G.~Hanrot, M.~Quercia, and P.~Zimmermann}, {\em The middle product
  algorithm. {I}}, Appl. Algebra Engrg. Comm. Comput., 14 (2004), pp.~415--438.

\bibitem{JeaMou10}
{\sc C.-P. Jeannerod and C.~Mouilleron}, {\em Computing specified generators of
  structured matrix inverses}, Proceedings of ISSAC'10, ACM, 2010,
  pp.~281--288.

\bibitem{KaKuMo79b}
{\sc T.~Kailath, S.~Y. Kung, and M.~Morf}, {\em Displacement ranks of a
  matrix}, Bull. Amer. Math. Soc. (N.S.), 1 (1979), pp.~769--773.

\bibitem{KaKuMo79}
{\sc T.~Kailath, S.~Y. Kung, and M.~Morf}, {\em Displacement ranks of matrices
  and linear equations}, J. Math. Anal. Appl., 68 (1979), pp.~395--407.

\bibitem{K94}
{\sc E.~Kaltofen}, {\em Asymptotically fast solution of {T}oeplitz-like
  singular linear systems}, Proceedings of ISSAC'94, ACM, 1994, pp.~297--304.

\bibitem{Kaltofen95}
{\sc E.~Kaltofen}, {\em Analysis of {C}oppersmith's block {W}iedemann algorithm
  for the parallel solution of sparse linear systems}, Math. Comp., 64 (1995),
  pp.~777--806.

\bibitem{KaSa91}
{\sc E.~Kaltofen and D.~Saunders}, {\em On {W}iedemann's method of solving
  sparse linear systems}, in AAECC-9, vol.~539 of Lecture Notes in Computer
  Science, Springer, 1991, pp.~29--38.

\bibitem{LaTi85}
{\sc P.~Lancaster and M.~Tismenetsky}, {\em The Theory of Matrices, Second
  Edition}, Computer Science and Applied Mathematics, Academic Press, 1985.

\bibitem{LeGall14}
{\sc F.~Le~Gall}, {\em Powers of tensors and fast matrix multiplication}, in
  Proceedings of ISSAC'14, ACM, 2014, pp.~296--303.

\bibitem{Morf74}
{\sc M.~Morf}, {\em Fast Algorithms for Multivariable Systems}, PhD thesis,
  Dept. of Electrical Engineering, Stanford University, Stanford, 1974.

\bibitem{M80}
{\sc M.~Morf}, {\em Doubling algorithms for {T}oeplitz and related equations},
  in Proceedings of the IEEE International Conference on Acoustics, Speech and
  Signal Processing, 1980, pp.~954--959.

\bibitem{OlsPan98}
{\sc V.~Olshevsky and V.~Pan}, {\em A unified superfast algorithm for boundary
  rational tangential interpolation problems and for inversion and
  factorization of dense structured matrices}, in Proc. 39th IEEE FOCS, 1998,
  pp.~192--201.

\bibitem{OlSh99}
{\sc V.~Olshevsky and A.~Shokrollahi}, {\em A unified superfast algorithm for
  confluent tangential interpolation problem and for structured matrices}, in
  Advanced Signal Processing Algorithms, Architectures, and Implementations,
  ASPAAI'IX, SPIE, 1999, pp.~312--323.

\bibitem{OlSh00}
{\sc V.~Olshevsky and A.~Shokrollahi}, {\em Matrix-vector product for confluent
  {C}auchy-like matrices with application to confluent rational interpolation},
  in STOC'00, ACM, 2000, pp.~573--581.

\bibitem{Pan82}
{\sc V.~Y. Pan}, {\em Trilinear aggregating with implicit canceling for a new
  acceleration of matrix multiplication}, Comp. \& Maths. with Appls., 8
  (1982), pp.~23--34.

\bibitem{Pan90}
{\sc V.~Y. Pan}, {\em On computations with dense structured matrices}, Math.
  Comp., 55 (1990), pp.~179--190.

\bibitem{Pan99}
{\sc V.~Y. Pan}, {\em A unified superfast divide-and-conquer algorithm for
  structured matrices}.
\newblock MSRI Preprint 1999-033, Mathematical Sciences Research Institute,
  Berkeley, CA, April 1999.

\bibitem{Pan00}
{\sc V.~Y. Pan}, {\em Nearly optimal computations with structured matrices}, in
  SODA'00, ACM, 2000, pp.~953--962.

\bibitem{Pan01}
{\sc V.~Y. Pan}, {\em Structured Matrices and Polynomials}, Birkh\"auser Boston
  Inc., 2001.

\bibitem{Pan15}
{\sc V.~Y. Pan}, {\em Transformations of matrix structures work again}, Linear
  Algebra Appl., 465 (2015), pp.~107--138.

\bibitem{PanWan03}
{\sc V.~Y. Pan and X.~Wang}, {\em Inversion of displacement operators}, SIAM J.
  Matrix Anal. Appl., 24 (2003), pp.~660--677.

\bibitem{PZ00}
{\sc V.~Y. Pan and A.~Zheng}, {\em {S}uperfast algorithms for {C}auchy-like
  matrix computations and extensions.}, Linear Algebra Appl., 310 (2000),
  pp.~83--108.

\bibitem{Schonhage77}
{\sc A.~Sch{\"o}nhage}, {\em {Schnelle Multiplikation von Polynomen {\"u}ber
  K{\"o}rpern der Charakteristik 2}}, Acta Inform., 7 (1977), pp.~395--398.

\bibitem{ScSt71}
{\sc A.~Sch{\"o}nhage and V.~Strassen}, {\em Schnelle {M}ultiplikation
  gro{\ss}er {Z}ahlen}, Computing, 7 (1971), pp.~281--292.

\bibitem{Sergeev10}
{\sc I.~S. Sergeev}, {\em Fast algorithms for elementary operations with
  complex power series}, Diskret. Mat., 22 (2010), pp.~17--49.

\bibitem{Strassen69}
{\sc V.~Strassen}, {\em Gaussian elimination is not optimal}, Numer. Math., 13
  (1969), pp.~354--356.

\end{thebibliography}

\appendix

\newpage
\section{Proof that {$\mcL(\mA)$} and $\mcL'(\mA^{-1})$ have the same rank for $\mcL$ and $\mcL'$ as in~(\ref{eq:mcLP})}
\label{app:proof-inverse-is-structured}
Consider first the Sylvester case, where 
$\mcL(\mA) = \mM\mA - \mA\mN$ and $\mcL'(\mA^{-1}) = \mN\mA^{-1} - \mA^{-1}\mM$.
If $\mcL(\mA) = \mG \mH^t$, then pre- and post-multiplying both sides of this equality by $\mA^{-1}$ 
gives $\mA^{-1}\mM - \mN\mA^{-1} = \mA^{-1}\mG \mH^t\mA^{-1}$,
so that $\mcL'(\mA^{-1})$ equals $-\mA^{-1}\mG (\mA^{-t} \mH)^t$ and thus has the same rank as $\mcL(\mA)$.

Consider now the Stein case, where $\mcL(\mA) = \mA - \mM\mA\mN$ and $\mcL'(\mA^{-1}) = \mA^{-1} - \mN\mA^{-1}\mM$.
Defining the matrix 
$$\mB = \begin{bmatrix}
 \mA & \mM \\ \mN & \mA^{-1}
\end{bmatrix}\!,$$
we have
$$
\begin{bmatrix}\I & -\mM\mA\\ & \I \end{bmatrix}\mB \begin{bmatrix}\I & \\ -\mA\mN& \I \end{bmatrix}
= \diag\big(\mA-\mM\mA\mN, \mA^{-1}\big)
$$
and
$$
\begin{bmatrix}\I &\\ -\mN\mA^{-1} & \I \end{bmatrix}\mB \begin{bmatrix}\I & -\mA^{-1}\mM\\ & \I \end{bmatrix}
= \diag\big(\mA, \mA^{-1}-\mN\mA^{-1}\mM\big).
$$
Since the four matrices applied to $\mB$ are invertible,
we deduce that
$$\rank(\mB) = \rank(\mA-\mM\mA\mN) + \rank(\mA^{-1}) = \rank(\mA) + \rank(\mA^{-1}-\mN\mA^{-1}\mM).$$
Using $\rank(\mA) = \rank(\mA^{-1})$, we see again that $\mcL(\mA)$ and $\mcL'(\mA^{-1})$ have the same rank.

\section{Proof of Lemma~\ref{lem:inverse-for-stein}}
\label{app:proof-lem:inverse-for-stein}
Fix $i$ and $j$. For all $\ell \ge 1$,~\cite[Theorem~4.7]{PanWan03} gives
$$ \mA'_{i,j} - \CC_{P_i}^\ell\, \mA'_{i,j}\, (\CC_{Q_j}^t)^\ell = \sum_{k \le
  \alpha}\K(\CC_{P_i}, \g_{i,k},\ell)\, \K(\CC_{Q_j}, \h_{j,k},\ell)^t$$
with $\K(\CC_{P_i},\g_{i,k},\ell)$ and $\K(\CC_{Q_j},
\h_{j,k},\ell)$ as in the proof of Lemma~\ref{lemma:Aij}.  Writing
$Q_j=q_{j,0} + \cdots +q_{j,n_j} x^{n_j}$, and multiplying the former
equality on the left by $q_{j,\ell}\, \CC_{P_i}^{n_j-\ell}$, we get, for $1
\le \ell \le n_j$,
\begin{alignat*}{3}
\hspace{-1cm}	\phantom{\sum_{k \le \alpha}} q_{j,\ell} \,\CC_{P_i}^{n_j-\ell}\, \mA'_{i,j} - q_{j,\ell}\, \CC_{P_i}^{n_j} \,  \mA'_{i,j}\, (\CC_{Q_j}^t)^\ell   \\ =
	 \sum_{k \le \alpha} \CC_{P_i}^{n_j-\ell}  \K(\CC_{P_i},\g_{i,k},\ell)& \, q_{j,\ell} \,\K(\CC_{Q_j}, \h_{j,k},\ell)^t \\
	  =  \sum_{k \le \alpha} & \K(\CC_{P_i}, \g_{i,k},n_j)\,\J_{n_j} \, q_{j,\ell}\, \J_{\ell,n_j}\, \K(\CC_{Q_j}, \h_{j,k},n_j)^t,
\end{alignat*}
since the last $\ell$ columns of $\K(\CC_{P_i}, \g_{i,k},n_j)$ are
precisely $\CC_{P_i}^{n_j-\ell}\, \K(\CC_{P_i},\g_{i,k},\ell)$; note that
the equality also holds for $\ell=0$. 
Summing over all
$\ell=0,\dots,n_j$, and using the fact that $Q_j(\CC_{Q_j}^t)=0$, we deduce
\begin{align*}
\widetilde{Q_j}(\CC_{P_i})\, \mA'_{i,j} 
&= \sum_{k \le \alpha}\K(\CC_{P_i},\g_{i,k},n_j)\, \J_{n_j}\, \Y_{Q_j}\, \K(\CC_{Q_j}, \h_{j,k},n_j)^t.
\end{align*}
The rest of the proof now follows exactly that of Lemma~\ref{lemma:Aij}.

\end{document}